\documentclass[11pt,a4paper]{article}
\usepackage[a4paper]{geometry}
\usepackage{amssymb,latexsym,amsmath,amsfonts,amsthm}
\usepackage{graphicx}
\usepackage{epstopdf}
\usepackage{float,amsmath,amssymb,mathrsfs,bm,multirow,graphics}
\usepackage{tikz}
\usepackage{pgflibraryshapes}
\usetikzlibrary{arrows,decorations.markings}
\usepackage{subfigure}
\usepackage{overpic}
\usepackage{fullpage}
\usepackage{comment,verbatim}
\usepackage{hyperref}
\usepackage{color}
\usepackage{mathrsfs}
\usepackage[title,titletoc]{appendix}
\usepackage{ulem}
\usepackage{enumitem}
\usepackage{mathabx}

\DeclareMathOperator{\diag}{diag}
\DeclareMathOperator*{\Res}{Res}

\DeclareMathOperator{\CHF}{CHF}
\DeclareMathOperator{\tr}{Tr}
\DeclareMathOperator{\Var}{Var}
\newcommand{\Bes}{\mathrm{Bes}}

\newcommand{\C}{\mathbb{C}}

\renewcommand{\Re}{\mathrm{Re}\,}
\renewcommand{\Im}{\mathrm{Im}\,}

\renewcommand{\vec}{\mathbf}
\newcommand{\ud}{\,\mathrm{d}}

\newcommand{\what}{\widehat}

\newcommand{\msf}{\mathsf}
\newcommand{\ii}{\mathrm{i}}

\newcommand{\Boh}{\mathcal{O}}

\newtheorem{theorem}{Theorem}[section]
\newtheorem{lemma}[theorem]{Lemma}
\newtheorem{proposition}[theorem]{Proposition}

\newtheorem{corollary}[theorem]{Corollary}

\newtheorem{rhp}[theorem]{RH problem}

\newcommand{\tac}{\mathrm{tac}}

\theoremstyle{definition}

\theoremstyle{remark}

\newtheorem{remark}[theorem]{Remark}

\numberwithin{equation}{section}

\begin{document}

\title{On the gap probability of the tacnode process}

\author{Luming Yao\footnotemark[1] ~~and ~Lun Zhang\footnotemark[1]}

\renewcommand{\thefootnote}{\fnsymbol{footnote}}
\footnotetext[1]{School of Mathematical Sciences, Fudan University, Shanghai 200433, China. E-mail: \{lumingyao, lunzhang\}\symbol{'100}fudan.edu.cn.}

\date{\today}

\maketitle

\begin{abstract}
The tacnode process is a universal determinantal point process arising from non-intersecting particle systems and tiling problems. It is the aim of this work to explore the integrable structure and large gap asymptotics for the gap probability of the thinned/unthinned tacnode process over $(-s,s)$. We establish an integral representation of the gap probability in terms of the Hamiltonian associated with a system of differential equations. With the aids of some remarkable differential identities for the Hamiltonian, we also compute large gap asymptotics, up to and including the constant term in the thinned case. As direct applications, we obtain expectation, variance and  a central limit theorem  for the associated counting function.
\end{abstract}

\tableofcontents

\section{Introduction}

Since the seminal work of Dyson on the Brownian motion model for eigenvalues of Gaussian unitary ensemble \cite{Dyson}, there has been significant interest in ensembles of non-intersecting paths. Among them, non-intersecting Brownian motion models and their variants have been most studied over the last few decades. Besides their intimate connections with a variety of physical, combinatorial and probabilistic models \cite{Fisher,For11,GOV,John05,John02,KT07,KT04,WFS}, the long-standing interest in non-intersecting Brownian motions is due in large part to the fact that the scaling limits lead to universal determinantal point processes related to random matrix theory and the KPZ universality class.

In a typical case, we consider $n$ 1D non-intersecting Brownian motion paths with several prescribed starting and ending points. For any fixed time, the positions of these paths form a determinantal point process. As the number of paths tends to infinity, these paths will, after proper scalings, fill out a region in the time-space plane with a deterministic limit shape. It comes out that the local statistics of this model is governed by sine process in the interior of the shape \cite{ABK05,BK07,DKV,Joh01,LSY19,LY17}, by the Airy process at the edge of the limit shape \cite{ABK05,BK07,DKV,Huang,LY17b}, and by the Pearcey process at the cusp \cite{AOV10,AV07,BK07,BH98a,BH98b,TW06}; see also \cite{ADV11,ADV10,CNV20,NV22} for relevant studies.

In this paper, we focus on a critical process arising from non-intersecting Brownian motions and random walk paths called \textit{tacnode process}.
This process appears in the case of critical separation, that is, two groups of Brownian motions are asymptotically  distributed in two ellipses in the time-space plane which are tangent to each other (critical separation) and create a tacnode point; see Figure \ref{fig:bm} for an illustration. The tacnode process then describes local correlations of the paths around this point. As a determinantal process, the tacnode process is characterized by a two-variable correlation function $K_{\tac}(x,y)$ called the tacnode kernel and the kernel also depends on some extra parameters relevant to the scalings. The tacnode process was studied by different groups of authors using different techniques. Adler, Ferrari and van Moerbeke \cite{AFV13} resolved the tacnode problem for non-intersecting random walks on $\mathbb{Z}$ (discrete space and continuous time). Johansson \cite{John13} gave an integral representation of the tacnode kernel in the continuous time-space setting. Ferrari and Vet\H{o} \cite{FV12} extended the results of Johansson to the non-symmetric case when the two touching groups of Brownian motions may have different sizes. In all these studies, the tacnode kernel is expressed using resolvents and Fredholm determinants of the Airy integral operator. An alternative expression of the tacnode kernel is given by Delvaux, Kuijlaars and the second author \cite{DKZ11} with the aid of a new $4\times 4$ matrix-valued Riemann-Hilbert (RH) problem. This RH problem has a remarkable connection with the Hastings-McLeod solution \cite{HM} of the homogeneous Painlev\'{e} II equation
\begin{align}\label{eq:PII}
q''(x)=2q(x)^3+xq(x).
\end{align}
A natural question is then to ask whether all these formulas for the tacnode kernel lead to the same process, although it is generally believed to be the case. The equivalence of the RH formulation of the tacnode kernel in \cite{DKZ11} and the Airy resolvent type formula of Johansson \cite{John13} was later established in \cite{Del}, while the two different Airy type formulas obtained in \cite{AFV13} and \cite{John13} was proved to be equivalent in \cite{AJV14} based on an indirect way.

\begin{figure}[t]
 \centering
  \includegraphics[scale=.55]{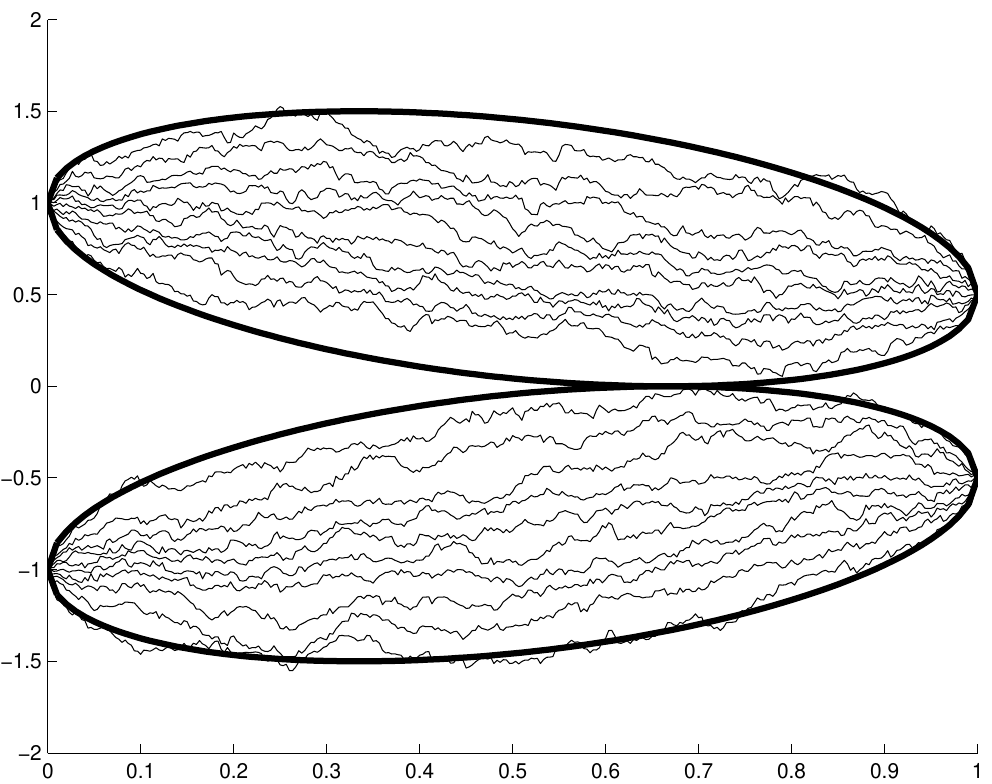}
  \caption{Simulation picture of 20 rescaled non-intersecting Brownian motion paths in the case of critical separation.}
  \label{fig:bm}
\end{figure}

Similar to the canonical Sine, Airy and Pearcey point processes, the tacnode process (or its variant) represents a universality class in a wide range of problems in probability and mathematical physics. Some concrete examples include non-intersecting Brownian motions on the unit circle \cite{BL19,LW16} and various random tilling models \cite{AJV22,AJV14,AV23}, among others.

We intend to investigate the gap probability of the tacnode process -- a basic object in the theory of point processes. More precisely, let $\mathcal{K}_{\tac}$ be the integral operator acting on $L^2\left(-s, s\right)$, $s\geq 0$, with the tacnode kernel $K_{\tac}$ and consider the associated Fredholm determinant $D(s;\gamma):=\det\left(I-\gamma  \mathcal K_{\tac}\right)$, where $0< \gamma \leq 1$ is a real parameter. The determinantal structure implies that $D(s; 1)$ can be interpreted as the probability of finding no particles (a.k.a. the gap probability) on the interval $(-s,s)$ for the tacnode process, while the deformed determinant $D(s; \gamma)$, $0<\gamma<1$, gives us gap probability for the thinned tacnode process. The thinned process is related to the original one by removing each particle independently with probability $1-\gamma$ (cf. \cite{IllianBook}), and according to a general result in \cite{CC21}, the information of $D(s; \gamma)$ is essential in establishing global rigidity result for the tacnode process.

Beyond the fundamental meaning of $D(s;\gamma)$ just described, our study is also highly motivated by rich structures of Fredholm determinants for canonical universality classes. For instance, the celebrated Tracy-Widom distribution established in \cite{TW94} shows that the Airy-kernel determinant admits an integral representation via the Hastings-McLeod solution of \eqref{eq:PII}, which also leads to a conjecture of the large gap asymptotic formula. This conjecture was rigorously proved in \cite{BRD08,DIK2008,DIKZ} using different approaches. The same integral formula holds for the deformed Airy-kernel determinant but in terms of the Ablowitz-Segur solution \cite{AS76,SA81} of \eqref{eq:PII} instead; cf. \cite{BCP,BB18}. Large gap asymptotics in this case, however, exhibits a significantly different behavior from the undeformed case, as conjectured in \cite{BCI} and lately proved in \cite{BB18,BIP}. Analogous results can be found in \cite{BW83,DIKZ,Dyson76,Eh06,JMMS,Krasovksy04,Widom94} for the sine-kernel determinant, and in \cite{BH98a,CM,DXZ22,DXZ21} for the Pearcey-kernel determinants.

Due to the highly transcendental form of the tacnode kernel, it remains an intriguing open problem to explore the integrable structure and large $s$ asymptotics of $D(s;\gamma)$; see however \cite{BC13,GZ,Gir14} for the transitions between the tacnode process and the Airy, Pearcey processes. It is the aim of this paper to resolve these problems and our results are stated in the next section.

\paragraph{Notations} Throughout this paper, the following notations are frequently used.
\begin{itemize}
  \item If $A$ is a matrix, then $(A)_{ij}$ stands for its $(i,j)$-th entry and $A^{\msf T}$ stands for its transpose. An unimportant entry of $A$ is denoted by $\ast$. We use  $I$ to denote an identity matrix, and the size might differ in different contexts. To emphasize a $k\times k$ identity matrix, we also use the notation $I_k$.
  \item  It is notationally convenient to denote by $E_{j,k}$ the $4\times 4$ elementary matrix
whose entries are all $0$, except for the $(j,k)$-entry, which is $1$, that is,
\begin{equation}\label{def:Eij}
E_{j,k}=\left( \delta_{l,j}\delta_{k,m} \right)_{l,m=1}^4,
\end{equation}
where $\delta_{j,k}$ is the Kronecker delta.

\item We denote by $D(z_0, r)$ the open disc centred at $z_0$ with radius $r > 0$, i.e.,
  \begin{equation}\label{def:dz0r}
   D(z_0, r) := \{ z\in \mathbb{C} \mid |z-z_0|<r \},
   \end{equation}
  and by $\partial D(z_0, r)$ its boundary. The orientation of $\partial D(z_0, r)$ is taken in a clockwise manner.
  \item As usual, the three Pauli matrices $\{\sigma_j\}_{j=1}^3$ are defined by
\begin{equation}\label{def:Pauli}
\sigma_1=\begin{pmatrix}
           0 & 1 \\
           1 & 0
        \end{pmatrix},
        \qquad
        \sigma_2=\begin{pmatrix}
        0 & -\ii \\
        \ii & 0
        \end{pmatrix},
        \qquad
        \sigma_3=
        \begin{pmatrix}
        1 & 0 \\
         0 & -1
         \end{pmatrix}.
\end{equation}

\item From time to time, we will encounter some functions that depend on the real parameters $r_1, r_2, s_1, s_2$ and $\tau$. If $x(\cdot; r_1, r_2, s_1, s_2, \tau)$ is such a function, we set
    \begin{align}
\widetilde x(\cdot; r_1, r_2, s_1, s_2, \tau) & = x(\cdot; r_2, r_1, s_2, s_1, \tau), \label{def:tildeX}
\\
\dot x(\cdot; r_1, r_2, s_1, s_2, \tau) &= x(\cdot; r_1, r_2, s_1, s_2, -\tau). \label{def:dotX}
\end{align}
Clearly, one has $x=\widetilde x$ if $r_1=r_2$ and $s_1=s_2$, and $x=\dot x$ if $\tau=0$.
\end{itemize}

\section{Main results}
\subsection{Definition of the tacnode kernel}
As mentioned previously, there exist several equivalent formulas of the tacnode kernel. We use the one that is defined through the following $4 \times 4$ tacnode RH problem \cite{DKZ11,DG13}.
\begin{rhp}\label{rhp:tac}
\hfill
\begin{itemize}
\item[\rm (a)]  $M(z)=M(z; r_1, r_2, s_1, s_2,\tau)$ is analytic for $ z \in \C \setminus \Gamma_M $, where the parameters $r_1,r_2,s_1,s_2,\tau$ are real with $r_i>0$, $i=1,2$, and
    \begin{equation}
    \Gamma_M:=\cup_{k=0}^5\Gamma_k \cup \{0\}
    \end{equation}
    with
    \begin{equation}\label{phi}
    \begin{aligned}
     &\Gamma_0= (0,+\infty),&& \Gamma_1=e^{\varphi \ii}(0,+\infty), &&\Gamma_2=e^{-\varphi \ii}(-\infty,0),\\
     &\Gamma_3= (-\infty,0),&& \Gamma_4=e^{\varphi \ii}(-\infty,0), &&\Gamma_5=e^{-\varphi \ii}(0,+\infty), \quad 0<\varphi<\frac{\pi}{3};
    \end{aligned}
    \end{equation}
    see Figure \ref{fig:tacnode} for an illustration of the contour $\Gamma_M$.
\item[\rm (b)]  For $z\in\Gamma_k$, $k=0,1,\ldots,5$, the limiting values
\[ M_+(z) = \lim_{\substack{\zeta \to z \\\zeta\textrm{ on $+$-side of }\Gamma_k}}M(\zeta), \qquad
   M_-(z) = \lim_{\substack{\zeta \to z \\\zeta\textrm{ on $-$-side of }\Gamma_k}}M(\zeta), \]
exist, where the $+$-side and $-$-side of $\Gamma_k$ are the sides
which lie on the left and right of $\Gamma_k$, respectively, when
traversing $\Gamma_k$ according to its orientation. These limiting
values satisfy the jump relation
\begin{equation}\label{jumps:M}
M_{+}(z) = M_{-}(z)J_k(z),\qquad k=0,\ldots,5,
\end{equation}
where the jump matrix $J_k(z)$ for each ray $\Gamma_k$ is shown in Figure \ref{fig:tacnode}.
\item[\rm (c)] As $z \to \infty$ with $z \in \C \setminus \Gamma_M$, we have
\begin{align}\label{eq:asy:M}
M(z)&=\left( I+\frac{M^{(1)}}{z}+ \Boh(z^{-2}) \right) \diag \left((-z)^{-\frac14},z^{-\frac14},(-z)^{\frac14},z^{\frac14} \right)
\nonumber  \\
& \quad \times A \diag \left(e^{-\theta_1(z)+\tau z}, e^{-\theta_2(z)- \tau z}, e^{\theta_1(z)+\tau z},e^{\theta_2(z)- \tau z} \right),
\end{align}
where the matrix $M^{(1)}$ is independent of $z$ but depends on the parameters,
\begin{align} \label{def:A}
A&:=\frac{1}{\sqrt 2} \begin{pmatrix} 1 & 0 & -\ii & 0 \\ 0 & 1& 0& \ii \\
-\ii & 0& 1& 0
\\
0 & \ii & 0 & 1 \end{pmatrix},
\\
\theta_1(z)&= \frac23 r_1(-z)^{\frac 32} +2 s_1 (-z)^{\frac 12}, \qquad z\in \mathbb{C}\setminus [0,\infty), \label{def:theta1}
\\
\theta_2(z)&=\frac23r_2z^{\frac 32} +2 s_2 z^{\frac 12}, \qquad z\in \mathbb{C}\setminus (-\infty,0]. \label{def:theta2}
\end{align}
\item[\rm (d)] $M(z)$ is bounded near $z=0$.
\end{itemize}
\end{rhp}

\begin{figure}[t]
\begin{center}
   \setlength{\unitlength}{1truemm}
   \begin{picture}(100,70)(-5,2)
       \put(40,40){\line(-1,-1){20}}
       \put(40,40){\line(-1,1){20}}
       \put(40,40){\line(-1,0){30}}
       \put(40,40){\line(1,0){30}}
       \put(40,40){\line(1,1){20}}
       \put(40,40){\line(1,-1){20}}

       \put(30,50){\thicklines\vector(1,-1){1}}
       \put(30,40){\thicklines\vector(1,0){1}}
       \put(30,30){\thicklines\vector(1,1){1}}
       \put(50,50){\thicklines\vector(1,1){1}}
       \put(50,40){\thicklines\vector(1,0){1}}
       \put(50,30){\thicklines\vector(1,-1){1}}

       \put(39,36.3){$0$}

       \put(27,22){$\Gamma_4$}
       \put(-10,18) {$\begin{pmatrix} 1&1&0&0\\0&1&0&0\\0&0&1&0\\0&1&-1&1 \end{pmatrix}$}

       \put(27,55){$\Gamma_2$}
       \put(-10,66){$\begin{pmatrix} 1&-1&0&0\\0&1&0&0\\0&0&1&0\\0&1&1&1 \end{pmatrix}$}

       \put(17,42){$\Gamma_3$}
       \put(-21,40){$\begin{pmatrix}1&0&0&0\\0&0&0&1\\0&0&1&0\\0&-1&0&0 \end{pmatrix}$}

       \put(48,22){$\Gamma_5$}
       \put(60,18){$\begin{pmatrix} 1&0&0&0\\1&1&0&0\\1&0&1&-1\\0&0&0&1 \end{pmatrix}$}

       \put(48,55){$\Gamma_1$}
       \put(60,66) {$\begin{pmatrix} 1&0&0&0 \\ -1 &1&0&0\\1&0&1&1\\0&0&0&1 \end{pmatrix}$}

       \put(58,42){$\Gamma_0$}
       \put(72,40){ $\begin{pmatrix}
0&0&1&0\\0&1&0&0\\-1&0&0&0\\0&0&0&1 \end{pmatrix}$ }

       \put(40,40){\thicklines\circle*{1}}

 \end{picture}

   \caption{The jump contours $\Gamma_k$ and the corresponding jump matrices $J_{k}$, $k=0,\ldots,5$, in the RH problem for $M$.}
   \label{fig:tacnode}
\end{center}
\end{figure}
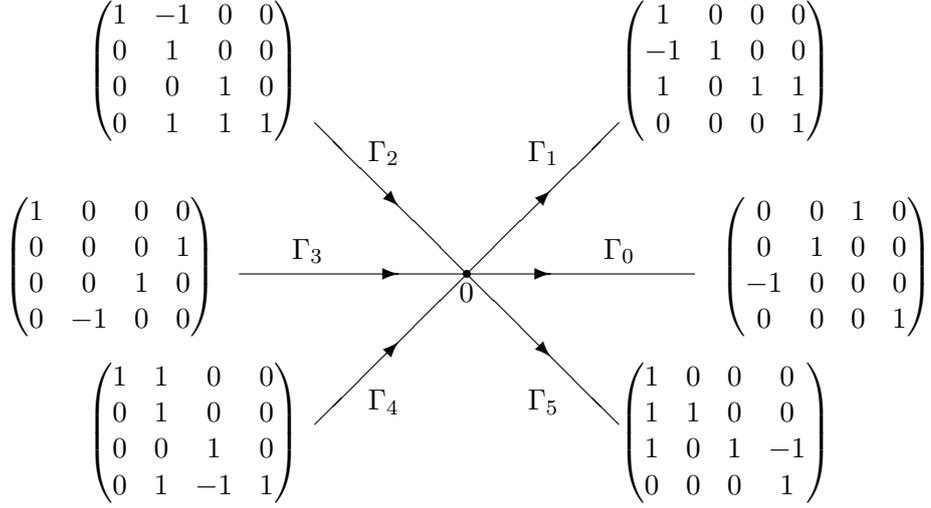

The jump contour of the original tacnode RH problem consists of ten rays emanating from the origin. Here, as in \cite{Kuij}, we reduce the number of rays to six by combining the two jumps in each of the open quadrants. The existence of a unique solution to the tacnode RH problem was proved for $\tau=0$ by Delvaux, Kuijlaars and the second author \cite{DKZ11}, for the symmetric case $r_1=r_2=1$, $s_1=s_2$ with general $\tau$ by Duits and Geudens \cite{DG13}, for the non-symmetric case by Delvaux \cite{Del}.

\begin{remark}
The RH problem for $X$ is related to the Hastings-McLeod solution of the Painlev\'{e} II equation \eqref{eq:PII} through the ``residue'' term $M^{(1)}$ in \eqref{eq:asy:M}. More precisely, the Hastings-McLeod solution and associated Hamiltonian appear in the top right $2 \times 2$ block of the matrix $M^{(1)}$; see \cite{Del,DKZ11,DG13}. Note that $M$ satisfies the following symmetric relations (see \cite{Del,DKZ11}):
\begin{align} \label{eq:symmM1}
  M (-z;r_1,r_2,s_1,s_2,\tau) &=
  \begin{pmatrix}
  J & 0
  \\
  0 & -J
  \end{pmatrix}\widetilde  M(z)
  \begin{pmatrix}
  J & 0
  \\
  0 & -J
  \end{pmatrix},
\\
  M (z;r_1,r_2,s_1,s_2,\tau)^{- \msf T} &=
  \begin{pmatrix}
  0 & -I_2
  \\
  I_2 & 0
  \end{pmatrix} \dot M(z)
  \begin{pmatrix}
  0 & I_2
  \\
  -I_2 & 0
  \end{pmatrix}, \label{eq:symmM2}
\end{align}
where
\begin{equation}\label{def-J}
J =  \begin{pmatrix}
  0 & 1
  \\
  1 & 0
  \end{pmatrix},
\end{equation}
$\widetilde M$ and $\dot M$ are defined through \eqref{def:tildeX} and \eqref{def:dotX}.  It is then readily seen that $M^{(1)}=M^{(1)}(r_1,r_2,s_1,s_2,\tau)$ satisfies the symmetric relations
\begin{align}
 M^{(1)} &=
  -\begin{pmatrix}
  J & 0
  \\
  0 & -J
  \end{pmatrix}\widetilde M^{(1)}
  \begin{pmatrix}
  J & 0
  \\
  0 & -J
  \end{pmatrix},\label{eq:symmM11}
\\
 ( M^{(1)})^{\msf T} &=
  -\begin{pmatrix}
  0 & -I_2
  \\
  I_2 & 0
  \end{pmatrix}\dot M^{(1)}
  \begin{pmatrix}
  0 & I_2
  \\
  -I_2 & 0
  \end{pmatrix}. \label{eq:symmM12}
\end{align}
As a consequence, we have
\begin{align}
M^{(1)}_{11}&= -\widetilde M^{(1)}_{22}=-\dot M^{(1)}_{33}=\dot{\widetilde{M}}^{(1)}_{44},\label{M11}\\
M^{(1)}_{13}&= \widetilde M^{(1)}_{24}=\dot M^{(1)}_{13},\label{M13}\\
M^{(1)}_{23}&= \widetilde M^{(1)}_{14}=\dot M^{(1)}_{14}.\label{M14}
\end{align}
\end{remark}

Let $\widehat M $ be the analytic continuation of the restriction of $M$ in the sector bounded by the rays $\Gamma_1$ and $\Gamma_2$ to the whole complex plane. The tacnode kernel $K_{\tac}(x,y):=K_{\tac}(x,y;r_1,r_2,s_1,s_2,\tau)$ is then given in terms of $\widehat M$ by \cite[Definition 2.6]{DKZ11} \footnote{There is a misprint in \cite[Definition 2.6]{DKZ11}. It should be `$\widehat M(v)^{-1}\widehat M(u)$' instead of `$\widehat M(u)^{-1}\widehat M(v)$'.}

\begin{equation} \label{def:tacnode kernel}
    K_{\tac}(x,y)
     = \frac{1}{2\pi \ii (x-y)} \begin{pmatrix} 0 & 0 & 1 & 1 \end{pmatrix}
    \widehat M(y)^{-1}
    \widehat M(x) \begin{pmatrix} 1 \\ 1 \\ 0 \\ 0
    \end{pmatrix}.
\end{equation}

Define
\begin{equation}\label{def:Fnotation}
F(s;\gamma)=F(s;\gamma, r_1, r_2,s_1,s_2, \tau):=\ln (D(s;\gamma))=\ln \det\left(I-\gamma  \mathcal K_{\tac}\right), \quad 0< \gamma \leq 1.
\end{equation}
Our first result is an integral representation of $F$ as stated in what follows.

\subsection{An integral representation of $F$}
The integral representation of $F$ involves the Hamiltonian of a system of coupled differential equations.
These differential equations are given as follows:
\begin{equation}\label{def:pq's}
\left\{
\begin{aligned}
p_1'(s)&=-\ii r_1 s p_3(s) -p_1(s)p_5(s)-\ii r_1 p_2(s)\widetilde q_6(s)-\ii r_1 p_3(s)q_5(s)+p_4(s)p_6(s)-\tau p_1(s)\\
&\quad+\ii s_1 p_3(s)-\frac{\widetilde p_2(s)}{s} (p_1(s) \widetilde q_2(s)+p_2(s) \widetilde q_1(s) - p_3(s) \widetilde q_4(s)-p_4(s) \widetilde q_3(s)),\\
p_2'(s)&=\ii r_2 s p_4(s) +\ii r_2 p_1(s) q_6(s)+p_2(s)\widetilde p_5(s) +p_3(s)\widetilde p_6(s)-\ii r_2 p_4(s)\widetilde q_5(s) +\tau p_2(s)\\
&\quad+\ii s_2 p_4(s)-\frac{\widetilde p_1(s)}{s} (p_1(s) \widetilde q_2(s)+p_2(s) \widetilde q_1(s) - p_3(s) \widetilde q_4(s)-p_4(s) \widetilde q_3(s)),\\
p_3'(s)&=-\ii r_1 p_1(s) +p_3(s)p_5(s)-\ii r_2p_4(s)\widetilde q_6(s)-\tau p_3(s) \\
&\quad+\frac{\widetilde p_4(s)}{s} (p_1(s) \widetilde q_2(s)+p_2(s) \widetilde q_1(s) - p_3(s) \widetilde q_4(s)-p_4(s) \widetilde q_3(s)),\\
p_4'(s)&= -\ii r_2 p_2(s) + \ii r_1 p_3(s)q_6(s) -p_4(s)\widetilde p_6(s)+\tau p_4(s)\\
&\quad+\frac{\widetilde p_3 (s)}{s} (p_1(s) \widetilde q_2(s)+p_2(s) \widetilde q_1(s) - p_3(s) \widetilde q_4(s)-p_4(s) \widetilde q_3(s)),\\
p_5'(s)&=-\ii r_1(p_3(s)q_1(s)+\widetilde p_4(s)\widetilde q_2(s)),\\
p_6'(s)&=\ii r_1 (p_3(s)q_4(s)-\widetilde p_4(s)\widetilde q_3(s))+\ii r_2 (p_1(s)q_2(s)-\widetilde p_2(s)\widetilde q_1(s)),\\
q_1'(s)&=p_5(s)q_1(s)-\ii r_2 q_2(s)q_6(s)+\ii r_1 q_3(s)+\tau q_1(s)\\
&\quad+\frac{\widetilde q_2(s)}{s}(\widetilde p_2(s)q_1(s)+\widetilde p_1(s)q_2(s) -\widetilde p_4(s) q_3(s) -\widetilde p_3(s)q_4(s)),\\
q_2'(s)&=\ii r_1 q_1(s)\widetilde{q}_6(s)-\widetilde{p}_5(s)q_2(s)+\ii r_2q_4(s)-\tau q_2(s)\\
&\quad+\frac{\widetilde q_1(s)}{s}(\widetilde p_2(s)q_1(s)+\widetilde p_1(s)q_2(s) -\widetilde p_4(s) q_3(s) -\widetilde p_3(s)q_4(s)),\\
q_3'(s)&= \ii r_1 sq_1(s)+\ii r_1 q_1(s)q_5(s) -\widetilde p_6(s) q_2(s)-p_5(s)q_3(s)-\ii r_1q_4(s)q_6(s)-\ii s_1 q_1(s)\\
&\quad+\tau q_3(s)-\frac{\widetilde q_4(s)}{s}(\widetilde p_2(s)q_1(s)+\widetilde p_1(s)q_2(s) -\widetilde p_4(s) q_3(s) -\widetilde p_3(s)q_4(s)),\\
q_4'(s)&=-\ii r_2sq_2(s)-p_6(s)q_1(s)+\ii r_2q_2(s)\widetilde q_5(s) + \ii r_2 q_3(s)\widetilde q_6(s)+\widetilde p_5(s)q_4(s)-\ii s_2q_2(s)\\
&\quad-\tau q_4(s)-\frac{\widetilde q_3(s)}{s}(\widetilde p_2(s)q_1(s)+\widetilde p_1(s)q_2(s) -\widetilde p_4(s) q_3(s) -\widetilde p_3(s)q_4(s)),\\
q_5'(s)&=p_1(s)q_1(s)-p_3(s)q_3(s)-\widetilde p_2(s) \widetilde q_2(s)+\widetilde p_4(s) \widetilde q_4(s),\\
q_6'(s)&=-p_4(s)q_1(s)-\widetilde p_3(s) \widetilde q_2(s),
\end{aligned}\right.
\end{equation}
where
$$
p_k(s)=p_k(s;\gamma, r_1, r_2,s_1,s_2, \tau), \qquad q_k(s)=q_k(s;\gamma, r_1, r_2,s_1,s_2, \tau), \qquad k=1,\ldots,6,
$$
are 12 unknown functions, $\widetilde p_k(s)$ and $\widetilde q_k(s)$ are related to $p_k(s)$ and  $q_k(s)$ by swapping the parameters $r_1\leftrightarrow r_2$ and $s_1\leftrightarrow s_2$; see the definition \eqref{def:tildeX}. By introducing the matrix-valued functions
\begin{equation}\label{def:A0}
A_0(s) = \begin{pmatrix}
p_5(s)+\tau & -\ii r_2 q_6(s) & \ii r_1 & 0\\
\ii r_1 \widetilde q_6(s) & -\widetilde p_5(s)-\tau & 0 & \ii r_2\\
\ii r_1 q_5(s)-\ii s_1 & -\widetilde p_6(s) & -p_5(s)+\tau & -\ii r_1 q_6(s)\\
-p_6(s) & \ii r_2 \widetilde q_5(s)-\ii s_2 & \ii r_2 \widetilde q_6(s) & \widetilde p_5(s)-\tau
\end{pmatrix},
\end{equation}
\begin{equation}\label{def:A1}
A_1(s) = \begin{pmatrix}
q_1(s)\\q_2(s)\\q_3(s)\\q_4(s)
\end{pmatrix}\begin{pmatrix}
p_1(s) & p_2(s)&p_3(s)&p_4(s)
\end{pmatrix},
\end{equation}
and
\begin{equation}\label{def:A2}
A_2(s) = \begin{pmatrix}
\widetilde q_2(s)\\\widetilde q_1(s)\\-\widetilde q_4(s)\\-\widetilde q_3(s)
\end{pmatrix}\begin{pmatrix}
\widetilde p_2(s) & \widetilde p_1(s)&-\widetilde p_4(s)&-\widetilde p_3(s)
\end{pmatrix},
\end{equation}
one can check
\begin{align}\label{def:H}
& H(s)=H(s;\gamma, r_1, r_2,s_1,s_2, \tau)
\nonumber
\\
&:= \begin{pmatrix}
p_1(s) & p_2(s) & p_3(s) & p_4(s)
\end{pmatrix}\left(\ii (r_1E_{3,1}-r_2 E_{4,2})s + A_0(s)+\frac{A_2(s)}{2s}\right)\begin{pmatrix}
q_1(s)\\q_2(s)\\q_3(s)\\q_4(s)
\end{pmatrix}\nonumber\\
&\quad+\begin{pmatrix}
\widetilde p_2(s) & \widetilde p_1(s) & -\widetilde p_4(s) & -\widetilde p_3(s)
\end{pmatrix} \left(\ii (r_1E_{3,1}-r_2 E_{4,2})s - A_0(s)+\frac{A_1(s)}{2s}\right) \begin{pmatrix}
\widetilde q_2(s)\\ \widetilde q_1(s)\\ -\widetilde q_4(s)\\ -\widetilde q_3(s)
\end{pmatrix},
\end{align}
with $E_{i,j}$, $i,j =1,\ldots,4$, being the matrices defined in \eqref{def:Eij},
is the Hamiltonian for the above system of differential equations, under
the extra condition
\begin{equation}\label{eq:sumpq}
\sum_{k=1}^4 p_k(s)q_k(s)=0.
\end{equation}
That is, we have
\begin{equation}\label{pq}
q_k'(s)=\frac{\partial H}{\partial p_k}, \qquad p_k'(s)=-\frac{\partial H}{\partial q_k}, \qquad k=1,...,6.
\end{equation}
\begin{theorem}\label{th:F-H}
For $\gamma \in [0, 1]$, with the function $F(s;\gamma)$ defined in \eqref{def:Fnotation}, we have,
\begin{align}\label{eq:F-H2}
F(s;\gamma) = \int_0^s H(t; \gamma) \ud t, \qquad s \in (0, \infty),
\end{align}
where $H$ is the Hamiltonian \eqref{def:H} associated with a family of special solutions to the system of differential equations \eqref{def:pq's}. Moreover, $H(s)$ satisfies the following asymptotic behaviors: as $s \to 0^+$,
\begin{equation}\label{eq:H-0}
H(s)=\Boh(1),
\end{equation}
and as $s \to +\infty$,
\begin{align}\label{asy:H}
H(s)=\begin{cases}
-\frac{r_1^2 + r_2^2}{4} s^2 + (r_1s_1 + r_2 s_2) s - s_1^2-s_2^2 - \frac{1}{4s} +\Boh(s^{-2}), &\qquad \gamma=1,\\
2 \ii \beta(r_1 + r_2) s^{\frac 12} - 2 \ii \beta (s_1+s_2) s^{-\frac 12} \\
~- (3 \beta^2 + \frac{\ii \beta}{2} \cos{(2 \vartheta(s))}+\frac{\ii \beta}{2} \cos{(2 \widetilde{\vartheta}(s))})s^{-1}+ \Boh(s^{-\frac 32}),& \qquad 0 \le \gamma <1,
\end{cases}
\end{align}
where
\begin{equation}\label{def:beta}
\beta := \frac{1}{2 \pi \ii} \ln (1-\gamma), \qquad \gamma \in [0, 1),
\end{equation}
and
\begin{align}\label{def:theta}
\vartheta (s) &= \vartheta (s;r_1,s_1)= \frac{2 r_1}{3} s^{\frac 32}-2s_1 s + \frac{3 \ii \beta}{2} \ln s + \ii \beta \ln (8(r_1 - \frac{s_1}{s})) + \arg \Gamma (1+\beta),
\\
\widetilde{\vartheta}(s) &= \vartheta (s;r_2,s_2), \label{def:tildetheta}
\end{align}
with $\Gamma$ being the Euler's gamma function.
\end{theorem}
The local behavior of $H$ near the origin in \eqref{eq:H-0} ensures the integral \eqref{eq:F-H2} is well-defined. For $0<\gamma<1$, we also derive asymptotics of the family of special solutions in Proposition \ref{th:pq} below, which plays an important role in asymptotic studies of $F$.

\subsection{Large gap asymptotics and applications}
A direct application of Theorem \ref{th:F-H} is that we can obtain the first few terms in the asymptotic expansion of $F(s;\gamma)$ as $s \to +\infty$ except for the constant term by inserting \eqref{asy:H} into \eqref{eq:F-H2}. For $0<\gamma<1$, we are also able to determine the notoriously difficult constant term.

\begin{theorem}\label{th:1}
With the function $F(s;\gamma)$ defined in \eqref{def:Fnotation}, we have, as $s \to +\infty$,
\begin{equation}\label{ds-F}
F(s;\gamma) = \begin{cases}
-\frac{r_1^2 + r_2^2}{12} s^3 + \frac{r_1s_1 + r_2 s_2}{2} s^2 - (s_1^2+s_2^2) s - \frac{\ln{s}}{4} + C + \Boh(s^{-1}), & \gamma=1,\\
\frac{4 \beta \ii (r_1 + r_2)}{3} s^{\frac 32} - 4 \beta \ii (s_1 + s_2) s^{\frac 12} - 3 \beta^2 \ln{s} + 2 \ln{(G(1+\beta)G(1-\beta))}\\
~ -\beta^2 \ln{(64 r_1 r_2)}+ \Boh(s^{-\frac 12}), & 0 \le \gamma <1,
 \end{cases}
\end{equation}
uniformly for $r_i>0, i=1,2$ and $s_i \in \mathbb{R}, i=1,2$, where $C$ is an undetermined constant independent of $s$, $\beta$ is given in \eqref{def:beta} and $G(z)$ is the Barnes G-function.
\end{theorem}

If $\gamma =0$, we have that $\beta =0$ and $G(1+\beta)=G(1)=1$. It is then straightforward to see $F(s;0)= \Boh(s^{-1/2})$, which matches the fact that $F(s;0)=0$. If $\gamma=1$, our result supports the so-called Forrester-Chen-Eriksen-Tracy conjecture \cite{chen, F1993}. This conjecture claims that the probability $E(s)$ of emptiness over the interval $(x^*-s,x^*+s)$ behaves like $\exp(-Ks^{2\alpha+2})$ for large positive $s$ with $K$ being some constant, provided the density of state behaves as $|x-x^*|^{\alpha}$ as $x\to x^*$. Since the limiting mean density for the non-intersecting Brownian paths at the time of tangency consists of two touching semicircles, it follows that $\alpha=1/2$ for the tacnode process. Thus, one should have $F(s;1)=\Boh(s^3)$, as confirmed in \eqref{ds-F}. Also, we cannot evaluate explicitly the constant $C$ therein with our method, which in general is a challenging task; cf. \cite{Kra09}.


Our final result is about counting statistics of the tacnode process. To proceed, we denote by $N(s)$ the random variable that counts the number of points in the tacnode process falling into the interval $(-s,s)$, $s\geq 0$. It is well known that the following generating functionx
\begin{equation}
\mathbb{E} \left(e^{-2\pi \nu N(s)}\right) = \sum_{k=0}^{\infty} \mathbb{P} (N(s)=k)e^{-2\pi \nu k}, \qquad \nu \ge 0,
\end{equation}
is equal to the deformed Fredholm determinant $\det \left(I - (1-e^{-2\pi \nu})\mathcal K_{\tac}\right)$. This, together with Theorem \ref{th:1}, allows us to establish various asymptotic statistical properties of $N$; see also \cite{Charlier21,CC21,CM,DXZ22,Sosh} for relevant results about the sine, Airy and Pearcey point determinantal processes.

\begin{corollary}\label{coro1}
As $s \to +\infty$, we have
\begin{align}
\mathbb{E}(N(s)) &=\mu (s) + \Boh \left(\frac{\ln s}{s^{1/2}}\right),\label{def:EN}\\
\Var (N(s)) & = \sigma (s)^2 + \frac{2 + 2\gamma_E + \ln(64 r_1 r_2)}{2 \pi^2} + \Boh \left(\frac{(\ln s)^2}{s^{1/2}}\right),\label{def:VarN}
\end{align}
where $\gamma_E=-\Gamma'(1)\approx 0.57721$ is Euler’s constant,
\begin{align}\label{def:mu-sigma}
\mu (s) = \frac{2(r_1+r_2)}{3 \pi} s^{\frac 32} - \frac{2(s_1 + s_2)}{\pi} s^{\frac 12}, \qquad \sigma (s)^2 = \frac{3}{2 \pi^2} \ln s.
\end{align}
Furthermore, the random variable $\frac{N(s) - \mu (s)}{\sqrt{\sigma (s)^2}}$ converges in distribution to the normal law $\mathcal{N} (0,1)$as $s \to +\infty$, and for any $\epsilon > 0$, we have
\begin{align}\label{ub}
\lim_{a \to \infty} \mathbb{P} \left(\sup_{s>a} \left|\frac{N(s) - \mu(s)}{\ln s}\right| \le \frac{3 \sqrt{2}}{2 \pi} + \epsilon\right)=1.
\end{align}
\end{corollary}
The probabilistic bound \eqref{ub} particularly implies that, for large positive $s$, the
counting function of the tacnode process lies in the interval $(\mu(s)-(3\sqrt{2}/(2\pi)+\epsilon) \ln s, \mu(s)+
(3\sqrt{2}/(2\pi)+\epsilon) \ln s)$ with high probability.

\paragraph{Organization of the rest of the paper}
The rest of this paper is devoted to the proofs of our main results. The idea is to relate various derivatives of $F$ to a $4 \times 4$ RH problem under the general framework \cite{BD02, DIKZ}. In Section \ref{sec:RHP}, we connect $\ud F/\ud s$ to a $4 \times 4$ RH problem for $X$ with constant jumps. We then derive a lax pair for $X$ in Section \ref{sec:Lax}, and some useful differential identities for the Hamiltonian will also be included for later calculation. We perform a Deift-Zhou steepest descent analysis \cite{DIZ97} on the RH problem for $X$ as $s \to +\infty$ in Sections \ref{sec:AsyX1} and \ref{sec:AsyXgamma} for the cases $\gamma=1$ and $0 \leq \gamma < 1$, respectively, and deal with the small positive $s$ case in Section \ref{sec:AsyX0}. After computing the asymptotics of a family of special solutions to the system of differential equations \eqref{def:pq's} and \eqref{eq:sumpq} in Section \ref{sec:pq}, we finally present the proofs of our main results in Section \ref{sec:proof}.

\section{Preliminaries} \label{sec:RHP}
We intend to establish a relation between $\ud F / \ud s$ and an RH problem with constant jumps. To proceed, we note that
\begin{align}
\frac{\ud}{\ud s}F(s;\gamma)= -\textrm{tr}\left((I-\gamma \mathcal{K}_{\tac})^{-1} \gamma
\frac{\ud}{\ud s}\mathcal{K}_{\tac}\right) = -(R(s,s)+R(-s,-s)), \label{eq:derivatives}
\end{align}
where $R(u,v)$ stands for the kernel of the resolvent operator.

By \eqref{def:tacnode kernel}, one readily sees that
\begin{equation}\label{eq:tildeKdef}
\gamma K_\tac (x,y) = \frac{\vec{f}(x)^{\msf T}\vec{h}(y)}{x-y},
\end{equation}
where
\begin{equation}\label{def:fh}
\vec{f}(x)=\begin{pmatrix}
f_1
\\
f_2
\\
f_3
\\
f_4
\end{pmatrix}:= \widehat M (x)
\begin{pmatrix}
1
\\
1
\\
0
\\
0
\end{pmatrix}, \qquad
\vec{h}(y)=\begin{pmatrix}
h_1
\\
h_2
\\
h_3
\\
h_4
\end{pmatrix}
:=
\frac{\gamma}{2 \pi \ii}
\widehat M(y)^{-\msf T}
\begin{pmatrix}
0
\\
0
\\
1
\\
1
\end{pmatrix}.
\end{equation}
This integrable structure of kernel $K_\tac$  in the sense of Its et al. \cite{IIKS90} implies that the resolvent kernel $R(u,v)$ is integrable as well; cf. \cite{DIZ97,IIKS90}. Indeed, by setting
\begin{equation}\label{def:FH}
\vec{F}(u)=
\begin{pmatrix}
F_1 \\
F_2 \\
F_3 \\
F_4
\end{pmatrix}:=\left(I-\gamma \mathcal{K}_{\tac}\right)^{-1}\vec{f}, \qquad \vec{H}(v)=\begin{pmatrix}
H_1 \\
H_2 \\
H_3 \\
H_4
\end{pmatrix}
:=\left(I-\gamma \mathcal{K}_{\tac}\right)^{-1}\vec{h},
\end{equation}
we have
\begin{equation}\label{eq:resolventexpli}
R(u,v)=\frac{\vec{F}(u)^{\msf T}\vec{H}(v)}{u-v}.
\end{equation}
Moreover, the functions $\vec{F}(u)$ and $\vec{H}(u)$ are closely related to the following RH problem.
\begin{rhp} \label{rhp:Y}
\hfill
\begin{enumerate}
\item[\rm (a)] $Y(z)$ is a $ 4\times 4$ matrix-valued function defined and analytic in $\mathbb{C}\setminus [-s,s]$, where the orientation is taken from the left to the right.

\item[\rm (b)] For $x\in(-s,s)$, we have
\begin{equation}\label{eq:Y-jump}
 Y_+(x)=Y_-(x)(I-2\pi \ii \vec{f}(x)\vec{h}(x)^{\msf T}),
 \end{equation}
where the functions $\vec{f}$ and $\vec{h}$ are defined in \eqref{def:fh}.
\item[\rm (c)] As $z \to \infty$, we have
\begin{equation}\label{eq:Y-infty}
 Y(z)=I+\frac{Y^{(1)}}{z}+ \Boh(z^{-2}).
 \end{equation}
where the function $Y^{(1)}$ is independent of $z$.
\item[\rm (d)] As $z \to  \pm s$, we have $Y(z) = \mathcal \Boh(\ln(z \mp s))$.
\end{enumerate}
\end{rhp}
By \cite{DIZ97}, it follows that
\begin{equation}\label{eq:Yexpli}
Y(z)=I-\int_{-s}^s\frac{\vec{F}(w)\vec{h}(w)^{\msf T}}{w-z}\ud w
\end{equation}
and
\begin{equation}\label{def:FH2}
\vec{F}(z)=Y(z)\vec{f}(z), \qquad \vec{H}(z)=Y(z)^{-\msf T}\vec{h}(z).
\end{equation}


We now make an undressing transformation to arrive at an RH problem that is related to $\ud F/ \ud s$ with the aid of the RH problems for $M$ and $Y$. We start with definitions
\begin{equation}\label{def:Gammajs}
 \begin{aligned}
     &\Gamma_0^{(s)}:=(s,+\infty), &&  \Gamma_1^{(s)}:=s+e^{\varphi \ii }(0,+\infty),  && \Gamma_2^{(s)}:=-s+e^{-\varphi \ii}(-\infty,0),\\
     &\Gamma_3^{(s)}:= (-\infty,-s),&& \Gamma_4^{(s)}:=-s+e^{\varphi \ii}(-\infty,0), &&\Gamma_5^{(s)}:=s+e^{-\varphi \ii}(0,+\infty), \quad 0<\varphi<\frac{\pi}{3}.
    \end{aligned}
    \end{equation}
Clearly, the rays $\Gamma_k^{(s)}$, $k=1,2,4,5$, and $\mathbb{R}$ divide the whole complex plane into 6 regions $\Omega_j^{(s)}$, $j=1,\ldots,6$; see Figure \ref{fig:X} for an illustration. The transformation is defined by
\begin{align}\label{eq:YtoX}
 &X(z) =  X(z; s,\gamma, r_1,r_2,s_1,s_2,\tau)
 \nonumber \\
 &= \begin{cases}
   Y(z)M(z), & \hbox{ $z \in \Omega_1^{(s)}\cup \Omega_3^{(s)} \cup \Omega_4^{(s)} \cup \Omega_6^{(s)} $,} \\
   Y(z)\widehat M(z), & \hbox{ $z \in \Omega_2^{(s)}$,} \\
   Y(z)\widehat M(z)\begin{pmatrix} 1 &0 &-1 &-1 \\0 &1 &-1 &-1 \\ 0&0&1&0\\0&0&0&1 \end{pmatrix}, & \hbox{ $z \in \Omega_5^{(s)}$.}
 \end{cases}
\end{align}
On account of the RH problems \ref{rhp:tac} and \ref{rhp:Y} for $M$ and $Y$, it is straightforward to check that $X$ satisfies the following RH problem.
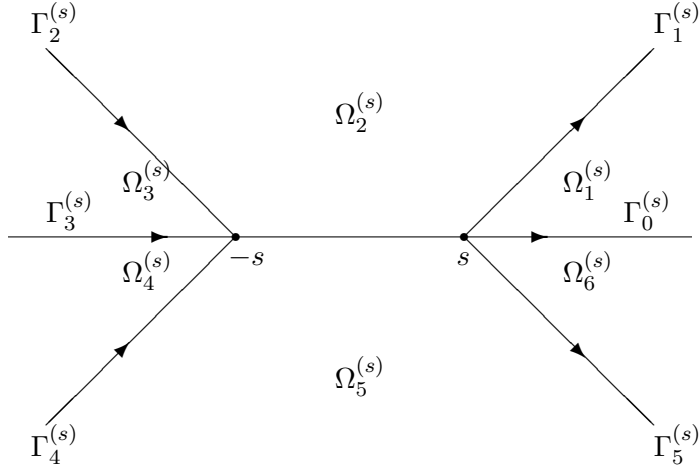
\begin{figure}[t]
\begin{center}
   \setlength{\unitlength}{1truemm}
   \begin{picture}(100,70)(-5,2)
       \put(25,40){\line(-1,0){30}}
       \put(55,40){\line(1,0){30}}

       \put(25,40){\line(1,0){30}}

       \put(25,40){\line(-1,-1){25}}
       \put(25,40){\line(-1,1){25}}

       \put(55,40){\line(1,1){25}}
       \put(55,40){\line(1,-1){25}}

       \put(15,40){\thicklines\vector(1,0){1}}
       \put(65,40){\thicklines\vector(1,0){1}}

       \put(10,55){\thicklines\vector(1,-1){1}}
       \put(10,25){\thicklines\vector(1,1){1}}
       \put(70,25){\thicklines\vector(1,-1){1}}
       \put(70,55){\thicklines\vector(1,1){1}}

       \put(-2,11){$\Gamma_4^{(s)}$}

       \put(-2,67){$\Gamma_2^{(s)}$}
       \put(0,42){$\Gamma_3^{(s)}$}
       \put(80,11){$\Gamma_5^{(s)}$}
       \put(80,67){$\Gamma_1^{(s)}$}
       \put(76,42){$\Gamma_0^{(s)}$}

       \put(10,46){$\Omega_3^{(s)}$}
       \put(10,34){$\Omega_4^{(s)}$}
       \put(68,46){$\Omega_1^{(s)}$}
       \put(68,34){$\Omega_6^{(s)}$}
       \put(38,55){$\Omega_2^{(s)}$}
       \put(38,20){$\Omega_5^{(s)}$}

       \put(25,40){\thicklines\circle*{1}}
       \put(55,40){\thicklines\circle*{1}}

       \put(24,36.3){$-s$}
       \put(54,36.3){$s$}

   \end{picture}
   \caption{Regions $\Omega_j^{(s)}$, $j=1,\ldots,6$, and the contours $\Gamma_{k}^{(s)}$, $k=0,1,\ldots,5$, for the RH problem for $\Phi$.}
   \label{fig:X}
\end{center}
\end{figure}

\begin{rhp}\label{rhp:X}
\hfill
\begin{enumerate}
\item[\rm (a)] $X(z)$ is defined and analytic in $\mathbb{C}\setminus \Gamma_X$,
where
\begin{equation}\label{def:gammaX}
\Gamma_X:=\cup^5_{j=0}\Gamma_j^{(s)} \cup[-s,s]
\end{equation}
with the rays $\Gamma_j^{(s)}$, $j=0,1,\ldots,5$, defined in \eqref{def:Gammajs} and \eqref{def:Gammajs2}; see Figure \ref{fig:X} for the orientations of $\Gamma_X$.

\item[\rm (b)] For $z \in \Gamma_X$, $X$ satisfies the jump condition
\begin{equation}\label{eq:X-jump}
 X_+(z)=X_-(z)J_X(z),
\end{equation}
where
\begin{equation}\label{def:JX}
J_X(z):=\left\{
 \begin{array}{ll}
          \begin{pmatrix}0&0&1&0\\0&1&0&0\\-1&0&0&0\\0&0&0&1 \end{pmatrix}, & \qquad \hbox{$z\in \Gamma_0^{(s)}$,} \\
          I-E_{2,1}+E_{3,1}+E_{3,4},  & \qquad  \hbox{$z\in \Gamma_1^{(s)}$,} \\
          I-E_{1,2}+E_{4,2}+E_{4,3},  & \qquad \hbox{$z\in \Gamma_2^{(s)}$,} \\
          \begin{pmatrix}1&0&0&0\\0&0&0&1\\0&0&1&0\\0&-1&0&0 \end{pmatrix}, & \qquad  \hbox{$z\in \Gamma_3^{(s)}$,} \\
          I+E_{1,2}+E_{4,2}-E_{4,3}, & \qquad  \hbox{$z\in \Gamma_4^{(s)}$,} \\
          I+E_{2,1}+E_{3,1}-E_{3,4}, & \qquad  \hbox{$z\in \Gamma_5^{(s)}$,} \\
          \begin{pmatrix}1&0&1-\gamma&1-\gamma\\0&1&1-\gamma&1-\gamma\\0&0&1&0\\0&0&0&1 \end{pmatrix}, & \qquad  \hbox{$z\in (-s,s)$.}
        \end{array}
      \right.
 \end{equation}

\item[\rm (c)]As $z \to \infty$ with $z\in \mathbb{C} \setminus \Gamma_X$, we have
\begin{align}\label{eq:asyX}
X(z)&=\left( I+\frac{X^{(1)}}{z}+ \Boh(z^{-2}) \right) \diag \left((-z)^{-\frac14},z^{-\frac14},(-z)^{\frac14},z^{\frac14} \right)
\nonumber  \\
& \quad \times A \diag \left(
e^{-\theta_1(z)+\tau z}, e^{-\theta_2(z)- \tau z}, e^{\theta_1(z)+\tau
z},e^{\theta_2(z)- \tau z} \right),
\end{align}
where $A$, $\theta_1$ and $\theta_2$ are defined in \eqref{def:A}--\eqref{def:theta2}, respectively and
\begin{equation}\label{def:X1}
X^{(1)} = Y^{(1)} + M^{(1)}
\end{equation}
with $Y^{(1)}$ and $M^{(1)}$ given in \eqref{eq:Y-infty} and \eqref{eq:asy:M}.

\item[\rm (d)] As $z \to s$, we have
\begin{align} \label{eq:X-near-s}
	X(z) &=  X_{R}(z) \begin{pmatrix}
	1 & 0 & -\frac{\gamma}{2\pi \ii} \ln(z-s) & -\frac{\gamma}{2\pi \ii} \ln(z-s)
   \\
	0 & 1 & -\frac{\gamma}{2\pi \ii} \ln(z-s) & -\frac{\gamma}{2\pi \ii} \ln(z-s)
    \\
	0 & 0 & 1 & 0
    \\
    0 & 0 & 0 & 1	
\end{pmatrix}
\nonumber
\\
& \quad \times
\begin{cases}
	I, & z \in \Omega_2^{(s)}, \\
	\begin{pmatrix}
    1&0& -1 & -1
    \\
    0& 1 & -1 & -1
    \\
    0&0&1&0
    \\
    0&0&0&1
    \end{pmatrix}, & z \in \Omega_5^{(s)},
	\end{cases}
\end{align}
where the principal branch is taken for $\ln(z-s)$, and $X_R(z)$ is analytic at $z=s$ satisfying
\begin{equation}\label{eq: Phi-expand-s}
  X_R(z) = X_{R,0}(s)\left(I+X_{R,1}(s)(z-s)+\Boh((z-s)^2)\right ),\qquad z\to s,
 \end{equation}
for some functions $X_{R,0}(s)$ and $ X_{R,1}(s)$ depending on the parameters $r_1$,$r_2$,$s_1$,$s_2$,$\tau$ and $\gamma$.

\item[\rm (e)] As $z \to -s$,
we have
\begin{align} \label{eq:X-near--s}
	X(z) &=  X_{L}(z) \begin{pmatrix}
	1 & 0 & -\frac{\gamma}{2\pi \ii} \ln(-z-s) & -\frac{\gamma}{2\pi \ii} \ln(-z-s)
   \\
	0 & 1 & -\frac{\gamma}{2\pi \ii} \ln(-z-s) & -\frac{\gamma}{2\pi \ii} \ln(-z-s)
    \\
	0 & 0 & 1 & 0
    \\
    0 & 0 & 0 & 1	
\end{pmatrix}
\nonumber
\\
& \quad \times
\begin{cases}
	 \begin{pmatrix}
    0&1& 0 & 0
    \\
    1& 0 & 0 & 0
    \\
    0&0&0&-1
    \\
    0&0&-1&0
    \end{pmatrix}, & z \in \Omega_2^{(s)}, \\
	\begin{pmatrix}
    0&1& 1 & 1
    \\
    1& 0 & 1 & 1
    \\
    0&0&0&-1
    \\
    0&0&-1&0
    \end{pmatrix}, & z \in \Omega_5^{(s)},
	\end{cases}
\end{align}
where $\ln(-z-s)$ is analytic for $z\in \mathbb{C}\setminus [-s,+\infty)$ and $X_L(z)$ is analytic at $z=-s$ satisfying
\begin{equation}\label{eq: Phi-expand--s}
  X_L(z) = X_{L,0}(s)\left(I+X_{L,1}(s)(z+s)+\Boh((z+s)^2)\right ),\qquad z\to -s,
 \end{equation}
for some functions $X_{L,0}(s)$ and $ X_{L,1}(s)$ depending on the parameters $r_1$,$r_2$,$s_1$,$s_2$,$\tau$ and $\gamma$.

\end{enumerate}
\end{rhp}
%
\begin{proposition}\label{prop:X1}
The matrix-valued function
$X(z)$ satisfies the following symmetric relations:
\begin{align} \label{eq:symm}
  X (-z) &=
  \begin{pmatrix}
  J & 0
  \\
  0 & -J
  \end{pmatrix}
  \widetilde X(z)
  \begin{pmatrix}
  J & 0
  \\
  0 & -J
  \end{pmatrix},
\\
  X (z)^{- \msf T} &=
  \begin{pmatrix}
  0 & -I_2
  \\
  I_2 & 0
  \end{pmatrix}\dot X(z)
  \begin{pmatrix}
  0 & I_2
  \\
  -I_2 & 0
  \end{pmatrix},\label{eq:symmX2}
\end{align}
where $J$ is given in \eqref{def-J}, $\widetilde X$ and $\dot X$ are defined through \eqref{def:tildeX} and \eqref{def:dotX}. Moreover, the matrix
$X^{(1)} = X^{(1)}(\gamma, r_1,r_2,s_1,s_2,\tau)$ in \eqref{eq:asyX} satisfies
\begin{align}
 X^{(1)}  &=
  -\begin{pmatrix}
  J & 0
  \\
  0 & -J
  \end{pmatrix} \widetilde X^{(1)}
  \begin{pmatrix}
  J & 0
  \\
  0 & -J
  \end{pmatrix},\label{eq:symmX11}
\\
  (X^{(1)})^{\msf T} &=
  -\begin{pmatrix}
  0 & -I_2
  \\
  I_2 & 0
  \end{pmatrix}\dot X^{(1)}
  \begin{pmatrix}
  0 & I_2
  \\
  -I_2 & 0
  \end{pmatrix}.\label{eq:symmX12}
\end{align}
\end{proposition}
\begin{proof}
One can check that the left and right hand sides of \eqref{eq:symm} satisfy the same RH problem. Then \eqref{eq:symm} follows from the uniqueness of
the solution to this RH problem. The same argument applies to \eqref{eq:symmX2}. By substituting the asymptotic behavior of $X(z)$ in \eqref{eq:asyX} into equations \eqref{eq:symm} and \eqref{eq:symmX2}, the symmetric relations \eqref{eq:symmX11} and \eqref{eq:symmX12} follow from a straightforward calculation.

This finishes the proof of Proposition \ref{prop:X1}.
\end{proof}

The connection between the derivatives of $F$ and the RH problem for $X$ is revealed in the following proposition.
\begin{proposition}\label{prop:derivativeandX}
With $F$ defined in \eqref{def:Fnotation}, we have
\begin{align}
\frac{\ud}{\ud s} F(s;\gamma, r_1, r_2,s_1,s_2, \tau) =-\frac{\gamma}{2 \pi \ii}\left[ \lim_{z \to s} \sum_{i=3}^4\sum_{j=1}^2 \left(X(z)^{-1}X'(z)\right)_{ij}+ \lim_{z \to -s} \sum_{i=3}^4\sum_{j=1}^2 \left(X(z)^{-1} X'(z)\right)_{ij} \right] ,
\label{eq:derivativeinsX-2}
\end{align}
where the limit is taken from $\Omega_2^{(s)}$ and
\begin{align}
\frac{\partial}{\partial s_{1}} F(s;\gamma, r_1, r_2,s_1,s_2, \tau) &= 2\ii\left(X^{(1)}_{13}-M^{(1)}_{13}\right),
\label{eq:derivativeFs12}\\
\frac{\partial}{\partial s_{2}} F(s;\gamma, r_1, r_2,s_1,s_2, \tau) &= 2\ii\left(\widetilde X^{(1)}_{13}-\widetilde M^{(1)}_{13}\right),
\label{eq:derivativeFs122}\\
\frac{\partial}{\partial \tau} F(s;\gamma, r_1, r_2,s_1,s_2, \tau) &= -X^{(1)}_{11}- \widetilde X^{(1)}_{11}+\dot X^{(1)}_{11}+\dot{\widetilde{X}}^{(1)}_{11}\nonumber\\
&\quad+M^{(1)}_{11}+ \widetilde M^{(1)}_{11}-\dot M^{(1)}_{11}-\dot{\widetilde{M}}^{(1)}_{11},
\label{eq:derivativeFtau}
\end{align}
where $X^{(1)}$ and $M^{(1)}$ are given in \eqref{eq:asyX} and \eqref{eq:asy:M}. Here, for a function $x=x(\cdot; r_1, r_2,s_1,s_2, \tau)$, we have $\dot{\widetilde{x}}=x(\cdot; r_2, r_1,s_2,s_1, -\tau))$; see the notations \eqref{def:tildeX} and \eqref{def:dotX}.
\end{proposition}
To prove Proposition \ref{prop:derivativeandX}, we need the following lemma.
\begin{lemma}\label{prop:M}
Let $M$ be the unique solution to the tacnode RH problem \ref{rhp:tac}, we have
\begin{align}
\frac{\partial M}{\partial s_{1}} &= -2\ii\left( 
E_{1,3}+zE_{3,1}+\left[M^{(1)}, E_{3,1}\right]\right) M, \label{eq:lax:M1}\\
\frac{\partial M}{\partial s_{2}} &= -2\ii\left( 
-E_{2,4}+zE_{4,2}+\left[M^{(1)}, E_{4,2}\right]\right) M, \label{eq:lax:M3}\\
\frac{\partial M}{\partial \tau} &= \left( \begin{pmatrix}
z & 0 & 0 & 0\\
0 & -z & 0 & 0\\
0 & 0 & z & 0\\
0 & 0 & 0 & -z
\end{pmatrix}+\left[M^{(1)}, \begin{pmatrix}
1 & 0 & 0 & 0\\
0 & -1 & 0 & 0\\
0 & 0 & 1 & 0\\
0 & 0 & 0 & -1
\end{pmatrix}\right]\right)M, \label{eq:lax:M2}
\end{align}
where $[A,B]$ denotes the commutation of two matrices, i.e., $[A,B]=AB-BA$.
\end{lemma}
\begin{proof}
Since the RH problem for $M$ has constant jumps, we obtain from the local behavior of $M$ near the origin that $(\partial M / \partial s_{1}) M^{-1}$ is entire.
As $z \to \infty$, we find from \eqref{eq:asy:M} that
\begin{align}
\frac{\partial M}{\partial s_{1}} M^{-1} &= \left( I+\frac{M^{(1)}}{z}+ \Boh(z^{-2}) \right) \diag \left((-z)^{-\frac14},z^{-\frac14},(-z)^{\frac14},z^{\frac14} \right)A \diag \left(-2(-z)^{\frac 12},0,2(-z)^{\frac 12},0 \right)
\nonumber  \\
& \quad \times A^{-1} \diag \left((-z)^{\frac14},z^{\frac14},(-z)^{-\frac14},z^{-\frac14} \right) \left( I-\frac{M^{(1)}}{z}+ \Boh(z^{-2}) \right) + \Boh(z^{-1}) \nonumber\\
&= -2\ii\left( E_{1,3}+zE_{3,1}+\left[M^{(1)}, E_{3,1}\right]\right) + \Boh(z^{-1})
\end{align}
Keeping only the polynomial terms in $z$, we obtain \eqref{eq:lax:M1}. A similar argument yields \eqref{eq:lax:M3} and \eqref{eq:lax:M2}.

This finishes the proof of Lemma \ref{prop:M}.
\end{proof}
\begin{proof}[Proof of Proposition \ref{prop:derivativeandX}]
For $z\in \Omega_2^{(s)}$, we see from \eqref{def:fh}, \eqref{def:FH2} and \eqref{eq:YtoX} that
\begin{equation}\label{eq:FHinX}
\vec{F}(z)=Y(z)\vec{f}(z)=Y(z)\widehat M(z)
\begin{pmatrix}
1
\\
1
\\
0
\\
0
\end{pmatrix}= X(z)
\begin{pmatrix}
1
\\
1
\\
0
\\
0
\end{pmatrix}
\end{equation}
and
\begin{align}
\vec{H}(z)&=Y(z)^{-\msf T}\vec{h}(z)=  X(z)^{-\msf T}\widehat M(z)^{\msf T} \cdot \frac{\gamma }{2 \pi \ii} \widehat M(z)^{-\msf T}
\begin{pmatrix}
0
\\
0
\\
1
\\
1
\end{pmatrix}
=\frac{\gamma }{2 \pi \ii}  X(z)^{-\msf T}
\begin{pmatrix}
0
\\
0
\\
1
\\
1
\end{pmatrix}.
\end{align}
Combining the above formulas and \eqref{eq:resolventexpli}, we obtain
\begin{equation} \label{eq: R-Xij}
	R(z,z) = \frac{\gamma}{2 \pi \ii}  \sum_{i=3}^4\sum_{j=1}^2 \left(X(z)^{-1}X'(z)\right)_{ij}, \qquad  z\in \Omega_2^{(s)}.
\end{equation}
This, together with \eqref{eq:derivatives}, gives us \eqref{eq:derivativeinsX-2}.

To show \eqref{eq:derivativeFs12}, we note from \eqref{def:fh} and \eqref{eq:lax:M1} that
\begin{align}
\frac{\partial \vec{f}}{\partial s_1}(x) = -2\ii\left( E_{1,3}+xE_{3,1}+\left[M^{(1)}, E_{3,1}\right]\right)\vec{f}(x),\\
\frac{\partial \vec{h}}{\partial s_1}(y) = 2\ii\left(  E_{1,3}+yE_{3,1}+\left[M^{(1)}, E_{3,1}\right]\right)^{\msf T}\vec{h}(y).
\end{align}
The above equations, together with \eqref{eq:tildeKdef}, imply that
\begin{align}
\frac{\partial}{\partial s_1} \left(\gamma K_\tac (x,y)\right)& = \frac{\frac{\partial \vec{f}^{\msf T}}{\partial s_1}(x)\vec{h}(y)+ \vec{f}(x)^{\msf T}\frac{\partial \vec{h}}{\partial s_1}(y)}{x-y} \nonumber\\
&=-2 \ii \vec{f}(x)^{\msf T} E_{1,3}\vec{h}(y) = -2 \ii f_1(x)h_3(y).
\end{align}
Thus, it is readily seen from \eqref{def:Fnotation} and \eqref{def:FH} that
\begin{align}
\frac{\partial}{\partial s_1}F(s;\gamma, r_1, r_2,s_1,s_2, \tau) &=\frac{\partial}{\partial s_1} \ln \det\left(I-\gamma  \mathcal K_{\tac}\right) = -\tr{\left((I - \gamma  \mathcal K_{\tac})^{-1} \frac{\partial}{\partial s_1} \left(\gamma \mathcal K_\tac \right)\right)}\nonumber\\
&=2\ii \int_{-s}^s F_1(v) h_3(v) \ud v
\end{align}
On the other hand, from \eqref{eq:Y-infty} and \eqref{eq:Yexpli} we have
\begin{align}
Y^{(1)} = \int_{-s}^s \vec{F}(v) \vec{h}(v)^{\msf T} \ud v.
\end{align}
A combination of the above two equations gives us
\begin{equation}
\frac{\partial}{\partial s_{1}} F(s;\gamma, r_1, r_2,s_1,s_2, \tau) = 2\ii Y^{(1)}_{13}.
\end{equation}
We thus obtain \eqref{eq:derivativeFs12} by applying \eqref{def:X1} to the above formula.

Through a similar calculation, we have 
\begin{align}
\frac{\partial}{\partial s_{2}} F(s;\gamma, r_1, r_2,s_1,s_2, \tau) &= 2\ii Y^{(1)}_{24},\\
\frac{\partial}{\partial \tau} F(s;\gamma, r_1, r_2,s_1,s_2, \tau) &= -Y^{(1)}_{11}+Y^{(1)}_{22}-Y^{(1)}_{33}+Y^{(1)}_{44}.
\end{align}
Substituting \eqref{def:X1} into the above equations and using the symmetric relations \eqref{M13}, \eqref{M11}, \eqref{eq:symmX11} and \eqref{eq:symmX12}, we arrive at \eqref{eq:derivativeFs122} and \eqref{eq:derivativeFtau}.

This completes the proof of Proposition \ref{prop:derivativeandX}.
\end{proof}

\section{Lax pair equations and differential identities for the Hamiltonian}\label{sec:Lax}
In this section, we will derive a Lax pair for $X(z;s)$ from RH problem \ref{rhp:X}. Several useful differential identities for the associated Hamiltonian will also be presented for later use.

\subsection{The Lax system for $X$}

\begin{proposition}\label{pro:lax}
For the matrix-valued function $X(z) = X(z;s)$ defined in \eqref{eq:YtoX}, we have
\begin{equation}\label{eq:lax}
\frac{\partial}{\partial z} X(z;s) = L(z;s) X(z;s), \qquad \frac{\partial}{\partial s} X(z;s) = U(z;s) X(z;s),
\end{equation}
where
\begin{equation}\label{eq:L}
L(z;s) =\ii(r_1 E_{3,1}- r_2 E_{4,2})z+A_0(s)+\frac{A_1(s)}{z-s}+\frac{A_2(s)}{z+s},
\end{equation}
and
\begin{equation}\label{eq:U}
U(z;s) = -\frac{A_1(s)}{z-s}+\frac{A_2(s)}{z+s}
\end{equation}
with the functions $A_k(s)$, $k=0,1,2$, given in \eqref{def:A0}--\eqref{def:A2}, respectively. Moreover, the functions $p_i (s)$ and $q_i(s)$, $i=1,\ldots,6$ in the definitions of $A_k(s)$ satisfy the equations \eqref{def:pq's}, \eqref{eq:sumpq} and
\begin{align}\label{def:q52}
p_3(s)q_1(s)-\widetilde p_4(s) \widetilde q_2(s) = -\ii r_1 q_5(s)-\frac{p_5(s)^2}{\ii r_1} + \ii r_2 q_6(s) \widetilde q_6(s)-\ii s_1.
\end{align}
\end{proposition}

\begin{proof}
The proof is based on the RH problem \ref{rhp:X} for $X$. Since the jump matrices for $X$ are all independent of $z$ and $s$, it's easily seen that
\begin{equation}
L(z;s)=\frac{\partial}{\partial z} X(z;s) X(z;s)^{-1}, \quad  U(z;s)=\frac{\partial}{\partial s} X(z;s) X(z;s)^{-1}
\end{equation}
are analytic in the complex $z$ plane except for possible isolated singularities at $z=\pm s$ and $z=\infty$. We next calculate the functions $L(z;s)$ and $U(z;s)$ one by one.

From the large $z$ behavior of $X$ given in \eqref{eq:asyX}, we have, as $z \to \infty$,
\begin{multline}\label{eq:inftyL}
L(z;s)=\ii(r_1 E_{3,1}- r_2 E_{4,2})z+\begin{pmatrix}
\tau & 0 & \ii r_1 & 0\\
0 & -\tau & 0 & \ii r_2\\
-\ii s_1 & 0 & \tau & 0\\
0 & -\ii s_2 & 0 & -\tau
\end{pmatrix}
\\
+ \left[X^{(1)},\ii(r_1 E_{3,1}- r_2 E_{4,2})\right]+\Boh(z^{-1}),
\end{multline}
where $X^{(1)}$ is given in \eqref{eq:asyX}. This gives us the leading term in \eqref{eq:L} and
\begin{equation}\label{A01}
A_0(s)=\begin{pmatrix}
\tau & 0 & \ii r_1 & 0\\
0 & -\tau & 0 & \ii r_2\\
-\ii s_1 & 0 & \tau & 0\\
0 & -\ii s_2 & 0 & -\tau
\end{pmatrix} + \left[X^{(1)},\ii(r_1 E_{3,1}- r_2 E_{4,2})\right].
\end{equation}
According to the symmetric relation of $X^{(1)}$ established in \eqref{eq:symmX11}, it follows that
\begin{equation}
\widetilde A_0(s)=-\begin{pmatrix}
J&0\\0&-J
\end{pmatrix}A_0(s)
\begin{pmatrix}
J&0\\0&-J
\end{pmatrix},
\end{equation}
where $J$ is given in \eqref{def-J}. Thus, if we define \begin{align}
p_5(s) & = \ii r_1 X^{(1)}_{13}, & p_6(s)&  = -\ii r_1 X^{(1)}_{43} -\ii r_2 X^{(1)}_{21},\label{def:p56}\\
q_5(s) & = X^{(1)}_{33}-X^{(1)}_{11}, & q_6(s)&  = X^{(1)}_{14},\label{def:q56}
\end{align}
the formula of $A_0(s)$ in \eqref{def:A0} follows directly by combining \eqref{A01}--\eqref{def:q56}.

If $z \to s$, it follows from \eqref{eq:X-near-s} that
\begin{equation}
L(z;s) \sim \frac{A_1(s)}{z-s},
\end{equation}
where
\begin{equation}\label{def:A1-1}
A_1(s) = -\frac{\gamma}{2 \pi \ii} X_{R,0}(s) \begin{pmatrix}
0 & 0 & 1 & 1\\
0 & 0 & 1 & 1\\
0 & 0 & 0 & 0\\
0 & 0 & 0 & 0
\end{pmatrix}X_{R,0}(s)^{-1}
\end{equation}
with $X_{R,0}(s)$ given in \eqref{eq: Phi-expand-s}. By setting
\begin{equation}\label{def:qkpk}
\begin{pmatrix}
q_1(s)\\q_2(s)\\q_3(s)\\q_4(s)
\end{pmatrix}=X_{R,0}(s) \begin{pmatrix}
1\\1\\0\\0
\end{pmatrix} \quad \textrm{and} \quad \begin{pmatrix}
p_1(s)\\p_2(s)\\p_3(s)\\p_4(s)
\end{pmatrix}=-\frac{\gamma}{2 \pi \ii}X_{R,0}(s)^{-\msf T} \begin{pmatrix}
0\\0\\1\\1
\end{pmatrix},
\end{equation}
we obtain the expression of $A_1(s)$ in \eqref{def:A1}. From \eqref{def:A1-1}, we also note that
\begin{equation}\label{eq:trA1}
\tr A_1(s) = \sum_{k=1}^4 q_k(s)p_k(s)=0.
\end{equation}
If $z \to -s$, we have from \eqref{eq:X-near--s} that
\begin{equation}
L(z;s) \sim \frac{A_2(s)}{z+s},
\end{equation}
where
\begin{equation}\label{def:A2-1}
A_2(s) = -\frac{\gamma}{2 \pi \ii} X_{L,0}(s) \begin{pmatrix}
0 & 0 & 1 & 1\\
0 & 0 & 1 & 1\\
0 & 0 & 0 & 0\\
0 & 0 & 0 & 0
\end{pmatrix}X_{L,0}(s)^{-1}
\end{equation}
with $X_{L,0}(s)$ given in \eqref{eq: Phi-expand--s}. On account of the symmetric relation \eqref{eq:symm}, it is readily seen that
\begin{equation}
X_{L,0}(s;r_1, r_2, s_1, s_2, \tau)=\begin{pmatrix}
J & 0\\
0 & -J
\end{pmatrix}X_{R,0}(s; r_2, r_1, s_2, s_1, \tau).
\end{equation}
This, together with \eqref{def:A1-1} and \eqref{def:A2-1}, implies that
\begin{equation}
A_2(s)=\begin{pmatrix}
J&0\\
0&-J
\end{pmatrix}
\widetilde A_1(s)\begin{pmatrix}
J&0\\
0&-J
\end{pmatrix},
\end{equation}
which leads to the expression for $A_2(s)$ in \eqref{def:A2}.

The computation of $U(z;s)$ is similar. It is easy to check that
\begin{equation}
U(z;s)=\Boh(z^{-1}), \qquad z \to \infty,
\end{equation}
and
\begin{equation}
U(z;s) \sim -\frac{A_1(s)}{z-s}, \quad z \to s; \quad U(z;s) \sim \frac{A_2(s)}{z+s}, \quad z \to -s,
\end{equation}
where $A_1(s)$ and $A_2(s)$ are given in \eqref{def:A1-1} and \eqref{def:A2-1}. The above equations imply the expression of $U(z;s)$ in \eqref{eq:U}.

It remains to establish the various relations satisfied by the functions $p_i (s)$ and $q_i(s)$, $i=1,\ldots,6$, in the definitions of $A_k(s)$, $k=0,1,2$. By \eqref{eq:trA1}, we have proved \eqref{eq:sumpq}.  The relation \eqref{def:q52} follows by computing the $(1,3)$ entry of the $\Boh(z^{-1})$ term on both sides of \eqref{eq:inftyL}. To show the differential equations in \eqref{def:pq's}, we recall that
the compatibility condition
\begin{equation}
\frac{\partial^2}{\partial z \partial s} X(z;s) = \frac{\partial^2}{\partial s \partial z} X(z;s)
\end{equation}
for the Lax pair \eqref{eq:lax} is the zero curvature relation
\begin{equation}\label{eq:compatibility}
\frac{\partial}{\partial s}L(s;z) - \frac{\partial}{\partial z}U(s;z) = A_0'(s)+\frac{A_1'(s)}{z-s}+\frac{A_2'(s)}{z+s}=[U,L].
\end{equation}
Inserting \eqref{eq:L} and \eqref{eq:U} into the above equation and taking $z \to \infty$, we get
\begin{equation}
A_0'(s) = \left[A_2(s)-A_1(s),\ii (r_1 E_{3,1}-r_2 E_{4,2})\right],
\end{equation}
which leads to
\begin{equation}
\begin{cases}
p_5'(s)=-\ii r_1(p_3(s)q_1(s)+\widetilde p_4(s)\widetilde q_2(s)),\\
p_6'(s)=\ii r_1 (p_3(s)q_4(s)-\widetilde p_4(s)\widetilde q_3(s))+\ii r_2 (p_1(s)q_2(s)-\widetilde p_2(s)\widetilde q_1(s)),\\
q_5'(s)=p_1(s)q_1(s)-p_3(s)q_3(s)-\widetilde p_2(s) \widetilde q_2(s)+\widetilde p_4(s) \widetilde q_4(s),\\
q_6'(s)=-p_4(s)q_1(s)-\widetilde p_3(s) \widetilde q_2(s).
\end{cases}
\end{equation}

If we calculate the residue at $z=s$ on the both sides of \eqref{eq:compatibility}, it is easily seen that
\begin{equation}\label{def:A1's}
A_1'(s)=-\left[A_1(s),\ii (r_1 E_{3,1}-r_2 E_{4,2})s+A_0(s)+\frac{A_2(s)}{s}\right].
\end{equation}
On account of \eqref{eq:lax}--\eqref{eq:U}, we obtain
\begin{align}
&\left(\frac{\partial}{\partial z}X(z;s)+\frac{\partial}{\partial s}X(z;s)\right)X(z;s)^{-1}=L(z;s)+U(z;s)\nonumber\\
&\sim \ii (r_1 E_{3,1}-r_2 E_{4,2})s+A_0(s)+\frac{A_2(s)}{s}, \quad z \to s.
\end{align}
On the other hand, substituting \eqref{eq:X-near-s} into the left hand side of the above equation gives us
\begin{equation}
\left(\frac{\partial}{\partial z}X(z;s)+\frac{\partial}{\partial s}X(z;s)\right)X(z;s)^{-1} \sim \frac{\ud}{\ud s} X_{R,0}(s) \cdot X_{R,0}(s)^{-1}, \quad z \to s.
\end{equation}
Therefore, we have from the above two formulas that
\begin{equation}
\frac{\ud}{\ud s} X_{R,0}(s)=\left(\ii (r_1 E_{3,1}-r_2 E_{4,2})s+A_0(s)+\frac{A_2(s)}{s}\right)X_{R,0}(s).
\end{equation}
Recall the definitions of $q_k(s)$, $k=1,\ldots,4$, given in \eqref{def:qkpk}, it is readily seen that
\begin{equation}\label{eq:q's}
\begin{pmatrix}
q_1'(s)\\q_2'(s)\\q_3'(s)\\q_4'(s)
\end{pmatrix}=\left(\ii (r_1 E_{3,1}-r_2 E_{4,2})s+A_0(s)+\frac{A_2(s)}{s}\right)\begin{pmatrix}
q_1(s)\\q_2(s)\\q_3(s)\\q_4(s)
\end{pmatrix}.
\end{equation}
It then follows from \eqref{def:A0}, \eqref{def:A2} and straightforward calculations that
\begin{equation}
 \left\{\begin{aligned}
  q_1'(s)&=p_5(s)q_1(s)-\ii r_2 q_2(s)q_6(s)+\ii r_1 q_3(s)+\tau q_1(s)\\
&\quad +\frac{\widetilde q_2(s)}{s}(\widetilde p_2(s)q_1(s)+\widetilde p_1(s)q_2(s) -\widetilde p_4(s) q_3(s) -\widetilde p_3(s)q_4(s)),\\
q_2'(s)&=\ii r_1 q_1(s)\widetilde{q}_6(s)-\widetilde{p}_5(s)q_2(s)+\ii r_2q_4(s)-\tau q_2(s)\\
&\quad+\frac{\widetilde q_1(s)}{s}(\widetilde p_2(s)q_1(s)+\widetilde p_1(s)q_2(s) -\widetilde p_4(s) q_3(s) -\widetilde p_3(s)q_4(s)),\\
q_3'(s)&= \ii r_1 sq_1(s)+\ii r_1 q_1(s)q_5(s) -\widetilde p_6(s) q_2(s)-p_5(s)q_3(s)-\ii r_1q_4(s)q_6(s)-\ii s_1 q_1(s)\\
&\quad+\tau q_3(s)-\frac{\widetilde q_4(s)}{s}(\widetilde p_2(s)q_1(s)+\widetilde p_1(s)q_2(s) -\widetilde p_4(s) q_3(s) -\widetilde p_3(s)q_4(s)),\\
q_4'(s)&=-\ii r_2sq_2(s)-p_6(s)q_1(s)+\ii r_2q_2(s)\widetilde q_5(s) + \ii r_2 q_3(s)\widetilde q_6(s)+\widetilde p_5(s)q_4(s)-\ii s_2 q_2(s)\\
&\quad-\tau q_4(s)-\frac{\widetilde q_3(s)}{s}(\widetilde p_2(s)q_1(s)+\widetilde p_1(s)q_2(s) -\widetilde p_4(s) q_3(s) -\widetilde p_3(s)q_4(s)).
 \end{aligned}\right.
\end{equation}
To show the derivatives of $p_i(s)$, $i=1,\ldots,4$, we see from \eqref{def:A1} and \eqref{def:A1's} that
\begin{align}
A_1'(s)=& \begin{pmatrix}
q_1'(s)\\q_2'(s)\\q_3'(s)\\q_4'(s)
\end{pmatrix}\begin{pmatrix}
p_1(s) & p_2(s)&p_3(s)&p_4(s)
\end{pmatrix}+ \begin{pmatrix}
q_1(s)\\q_2(s)\\q_3(s)\\q_4(s)
\end{pmatrix}\begin{pmatrix}
p_1'(s) & p_2'(s)&p_3'(s)&p_4'(s)
\end{pmatrix}\nonumber\\
=&-\left[A_1(s),\ii (r_1 E_{3,1}-r_2 E_{4,2})s+A_0(s)+\frac{A_2(s)}{s}\right]\nonumber\\
=&-\begin{pmatrix}
q_1(s)\\q_2(s)\\q_3(s)\\q_4(s)
\end{pmatrix}\begin{pmatrix}
p_1(s) & p_2(s)&p_3(s)&p_4(s)
\end{pmatrix}\left(\ii (r_1 E_{3,1}-r_2 E_{4,2})s+A_0(s)+\frac{A_2(s)}{s}\right)\nonumber\\
&+\left(\ii (r_1 E_{3,1}-r_2 E_{4,2})s+A_0(s)+\frac{A_2(s)}{s}\right)\begin{pmatrix}
q_1(s)\\q_2(s)\\q_3(s)\\q_4(s)
\end{pmatrix}\begin{pmatrix}
p_1(s) & p_2(s)&p_3(s)&p_4(s)
\end{pmatrix}
\end{align}
A combination of this formula and \eqref{eq:q's} gives us
\begin{align}
&\begin{pmatrix}
p_1'(s) & p_2'(s)&p_3'(s)&p_4'(s)
\end{pmatrix}\nonumber\\
=&-\begin{pmatrix}
p_1(s) & p_2(s)&p_3(s)&p_4(s)
\end{pmatrix}\left(\ii (r_1 E_{3,1}-r_2 E_{4,2})s+A_0(s)+\frac{A_2(s)}{s}\right),
\end{align}
or equivalently, by using \eqref{def:A0}, \eqref{def:A2},
\begin{equation}
\left\{
\begin{aligned}
p_1'(s)&=-\ii r_1 s p_3(s) -p_1(s)p_5(s)-\ii r_1 p_2(s)\widetilde q_6(s)-\ii r_1 p_3(s)q_5(s)+p_4(s)p_6(s)-\tau p_1(s)\\
&\quad+\ii s_1 p_3(s)-\frac{\widetilde p_2(s)}{s} (p_1(s) \widetilde q_2(s)+p_2(s) \widetilde q_1(s) - p_3(s) \widetilde q_4(s)-p_4(s) \widetilde q_3(s)),\\
p_2'(s)&=\ii r_2 s p_4(s) +\ii r_2 p_1(s) q_6(s)+p_2(s)\widetilde p_5(s) +p_3(s)\widetilde p_6(s)-\ii r_2 p_4(s)\widetilde q_5(s) +\tau p_2(s)\\
&\quad+\ii s_2 p_4(s)+\frac{\widetilde p_1(s)}{s} (p_1(s) \widetilde q_2(s)+p_2(s) \widetilde q_1(s) - p_3(s) \widetilde q_4(s)-p_4(s) \widetilde q_3(s)),\\
p_3'(s)&=-\ii r_1 p_1(s) +p_3(s)p_5(s)-\ii r_2p_4(s)\widetilde q_6(s) -\tau p_3(s)\\
&\quad+\frac{\widetilde p_4(s)}{s} (p_1(s) \widetilde q_2(s)+p_2(s) \widetilde q_1(s) - p_3(s) \widetilde q_4(s)-p_4(s) \widetilde q_3(s)),\\
p_4'(s) &= -\ii r_2 p_2(s) + \ii r_1 p_3(s)q_6(s) -p_4(s)\widetilde p_6(s)+\tau p_4(s)\\
&\quad+\frac{\widetilde p_3 (s)}{s} (p_1(s) \widetilde q_2(s)+p_2(s) \widetilde q_1(s) - p_3(s) \widetilde q_4(s)-p_4(s) \widetilde q_3(s)).
\end{aligned}\right.
\end{equation}

This completes the proof of Proposition \ref{pro:lax}.
\end{proof}

From the general theory of Jimbo-Miwa-Ueno \cite{JMU81}, we have the Hamiltonian associated with the Lax system \eqref{eq:lax} is given by
\begin{align}\label{eq:JMU}
H(s) &= \frac{\gamma}{2 \pi \ii} \tr{\left( \left(X_{L,1}(s)-X_{R,1}(s)\right)\begin{pmatrix}
0 & 0 & 1 & 1\\
0 & 0 & 1 & 1\\
0 & 0 & 0 & 0\\
0 & 0 & 0 & 0
\end{pmatrix}\right)}\nonumber\\
& = \frac{\gamma}{2 \pi \ii} \sum_{i=3}^4 \sum_{j=1}^2 \left(X_{L,1}(s)-X_{R,1}(s)\right)_{ij},
\end{align}
where $X_{R,1}(s)$ and $X_{L,1}(s)$ are given in \eqref{eq: Phi-expand-s} and \eqref{eq: Phi-expand--s}, respectively.
Taking $z \to s$ in the first equation of the Lax pair \eqref{eq:lax}, we obtain from \eqref{eq:L} and \eqref{eq:X-near-s} that the $\Boh(1)$ term gives
\begin{align}
X_{R,1}(s) &= \frac{\gamma}{2 \pi \ii}\left[X_{R,1}(s), \begin{pmatrix}
0 & 0 & 1 & 1\\
0 & 0 & 1 & 1\\
0 & 0 & 0 & 0\\
0 & 0 & 0 & 0
\end{pmatrix}\right] \nonumber\\
& \quad + X_{R,0}(s)^{-1} \left[\ii (r_1 E_{3,1}-r_2 E_{4,2})s + A_0(s)+\frac{A_2(s)}{2s}\right] X_{R,0}(s).
\end{align}
Similarly, by taking $z \to -s$ we see from \eqref{eq:X-near--s} that
\begin{align}
X_{L,1}(s) &= \frac{\gamma}{2 \pi \ii}\left[X_{L,1}(s), \begin{pmatrix}
0 & 0 & 1 & 1\\
0 & 0 & 1 & 1\\
0 & 0 & 0 & 0\\
0 & 0 & 0 & 0
\end{pmatrix}\right] \nonumber\\
& \quad + X_{L,0}(s)^{-1} \left[-\ii (r_1 E_{3,1}-r_2 E_{4,2})s + A_0(s)-\frac{A_1(s)}{2s}\right] X_{L,0}(s).
\end{align}
Inserting the above two equations into \eqref{eq:JMU} yields
\begin{align}
H(s) &= -\frac{\gamma}{2 \pi \ii} \begin{pmatrix}
0 & 0 & 1 & 1
\end{pmatrix}\left(X_{R,0}(s)^{-1} \left[\ii (r_1 E_{3,1}-r_2 E_{4,2})s + A_0(s)+\frac{A_2(s)}{2s}\right] X_{R,0}(s)\right.\nonumber\\
&\left.\quad+X_{L,0}(s)^{-1} \left[\ii (r_1 E_{3,1}-r_2 E_{4,2})s - A_0(s)+\frac{A_1(s)}{2s}\right] X_{L,0}(s)\right) \begin{pmatrix}
1\\1\\0\\0
\end{pmatrix}.
\end{align}
Recall the definitions of $A_k(s)$, $k=0, 1, 2$, and $q_i(s), p_i(s)$, $i=1, \ldots 4$, given in \eqref{def:A0}--\eqref{def:A2} and \eqref{def:qkpk}, we recover the definition of $H$ given in \eqref{def:H}, or equivalently,
\begin{align}\label{def:H2}
H(s)=&~s\left(\ii r_1(p_3(s)q_1(s)-\widetilde p_4(s) \widetilde q_2(s))+\ii r_2(\widetilde p_3(s) \widetilde q_1(s)-p_4(s)q_2(s))\right)+p_5(s)(p_1(s)q_1(s)\nonumber\\
&-p_3(s)q_3(s)-\widetilde p_2(s) \widetilde q_2(s)+\widetilde p_4(s) \widetilde q_4(s))+\widetilde p_5(s)(\widetilde p_1(s) \widetilde q_1(s)-\widetilde p_3(s) \widetilde q_3(s)-p_2(s)q_2(s)\nonumber\\
&+p_4(s)q_4(s))-p_6(s)(\widetilde p_3(s) \widetilde q_2(s) +p_4(s)q_1(s))-\widetilde p_6(s)(p_3(s)q_2(s)+\widetilde p_4(s) \widetilde q_1(s))\nonumber\\
&+\ii r_1 q_5(s)(p_3(s)q_1(s)+\widetilde p_4(s) \widetilde q_2(s)) +\ii r_2  \widetilde q_5(s)(\widetilde p_3(s) \widetilde q_1(s)+p_4(s)q_2(s))\nonumber\\
&+\ii r_1 q_6(s) (\widetilde p_4(s) \widetilde q_3(s)-p_3(s)q_4(s))+\ii r_2 q_6(s)(\widetilde p_2(s) \widetilde q_1(s) -p_1(s)q_2(s))\nonumber\\
&+\ii r_1 \widetilde q_6(s)(p_2(s)q_1(s)-\widetilde p_1(s) \widetilde q_2(s))+\ii r_2 \widetilde q_6(s)(p_4(s)q_3(s) -\widetilde p_3(s) \widetilde q_4(s))+\ii r_1(\widetilde p_2(s) \widetilde q_4(s)\nonumber\\
 &+p_1(s)q_3(s))+\ii r_2(p_2(s)q_4(s)+\widetilde p_1(s) \widetilde q_3(s)) +\tau (p_1(s)q_1(s)+ p_3(s)q_3(s)-p_2(s)q_2(s)\nonumber\\
 &-p_4(s)q_4(s)+\widetilde p_1(s) \widetilde q_1(s)+\widetilde p_3(s) \widetilde q_3(s)-\widetilde p_2(s) \widetilde q_2(s)-\widetilde p_4(s) \widetilde q_4(s))-\ii s_1(p_3(s)q_1(s)\nonumber\\
 &+\widetilde p_4(s) \widetilde q_2(s))-\ii s_2(\widetilde p_3(s) \widetilde q_1(s)+p_4(s)q_2(s))+\frac{1}{s}(\widetilde p_2(s) q_1(s) +\widetilde p_1(s) q_2(s)-\widetilde p_4(s)q_3(s)\nonumber\\
 &-\widetilde p_3(s)q_4(s))(p_2(s)\widetilde q_1(s)+p_1(s) \widetilde q_2(s)-p_4(s)\widetilde q_3(s)-p_3(s) \widetilde q_4(s)).
\end{align}

\subsection{Differential identities for the Hamiltonian}
\begin{proposition}\label{prop:H}
With the Hamiltonian $H$ defined in \eqref{def:H}, we have
\begin{align}\label{eq:differentialH}
\frac{\ud}{\ud s}H(s) &= \ii r_1(p_3(s)q_1(s)-\widetilde p_4(s) \widetilde q_2(s))+\ii r_2(\widetilde p_3(s) \widetilde q_1(s)-p_4(s)q_2(s))-\frac{1}{s^2}(\widetilde p_2(s) q_1(s) \widetilde p_1(s) q_2(s)\nonumber\\
&\quad-\widetilde p_4(s)q_3(s)-\widetilde p_3(s)q_4(s))(p_2(s)\widetilde q_1(s)+p_1(s) \widetilde q_2(s)-p_4(s)\widetilde q_3(s)-p_3(s) \widetilde q_4(s)),
\end{align}
and when $\tau = 0$,
\begin{align}\label{eq:differentialH1}
&\sum_{k=1}^6 \left(p_k(s)q'_k(s) +\widetilde p_k(s)\widetilde q'_k(s) \right) - H(s) \nonumber\\
&=H(s) - \frac 13 \frac{\ud}{\ud s} \left(2s H(s) + p_1(s)q_1(s) + p_2(s)q_2(s) + \widetilde p_1(s)\widetilde q_1(s) + \widetilde p_2(s)\widetilde q_2(s) - 2p_5(s)q_5(s)\right.\nonumber\\
&\quad-2\widetilde p_5(s)\widetilde q_5(s) - p_6(s)q_6(s)-\widetilde p_6(s)\widetilde q_6(s)+2\frac{s_1}{r_1} p_5(s) + 2 \frac{s_2}{r_2} \widetilde p_5(s)).
\end{align}
We also have the following differential identity with respect to the parameter $\gamma$:
\begin{align}\label{eq:differential:gamma}
\frac{\partial}{\partial \gamma}\left(\sum_{k=1}^6 \left(p_k(s)q'_k(s) +\widetilde p_k(s)\widetilde q'_k(s) \right) - H(s)\right) = \frac{\ud}{\ud s}\sum_{k=1}^6\left(p_k(s)\frac{\partial}{\partial \gamma}q_k(s) + \widetilde p_k(s)\frac{\partial}{\partial \gamma}\widetilde q_k(s)\right).
\end{align}
\end{proposition}

\begin{proof}
The differential indentities \eqref{eq:differentialH} and \eqref{eq:differentialH1} follows directly from \eqref{def:H2}, \eqref{def:pq's} and cumbersome calculations. To see the differential identity with respect to the parameter $\gamma$, we have from \eqref{pq} that
\begin{align}
\frac{\partial}{\partial \gamma}H(s)& = \sum_{k=1}^6\left(\frac{\partial H}{\partial p_k} \frac{\partial}{\partial \gamma} p_k(s) + \frac{\partial H}{\partial q_k} \frac{\partial}{\partial \gamma}q_k(s) +\frac{\partial H}{\partial \widetilde p_k} \frac{\partial}{\partial \gamma} \widetilde p_k(s)+\frac{\partial H}{\partial \widetilde q_k} \frac{\partial}{\partial \gamma} \widetilde q_k(s)\right)\nonumber\\
&=\sum_{k=1}^6\left(q'_k(s) \frac{\partial}{\partial \gamma} p_k(s)  -p'_k(s) \frac{\partial}{\partial \gamma}q_k(s) +\widetilde q'_k(s) \frac{\partial}{\partial \gamma} \widetilde p_k(s)-\widetilde p'_k(s) \frac{\partial}{\partial \gamma} \widetilde q_k(s)\right),
\end{align}
which leads to \eqref{eq:differential:gamma}.

This completes the proof of Proposition \ref{prop:H}.
\end{proof}


\section{Asymptotic analysis of the RH problem for $X$ as $s \to +\infty$ with $\gamma=1$}\label{sec:AsyX1}

In this section, we will analyze the RH problem for $X$ as $s \to +\infty$ with $\gamma=1$. In this case, it is readily seen from \eqref{def:JX} that $X$ has no jump on the interval $(-s,s)$.
	
\subsection{First transformation: $X \to T$}
This transformation is a rescaling of the RH problem for $X$, which is defined by
\begin{equation}\label{def:PhiToT}
T(z)= \diag \left( s^{\frac14},s^{\frac14}, s^{-\frac14},s^{-\frac14} \right) X(sz).
\end{equation}
In view of RH problem \ref{rhp:X} for $X$, it is readily seen that $T$ satisfies the following RH problem.

\begin{rhp}\label{rhp:T}
\hfill
\begin{enumerate}
\item[\rm (a)] $T(z)$ is defined and analytic in $\mathbb{C}\setminus \{\Gamma_T \cup \{-1\} \cup \{1\} \}$, where
\begin{equation}\label{def:gammaT}
\Gamma_T:=\cup^5_{j=0}\Gamma_j^{(1)},
\end{equation}
and where the contours $\Gamma_j^{(1)}$, $j=0,1,\ldots,5$, are defined in \eqref{def:Gammajs} with $s=1$.

\item[\rm (b)] For $z\in \Gamma_T$, $T(z)$ satisfies the jump condition
\begin{equation}\label{eq:T-jump}
 T_+(z)=T_-(z)J_T(z),
\end{equation}
where
\begin{equation}\label{def:JT}
J_T(z):=\left\{
 \begin{array}{ll}
          \begin{pmatrix}0&0&1&0\\0&1&0&0\\-1&0&0&0\\0&0&0&1 \end{pmatrix}, & \qquad \hbox{$z\in \Gamma_0^{(1)}$,} \\
          I-E_{2,1}+E_{3,1}+E_{3,4},  & \qquad  \hbox{$z\in \Gamma_1^{(1)}$,} \\
          I-E_{1,2}+E_{4,2}+E_{4,3},  & \qquad \hbox{$z\in \Gamma_2^{(1)}$,} \\
          \begin{pmatrix}1&0&0&0\\0&0&0&1\\0&0&1&0\\0&-1&0&0 \end{pmatrix}, & \qquad  \hbox{$z\in \Gamma_3^{(1)}$,} \\
          I+E_{1,2}+E_{4,2}-E_{4,3}, & \qquad  \hbox{$z\in \Gamma_4^{(1)}$,} \\
          I+E_{2,1}+E_{3,1}-E_{3,4}, & \qquad  \hbox{$z\in \Gamma_5^{(1)}$.} \\
        \end{array}
      \right.
 \end{equation}

\item[\rm (c)]As $z \to \infty$ with $z\in \mathbb{C} \setminus \Gamma_T$, we have
\begin{align}\label{eq:asyT}
T(z)&=\left( I+\frac{T^{(1)}}{z}+ \Boh(z^{-2}) \right) \diag \left((-z)^{-\frac14},z^{-\frac14},(-z)^{\frac14},z^{\frac14} \right)
\nonumber  \\
& \quad \times A \diag \left(
e^{-s^{\frac32}\what\theta_1(z;s)+\tau s z}, e^{-s^{\frac32}\what\theta_2(z;s)- \tau s z}, e^{s^{\frac32}\what\theta_1(z;s)+\tau
s z},e^{s^{\frac32}\what\theta_2(z;s)- \tau s z} \right),
\end{align}
where $T^{(1)}$ is independent of $z$ and $A$ is defined in \eqref{def:A} and
\begin{align}\label{def:wtilthetai}
\what\theta_1(z)&= \frac23 r_1(-z)^{\frac 32} + \frac{2 s_1}{s} (-z)^{\frac 12}, \qquad z \in \mathbb{C} \setminus [0, \infty),\\
\what\theta_2(z)&=\frac23 r_2z^{\frac 32} +\frac{2 s_2}{s} z^{\frac 12}, \qquad z \in \mathbb{C} \setminus (-\infty, 0].
\end{align}

\item[\rm (d)]
As $z \to  \pm 1$, we have $T(z)=\Boh(\ln(z \mp 1))$.
\end{enumerate}
\end{rhp}

\subsection{Second transformation: $T \to S$}
In this transformation we partially normalize RH problem \ref{rhp:T} for $T$ at infinity. For this purpose, we introduce the following two $g$-functions:
\begin{align}
g_1(z)&=\frac23 r_1(1-z)^{\frac 32}+\left(-r_1+\frac{2s_1}{s}\right)(1-z)^{\frac12}, \qquad z\in \mathbb{C} \setminus [1,+\infty),
\label{def:g1} \\
g_2(z)&=\frac23 r_2(z+1)^{\frac 32}+\left(-r_2+\frac{2s_2}{s}\right)(z+1)^{\frac12}, \qquad z\in \mathbb{C} \setminus (-\infty,-1].
\label{def:g2}
\end{align}
As $z\to \infty$, it is readily seen that
\begin{align}
g_1(z)&=\what\theta_{1}(z;s)+\left(-\frac{r_1}{4}+\frac{s_1}{s}\right)(-z)^{-\frac12}+\Boh(z^{-\frac32}), \label{eq:asyg1}
\\
g_2(z)&=\what\theta_{2}(z;s)+\left(-\frac{r_2}{4}+\frac{s_2}{s}\right)z^{-\frac12}+\Boh(z^{-\frac32}), \label{eq:asyg2}
\end{align}
where $\what \theta_{i}(z;s)$, $i=1,2$, are defined in \eqref{def:wtilthetai}. The second transformation is set to be
\begin{align}\label{def:TtoS}
S(z)&=\left(I+s^{\frac32}\left(-\frac{r_1}{4}+\frac{s_1}{s}\right)\ii E_{3,1}-s^{\frac32}\left(-\frac{r_2}{4}+\frac{s_2}{s}\right)\ii E_{4,2}\right)T(z)
\nonumber
\\
&\quad \times \diag \left(
e^{s^{\frac 32}g_1(z)-\tau s z}, e^{s^{\frac 32}g_2(z)+ \tau s z}, e^{-s^{\frac 32}g_1(z)-\tau
s z}, e^{-s^{\frac 32}g_2(z) + \tau s z} \right).
\end{align}
Then, $S$ satisfies the following RH problem.

\begin{proposition}
The matrix-valued function $S$ defined in \eqref{def:TtoS} has the following properties:
\begin{enumerate}
\item[\rm (a)] $S(z)$ is defined and analytic in $\mathbb{C}\setminus \{\Gamma_T \cup \{-1\} \cup \{1\} \}$, where $\Gamma_T$ is defined in \eqref{def:gammaT}.

\item[\rm (b)] For $z\in \Gamma_T$, $S(z)$ satisfies the jump condition
\begin{equation}\label{eq:S-jump}
 S_+(z)=S_-(z)J_S(z),
\end{equation}
where
\begin{equation}\label{def:JS}
J_S(z):=\left\{
 \begin{array}{ll}
          \begin{pmatrix}0&0&1&0\\0&1&0&0\\-1&0&0&0\\0&0&0&1 \end{pmatrix}, & \qquad \hbox{$z\in \Gamma_0^{(1)}$,} \\
          I-e^{s^{\frac32}(g_1(z)-g_2(z))-2\tau s z}E_{2,1}+e^{2s^{\frac32}g_1(z)}E_{3,1}
          \\
          ~ +e^{s^{\frac32}(g_1(z)-g_2(z))+2\tau s z} E_{3,4},  & \qquad  \hbox{$z\in \Gamma_1^{(1)}$,} \\
          I- e^{-s^{\frac32}(g_1(z)-g_2(z))+2\tau s z} E_{1,2}+ e^{2s^{\frac32}g_2(z)} E_{4,2}
          \\
          ~ + e^{-s^{\frac32}(g_1(z)-g_2(z))-2\tau s z} E_{4,3},  & \qquad \hbox{$z\in \Gamma_2^{(1)}$,} \\
          \begin{pmatrix}1&0&0&0\\0&0&0&1\\0&0&1&0\\0&-1&0&0 \end{pmatrix}, & \qquad  \hbox{$z\in \Gamma_3^{(1)}$,} \\
          I+ e^{-s^{\frac32}(g_1(z)-g_2(z))+2\tau s z} E_{1,2}+ e^{2s^{\frac32}g_2(z)}E_{4,2}
          \\
          ~ - e^{-s^{\frac32}(g_1(z)-g_2(z))-2\tau s z}  E_{4,3}, & \qquad  \hbox{$z\in \Gamma_4^{(1)}$,} \\
          I+e^{s^{\frac32}(g_1(z)-g_2(z))-2\tau s z} E_{2,1}+e^{2s^{\frac32}g_1(z)} E_{3,1}
          \\
          ~ -e^{s^{\frac32}(g_1(z)-g_2(z))+2\tau s z}E_{3,4}, & \qquad  \hbox{$z\in \Gamma_5^{(1)}$.} \\
        \end{array}
      \right.
 \end{equation}

\item[\rm (c)]As $z \to \infty$ with $z\in \mathbb{C} \setminus \Gamma_T$, we have
\begin{align}\label{eq:asyS}
S(z)=\left( I+\frac{S^{(1)}}{z}+ \Boh(z^{-2}) \right) \diag \left((-z)^{-\frac14},z^{-\frac14},(-z)^{\frac14},z^{\frac14} \right)A,
\end{align}
where $S^{(1)}$ is independent of $z$ and $A$ is defined in \eqref{def:A}.

\item[\rm (d)]
As $z \to  \pm 1$, we have $S(z)=\Boh(\ln(z \mp 1))$.
\end{enumerate}
\end{proposition}
\begin{proof}
All the items follow directly from \eqref{def:TtoS} and RH problem \ref{rhp:T} for $T$. In particular, to check the jump condition of $S$ on $\mathbb{R}$, we need the facts that
\begin{align*}
g_{1,+}(x)+g_{1,-}(x)&=0, \qquad x \in [1,+\infty),
\\
g_{2,+}(x)+g_{2,-}(x)&=0, \qquad x \in (-\infty,1].
\end{align*}

To establish the large $z$ behavior of $S$ shown in item (c), we observe from \eqref{eq:asyT}, \eqref{eq:asyg1} and \eqref{eq:asyg2} that, as $z\to \infty$,
\begin{align}
& T(z)\diag \left(
e^{s^{\frac 32}g_1(z)-\tau s z}, e^{s^{\frac 32}g_2(z)+ \tau s z}, e^{-s^{\frac 32}g_1(z)-\tau
s z}, e^{-s^{\frac 32}g_2(z) + \tau s z} \right)
\nonumber
\\
& = \left( I+ \Boh(z^{-1}) \right) \diag \left((-z)^{-\frac14},z^{-\frac14},(-z)^{\frac14},z^{\frac14} \right)A
\nonumber
\\
& \quad \times \Big (I+ s^{\frac32}\left(-\frac{r_1}{4}+\frac{s_1}{s}\right)(-z)^{-\frac12}E_{1,1} + s^{\frac32}\left(-\frac{r_2}{4}+\frac{s_2}{s}\right)z^{\frac12}E_{2,2}
\nonumber
\\
&\quad \quad -s^{\frac32}\left(-\frac{r_1}{4}+\frac{s_1}{s}\right)(-z)^{-\frac12}E_{3,3} -s^{\frac32}\left(-\frac{r_2}{4}+\frac{s_2}{s}\right)z^{-\frac12}E_{4,4}+\Boh (z^{-\frac32}) \Big).
\label{eq:Texp}
\end{align}
By a direct calculation, it follows that
\begin{align}
&\diag \left((-z)^{-\frac14},z^{-\frac14},(-z)^{\frac14},z^{\frac14} \right)A \Big (I+ s^{\frac32}\left(-\frac{r_1}{4}+\frac{s_1}{s}\right)(-z)^{-\frac12}E_{1,1}
\nonumber
\\
& \quad + s^{\frac32}\left(-\frac{r_2}{4}+\frac{s_2}{s}\right)z^{-\frac12}E_{2,2}
-s^{\frac32}\left(-\frac{r_1}{4}+\frac{s_1}{s}\right)(-z)^{-\frac12}E_{3,3} -s^{\frac32}\left(-\frac{r_2}{4}+\frac{s_2}{s}\right)z^{-\frac12}E_{4,4} \Big)
\nonumber
\\
&=\left(I-s^{\frac32}\left(-\frac{r_1}{4}+\frac{s_1}{s}\right)\ii E_{3,1}+s^{\frac32}\left(-\frac{r_2}{4}+\frac{s_2}{s}\right)\ii E_{4,2}
+\Boh(z^{-1})\right)
\nonumber
\\
&\quad \times \diag \left((-z)^{-\frac14},z^{-\frac14},(-z)^{\frac14},z^{\frac14} \right)A.
\end{align}
This, together with \eqref{def:TtoS} and \eqref{eq:Texp}, gives us \eqref{eq:asyS}.
\end{proof}

\subsection{Global parametrix}
As $s\to +\infty$, it comes out that the jump matrix $J_S(z)$ of $S$ given in \eqref{def:JS} tends to the identity matrix exponentially fast except for $z\in \Gamma_0^{(1)} \cup \Gamma_3^{(1)}$. Indeed, by \eqref{def:g1} and \eqref{def:g2}, it is readily seen that as $z \to \infty$,
\begin{equation}\label{eq:largezg1}
\Re g_1(z) \sim \Re \left( \frac23 r_1(1-z)^{\frac 32} \right)\left\{
                                                                \begin{array}{ll}
                                                                  <0, & \hbox{$\arg (z-1) \in (0,\frac{2\pi}{3})\cup (\frac{4\pi}{3},2\pi)$,} \\
                                                                  >0, & \hbox{$\arg (z-1) \in (\frac{2\pi}{3},\frac{4\pi}{3})$,}
                                                                \end{array}
                                                              \right.
\end{equation}
and
\begin{equation}
\Re g_2(z) \sim \Re \left( \frac23 r_2(z+1)^{\frac 32} \right)\left\{
                                                                \begin{array}{ll}
                                                                  <0, & \hbox{$\arg (z+1) \in (\frac{\pi}{3},\pi)\cup (-\pi,-\frac{\pi}{3})$,} \\
                                                                  >0, & \hbox{$\arg (z+1) \in (-\frac{\pi}{3},\frac{\pi}{3})$,}
                                                                \end{array}
                                                              \right.
\end{equation}
where we have made use of the fact that $r_i > 0$, $i=1, 2$.

Moreover, for large positive $s$, we have
\begin{align}
g_1(z)-g_2(z) &= -\frac{\sqrt{2}}{3}r_2+o(1), \qquad z \to 1,
\\
g_1(z)-g_2(z) &= \frac{\sqrt{2}}{3}r_1+o(1), \qquad z \to -1. \label{eq:g1-g2}
\end{align}
A combination of \eqref{eq:largezg1}--\eqref{eq:g1-g2} and \eqref{def:JS} implies that, by deforming the contours if necessary, we may assume that $J_S(z) \to I$ as $s\to +\infty$ for $z\in \Gamma_T\setminus (\Gamma_0^{(1)} \cup \Gamma_3^{(1)})$. As a consequence, away from the points $\pm 1$, we expect that $S$ should be well approximated by the following global parametrix.
\begin{rhp}
\hfill
\begin{enumerate}
\item[\rm (a)] $N(z)$ is defined and analytic in $\mathbb{C}\setminus \{(-\infty,-1] \cup [1,\infty) \}$.

\item[\rm (b)] For $x\in (-\infty,-1) \cup (1,\infty)$, $N(x)$ satisfies the jump condition
\begin{equation}\label{eq:N-jump}
 N_+(x)=N_-(x)\left\{
                \begin{array}{ll}
                  \begin{pmatrix}0&0&1&0\\0&1&0&0\\-1&0&0&0\\0&0&0&1 \end{pmatrix}, & \hbox{$x>1$,} \\
                   \begin{pmatrix}1&0&0&0\\0&0&0&1\\0&0&1&0\\0&-1&0&0 \end{pmatrix}, & \hbox{$x<-1$.}
                \end{array}
              \right.
\end{equation}

\item[\rm (c)]As $z \to \infty$ with $z \in \mathbb{C} \setminus \mathbb{R}$,  we have
\begin{align}\label{eq:asyN}
N(z)=\left(  I+\frac{N^{(1)}}{z}+ \Boh(z^{-2}) \right) \diag \left((-z)^{-\frac14},z^{-\frac14},(-z)^{\frac14},z^{\frac14} \right)A,
\end{align}
where $N^{(1)}$ is independent of $z$ and $A$ is defined in \eqref{def:A}.

\end{enumerate}
\end{rhp}
The above RH problem can be solved explicitly, and its solution is given by
\begin{equation}\label{def:N}
N(z)=\diag \left((1-z)^{-\frac14},(z+1)^{-\frac14},(1-z)^{\frac14},(z+1)^{\frac14} \right)A,
\end{equation}
where we take the branch cuts of $(1-z)^{\frac14}$ and $(z+1)^{\frac14}$ along $[1,\infty)$ and $(-\infty,-1]$, respectively.

\subsection{Local parametrix near $z=-1$}

For $z$ near the endpoints $-1$ and $1$, $S(z)$ and $N(z)$ are not uniformly close to each other, hence, local parametrices need to be constructed near these endpoints. We start with the local parametrix $P^{(-1)}(z)$ near $-1$, which satisfies the following RH problem.

\begin{rhp}\label{rhp:P-1}
\hfill
\begin{itemize}
\item [\rm{(a)}] $P^{(-1)}(z)$ is defined and analytic in $D(-1, \varepsilon) \setminus \Gamma_T$, where $D(z_0, \varepsilon)$ and $\Gamma_T$ are defined in \eqref{def:dz0r} and \eqref{def:gammaT}, respectively.
\item [\rm{(b)}] For $z \in D(-1, \varepsilon) \cap \Gamma_T$, we have
\begin{equation}\label{eq:P-1-jump}
P^{(-1)}_+(z) = P^{(-1)}_-(z) J_{S}(z),
\end{equation}
where $J_{S}(z)$ is defined in \eqref{def:JS}.
\item [\rm{(c)}] As $s \to \infty$, $P^{(-1)}(z)$ satisfies the following matching condition
\begin{equation}\label{eq:P-1:matching}
P^{(-1)}(z) = (I+ \Boh(s^{-\frac{3}{2}})) N(z), \quad z \in \partial D(-1, \varepsilon),
\end{equation}
where $N(z)$ is given in \eqref{def:N}.
\end{itemize}
\end{rhp}

RH problem \ref{rhp:P-1} can be solved explicitly by using the Bessel parametrix $\Phi^{(\Bes)}(z)$ defined in Appendix \ref{sec:Bessel}. To do this, we introduce the function
\begin{align}\label{def:f}
f_{-1}(z): = g_2(z)^2 =\left(r_2-\frac{2s_2}{s}\right)^2(z+1) -\frac43 r_2\left(r_2-\frac{2s_2}{s}\right)(z+1)^{2} +\frac 49 r_2^2 (z+1)^3,
\end{align}
where $g_2$ is given in \eqref{def:g2}. Clearly, $f_{-1}(z)$ is analytic in $D(-1, \varepsilon)$ and is a conformal mapping for large positive $s$.

Let $\Omega_j^{(1)}$, $j=1,\ldots,6$, be the six regions shown in Figure \ref{fig:X} with $s=1$, we now define
\begin{align}\label{def:P-1}
P^{(-1)}(z) & = E_{-1}(z) \begin{pmatrix}
1 & 0 & 0 & 0\\
0 & \Phi^{(\Bes)}_{11}(s^3f_{-1}(z)) & 0 & \Phi^{(\Bes)}_{12}(s^3f_{-1}(z))\\
0 & 0 & 1 & 0\\
0 & \Phi^{(\Bes)}_{21}(s^3f_{-1}(z)) & 0 & \Phi^{(\Bes)}_{22}(s^3f_{-1}(z))
\end{pmatrix} \nonumber\\
&\quad\times\begin{pmatrix}
1 & 0 & 0 & 0\\
0 & e^{s^{\frac32}g_2(z)} & 0 & 0\\
0 & 0 & 1 & 0\\
0 & 0 & 0 & e^{-s^{\frac32}g_2(z)}
\end{pmatrix} \nonumber\\
&\quad\times \begin{cases}
I -e^{-s^{\frac32}(g_1(z)-g_2(z)) + 2 \tau s z} E_{1,2} +e^{-s^{\frac32}(g_1(z)-g_2(z)) - 2 \tau s z}  E_{4,3}, & z \in \Omega_2^{(1)}\cup\Omega_5^{(1)}\cap D(-1, \varepsilon),\\
I, & z \in \Omega_3^{(1)}\cup\Omega_4^{(1)}\cap D(-1, \varepsilon),
\end{cases}
\end{align}
where $f_{-1}(z)$ is defined in \eqref{def:f} and
\begin{equation}\label{def:E-1}
E_{-1}(z) = \frac{1}{\sqrt{2}} N(z) \begin{pmatrix}
\sqrt{2} & 0 & 0 & 0\\
0 & \pi^{\frac 12} s^{\frac 34} f_{-1}(z)^{\frac 14} & 0 & - \ii \pi^{-\frac 12} s^{-\frac 34} f_{-1}(z)^{-\frac 14}\\
0 & 0 & \sqrt{2} & 0\\
0 & -\ii \pi^{\frac 12} s^{\frac 34} f_{-1}(z)^{\frac 14} & 0 & \pi^{-\frac 12} s^{-\frac 34} f_{-1}(z)^{-\frac 14}
\end{pmatrix}.
\end{equation}

\begin{proposition}\label{pro:P-1}
The local parametrix $P^{(-1)}(z)$ defined in \eqref{def:P-1} solves RH problem \ref{rhp:P-1}.
\end{proposition}
\begin{proof}
First, we show the prefactor $E_{-1}(z)$ is analytic near $z=-1$. To achieve this, we notice that the only possible jump for $E_{-1}(z)$ is on the interval $(-1-\varepsilon, -1)$. For $z \in (-1-\varepsilon, -1)$, we see from \eqref{eq:N-jump} and \eqref{def:E-1} that
\begin{align}
&E_{-1,-}(z)^{-1}E_{-1,+}(z)\nonumber\\
 & = \frac{1}{2} \begin{pmatrix}
\sqrt{2} & 0 & 0 & 0\\
0 & \pi^{-\frac 12} s^{-\frac 34} f_{-1,-}(z)^{-\frac 14} & 0 & \ii \pi^{-\frac 12} s^{-\frac 34} f_{-1,-}(z)^{-\frac 14}\\
0 & 0 & \sqrt{2} & 0\\
0 & \ii \pi^{\frac 12} s^{\frac 34} f_{-1,-}(z)^{\frac 14} & 0 & \pi^{\frac 12} s^{\frac 34} f_{-1,-}(z)^{\frac 14}
\end{pmatrix}\begin{pmatrix}1&0&0&0\\0&0&0&1\\0&0&1&0\\0&-1&0&0 \end{pmatrix}\nonumber \\
& \quad\times\begin{pmatrix}
\sqrt{2} & 0 & 0 & 0\\
0 & \pi^{\frac 12} s^{\frac 34} f_{-1,+}(z)^{\frac 14} & 0 & - \ii \pi^{-\frac 12} s^{-\frac 34} f_{-1,+}(z)^{-\frac 14}\\
0 & 0 & \sqrt{2} & 0\\
0 & -\ii \pi^{\frac 12} s^{\frac 34} f_{-1,+}(z)^{\frac 14} & 0 & \pi^{-\frac 12} s^{-\frac 34} f_{-1,+}(z)^{-\frac 14}
\end{pmatrix}=I.
\end{align}
Moreover, as $z \to -1$, we have
\begin{equation}\label{def:localE-1}
E_{-1}(z) = E_{-1}(-1) + E_{-1}'(-1)(z+1) + \Boh\left((z+1)^2\right),
\end{equation}
where
\begin{equation}
E_{-1}(-1) = \begin{pmatrix}
2^{-\frac 34} & 0 & -\ii 2^{-\frac 34} & 0\\
0 & \pi^{\frac 12} s^{\frac 34} \left(r_2 - \frac{2s_2}{s}\right)^{\frac 12} & 0 & 0\\
-\ii 2^{-\frac 14} & 0 & 2^{-\frac 14} & 0\\
0 & 0 & 0 & \pi^{-\frac 12} s^{-\frac 34} \left(r_2 - \frac{2s_2}{s}\right)^{-\frac 12}
\end{pmatrix}
\end{equation}
and
\begin{equation}
E_{-1}'(-1) = \begin{pmatrix}
\frac{1}{8 \cdot 2^{3/4}} & 0 & -\frac{\ii}{8 \cdot 2^{3/4}} & 0\\
0 & -\frac{\pi^{1/2} s^{3/4}r_2}{3} \left(r_2 - \frac{2s_2}{s}\right)^{-\frac 12} & 0 & 0\\
\frac{\ii}{8 \cdot 2^{1/4}} & 0 & -\frac{1}{8 \cdot 2^{1/4}} & 0\\
0 & 0 & 0 & \frac{\pi^{-1/2} s^{-3/4} r_2}{3} \left(r_2 - \frac{2s_2}{s}\right)^{-\frac 32}
\end{pmatrix}.
\end{equation}
Therefore, $E_{-1}(z)$ is indeed analytic in $D(-1, \varepsilon)$. It is then straightforward to verify the jump condition of $P^{(-1)}$ in \eqref{eq:P-1-jump} by using the analyticity of $E_{-1}$ and \eqref{eq:jump:Bessel}.

Next, we check the matching condition \eqref{eq:P-1:matching}. From the definitions of $g_1(z)$ and $g_2(z)$ in \eqref{def:g1} and \eqref{def:g2}, it is clear that functions $e^{-s^{3/2}(g_1(z)-g_2(z)) + 2 \tau s z}$ and $e^{-s^{3/2}(g_1(z)-g_2(z)) - 2 \tau s z}$ in \eqref{def:P1} are exponentially small as $s \to +\infty$ for $z \in D(-1, \varepsilon)$; cf. \eqref{eq:g1-g2}. Thus, it follows from the asymptotic behavior of the Bessel parametrix at infinity in \eqref{eq:infty:Bessel} that, as $s \to +\infty$,
\begin{equation}
P^{(-1)}(z) N(z)^{-1} = I + \frac{J^{(-1)}_1(z)}{s^{3/2}} + \Boh(s^{-3}), \quad z \in \partial D(-1, \varepsilon),
\end{equation}
where
\begin{equation}\label{def:J-1}
J^{(-1)}_1(z) = \frac{1}{8f_{-1}(z)^{1/2}} N(z) \begin{pmatrix}
0 & 0 & 0 & 0\\
0 & -1 & 0 & -2\ii\\
0 & 0 & 0 & 0\\
0 & -2\ii & 0 & 1
\end{pmatrix}N(z)^{-1}.
\end{equation}

This completes the proof of Proposition \ref{pro:P-1}.
\end{proof}
For later use, we include following local behavior of $J^{(-1)}_1(z)$ near $z = -1$:
\begin{align}\label{J1--1}
J^{(-1)}_1(z)&  =  -\frac{\ii}{8\left(r_2 - 2s_2/s\right)(z+1)} E_{2,4} - \frac{\ii r_2}{12\left(r_2 - 2s_2/s\right)^2} E_{2,4} -\frac{3\ii}{8\left(r_2 - 2s_2/s\right)} E_{4,2}
\nonumber\\
&\quad-\left(\frac{\ii r_2^2}{18\left(r_2 -2s_2/s\right)^3} E_{2,4} +\frac{\ii r_2}{4\left(r_2 - 2s_2/s\right)^2}E_{4,2}\right)(z+1) + \Boh\left((z+1)^2\right), \quad z \to -1.
\end{align}

\subsection{Local parametrix near $z=1$}
In a small disc $D(1, \varepsilon)$ around $z=1$, the local parametrix $P^{(1)}(z)$ reads as follows.
\begin{rhp}\label{rhp:P1}
\hfill
\begin{itemize}
\item [\rm{(a)}] $P^{(1)}(z)$ is defined and analytic in $D(1, \varepsilon) \setminus \Gamma_T$, where $\Gamma_T$ is defined in \eqref{def:gammaT}.
\item [\rm{(b)}] For $z \in D(1, \varepsilon) \cap \Gamma_T$, we have
\begin{equation}\label{eq:P1-jump}
P^{(1)}_+(z) = P^{(1)}_-(z) J_{S}(z),
\end{equation}
where $J_{S}(z)$ is defined in \eqref{def:JS}.
\item [\rm{(c)}] As $s \to \infty$, $P^{(1)}(z)$ satisfies the following matching condition
\begin{equation}\label{eq:P1:matching}
P^{(1)}(z) = (I+ \Boh(s^{-\frac{3}{2}})) N(z), \quad z \in \partial D(1, \varepsilon),
\end{equation}
where $N(z)$ is given in \eqref{def:N}.
\end{itemize}
\end{rhp}

Similar to the construction of $P^{(-1)}$, the above RH problem can be solved again by using
the Bessel parametrix $\Phi^{(\Bes)}(z)$ defined in Appendix \ref{sec:Bessel}. In this case, we need the following function
\begin{equation}\label{def:tildef}
f_1(z) := g_1(z)^2=\left(r_1 - \frac{2s_1}{s}\right)^2 (1-z) -\frac43 r_1\left(r_1-\frac{2s_1}{s}\right)(1-z)^{2} +\frac 49 r_1^2 (1-z)^3,
\end{equation}
where $g_1(z)$ is defined in \eqref{def:g1}. Clearly, $f_1(z)$ is analytic in $D(1, \varepsilon)$ and is a conformal mapping for large positive $s$. We now define
\begin{align}\label{def:P1}
P^{(1)}(z) &= E_1(z) \begin{pmatrix}
\Phi^{(\Bes)}_{11}(s^3f_1(z)) & 0 & -\Phi^{(\Bes)}_{12}(s^3f_1(z)) & 0\\
0 & 1 & 0 & 0\\
-\Phi^{(\Bes)}_{21}(s^3f_1(z)) & 0 & \Phi^{(\Bes)}_{22}(s^3f_1(z)) & 0\\
0 & 0 & 0 & 1
\end{pmatrix}\nonumber\\
&\quad \times \begin{pmatrix}
e^{s^{\frac32}g_1(z)} & 0 & 0 & 0\\
0 & 1 & 0 & 0\\
0 & 0 & e^{-s^{\frac32}g_1(z)} & 0\\
0 & 0 & 0 & 1
\end{pmatrix}\nonumber\\
&\quad \times \begin{cases}
I -e^{s^{\frac32}(g_1(z)-g_2(z)) - 2 \tau s z} E_{2,1}+e^{s^{\frac32}(g_1(z)-g_2(z)) + 2 \tau s z}E_{3,4}, & z \in \Omega_2^{(1)}\cup\Omega_5^{(1)}\cap D(1, \varepsilon),\\
I, & z \in \Omega_1^{(1)}\cup\Omega_6^{(1)}\cap D(1, \varepsilon),
\end{cases}
\end{align}
where $f_1(z)$ is defined in \eqref{def:tildef} and
\begin{equation}\label{def:E1}
E_1(z) = \frac{1}{\sqrt{2}} N(z) \begin{pmatrix}
\pi^{\frac 12} s^{\frac 34} f_1(z)^{\frac 14} & 0 &  \ii \pi^{-\frac 12} s^{-\frac 34} f_1(z)^{-\frac 14} & 0\\
0 & \sqrt{2} & 0 & 0\\
 \ii \pi^{\frac 12} s^{\frac 34} f_1(z)^{\frac 14} & 0 & \pi^{-\frac 12} s^{-\frac 34} f_1(z)^{-\frac 14} & 0\\
 0 & 0 & 0 & \sqrt{2}
\end{pmatrix}.
\end{equation}
\begin{proposition}\label{pro:P1}
$P^{(1)}(z)$ defined in \eqref{def:P1} solves RH problem \ref{rhp:P1}.
\end{proposition}
\begin{proof}
First, we show $E_1(z)$ is analytic in $D(1, \varepsilon)$. From \eqref{def:E1}, the only possible jump is on $(1, 1 + \varepsilon)$, and for $z \in (1, 1 + \varepsilon)$,
\begin{align}
E_{1,-}(z)^{-1}E_{1,+}(z)&=\frac 12  \begin{pmatrix}
\pi^{-\frac 12} s^{-\frac 34} f_{1,-}(z)^{-\frac 14} & 0 &  -\ii \pi^{-\frac 12} s^{-\frac 34} f_{1,-}(z)^{-\frac 14} & 0\\
0 & \sqrt{2} & 0 & 0\\
 -\ii \pi^{\frac 12} s^{\frac 34} f_{1,-}(z)^{\frac 14} & 0 & \pi^{\frac 12} s^{\frac 34} f_{1,-}(z)^{\frac 14} & 0\\
 0 & 0 & 0 & \sqrt{2}
\end{pmatrix}
\begin{pmatrix}0&0&1&0\\0&1&0&0\\-1&0&0&0\\0&0&0&1 \end{pmatrix}\nonumber\\
&\quad\times \begin{pmatrix}
\pi^{\frac 12} s^{\frac 34} f_{1,+}(z)^{\frac 14} & 0 &  \ii \pi^{-\frac 12} s^{-\frac 34} f_{1,+}(z)^{-\frac 14} & 0\\
0 & \sqrt{2} & 0 & 0\\
 \ii \pi^{\frac 12} s^{\frac 34} f_{1,+}(z)^{\frac 14} & 0 & \pi^{-\frac 12} s^{-\frac 34} f_{1,+}(z)^{-\frac 14} & 0\\
 0 & 0 & 0 & \sqrt{2}
\end{pmatrix}=I
\end{align}
Moreover, as $z \to 1$, we have
\begin{equation}\label{def:localE1}
E_1(z) = E_1(1) + E_1'(1)(z-1) + \Boh\left((z-1)^2\right),
\end{equation}
where
\begin{equation}
E_1(1) = \begin{pmatrix}
\pi^{\frac 12} s^{\frac 34} \left(r_1 - \frac{2s_1}{s}\right)^{\frac 12} & 0 & 0 & 0\\
0 & 2^{-\frac 34} & 0 & \ii 2^{-\frac 34}\\
0 & 0 & \pi^{-\frac 12} s^{-\frac 34} \left(r_1 - \frac{2s_1}{s}\right)^{-\frac 12} & 0\\
0 & \ii 2^{-\frac 14} & 0 & 2^{-\frac 14}
\end{pmatrix}
\end{equation}
and
\begin{equation}
E_1'(1) = \begin{pmatrix}
\frac{\pi^{1/2} s^{3/4}r_1}{3} \left(r_1 - \frac{2s_1}{s}\right)^{-\frac 12} & 0 & 0 & 0\\
0 & -\frac{1}{8 \cdot 2^{3/4}} & 0 & -\frac{\ii}{8 \cdot 2^{3/4}}\\
0 & 0 & \frac{\pi^{-1/2} s^{-3/4} r_1}{3} \left(r_1 - \frac{2s_1}{s}\right)^{-\frac 32} & 0\\
0 & \frac{\ii}{8 \cdot 2^{1/4}} & 0 & \frac{1}{8 \cdot 2^{1/4}}
\end{pmatrix}.
\end{equation}
Therefore, $E_1(z)$ is indeed analytic in $D(1, \varepsilon)$. The jump condition of $P^{(1)}(z)$ in \eqref{eq:P1-jump} can be verified from the analyticity of $E_1(z)$ and the jump condition in \eqref{eq:jump:Bessel}.

Finally, we check the matching condition in \eqref{eq:P1:matching}, it follows from the asymptotic behavior of the Bessel parametrix at infinity in \eqref{eq:infty:Bessel} that, as $s \to +\infty$,
\begin{equation}
P^{(1)}(z) N(z)^{-1} = I + \frac{J^{(1)}_1(z)}{s^{3/2}} + \Boh(s^{-3}),
\end{equation}
where
\begin{equation}\label{def:J1}
J^{(1)}_1(z) = \frac{1}{8f_1(z)^{1/2}} N(z) \begin{pmatrix}
-1 & 0 & 2\ii & 0\\
0 & 0 & 0 & 0\\
2\ii & 0 & 1 & 0\\
0 & 0 & 0 & 0
\end{pmatrix}N(z)^{-1}.
\end{equation}

This completes the proof of Proposition \ref{pro:P1}.
\end{proof}
For later use, we calculate the behavior of $J^{(1)}_1(z)$ near $z = 1$ as follows,
\begin{align}\label{J1-1}
J^{(1)}_1(z) &=  \frac{\ii}{8\left(r_1 - 2s_1/s\right)(1-z)} E_{1,3} + \frac{\ii r_1}{12\left(r_1 - 2s_1/s\right)^2}E_{1,3}+\frac{3\ii}{8\left(r_1 - 2s_1/s\right)}E_{3,1}\nonumber\\
&\quad+\left(\frac{\ii r_1^2}{18\left(r_1 - 2s_1/s\right)^3}E_{1,3}+\frac{\ii r_1}{4\left(r_1 - 2s_1/s\right)^2}E_{3,1}\right)(1-z) + \Boh\left((1-z)^2\right), \quad z \to 1.
\end{align}

\subsection{Final transformation}
The final transformation is defined by
\begin{equation}\label{def:R}
R(z) = \begin{cases}
S(z) P^{(-1)}(z)^{-1}, & \quad z \in D(-1, \varepsilon),\\
S(z) P^{(1)}(z)^{-1}, & \quad z \in D(1, \varepsilon),\\
S(z) N(z)^{-1}, & \quad \textrm{elsewhere}.
\end{cases}
\end{equation}
From the RH problems for $S$, $N$ and $P^{(\pm1)}$, it follows that $R$ satisfies the following RH problem.
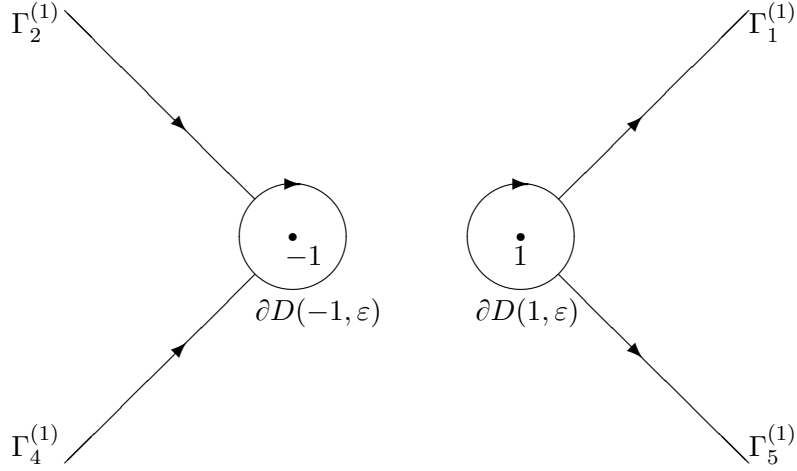
\begin{figure}[t]
\begin{center}
   \setlength{\unitlength}{1truemm}
   \begin{picture}(100,70)(-13,2)


       \put(20,35){\line(-1,-1){25}}
       \put(20,45){\line(-1,1){25}}

       \put(60,45){\line(1,1){25}}
       \put(60,35){\line(1,-1){25}}


       \put(10,55){\thicklines\vector(1,-1){1}}
       \put(10,25){\thicklines\vector(1,1){1}}
       \put(70,25){\thicklines\vector(1,-1){1}}
       \put(70,55){\thicklines\vector(1,1){1}}
        \put(25,47){\thicklines\vector(1,0){1}}
         \put(55,47){\thicklines\vector(1,0){1}}

       \put(-12,11){$\Gamma_4^{(1)}$}
       \put(-12,67){$\Gamma_2^{(1)}$}
       \put(85,11){$\Gamma_5^{(1)}$}
       \put(85,67){$\Gamma_1^{(1)}$}


       \put(25,40){\thicklines\circle*{1}}
       \put(55,40){\thicklines\circle*{1}}

	\put(25,40){\circle{20}}
	 \put(55,40){\circle{20}}
	
       \put(24,36.3){$-1$}
       \put(54,36.3){$1$}
        \put(20,29){$\partial D(-1, \varepsilon)$}
        \put(49,29){$\partial D(1, \varepsilon)$}
   \end{picture}
   \caption{The contour $\Sigma_{R}$ of the RH problem for $R$.}
   \label{fig:R}
\end{center}
\end{figure}
\begin{rhp}
\hfill
\begin{itemize}
\item [\rm{(a)}] $R(z)$ is defined and analytic in $\mathbb{C} \setminus \Gamma_{R}$, where
\begin{equation}
\Sigma_R:=\Gamma_T \cup \partial D(-1,\varepsilon) \cup \partial D(1,\varepsilon) \setminus \{\mathbb{R}
\cup D(-1,\varepsilon) \cup D(1,\varepsilon) \};
\end{equation}
see Figure\ref{fig:R} for an illustration.
\item [\rm{(b)}] For $z \in \Gamma_{R}$, we have
\begin{equation}\label{eq:Rjump}
R_+(z) = R_-(z) J_R (z),
\end{equation}
 where
\begin{equation}
J_R(z) = \begin{cases}
P^{(-1)}(z) N(z)^{-1}, & \quad z \in \partial D(-1, \varepsilon),\\
P^{(1)}(z) N(z)^{-1}, & \quad z \in \partial D(1, \varepsilon),\\
N(z) J_S(z) N(z)^{-1}, & \quad z \in \Gamma_{R} \setminus \partial D(\pm1, \varepsilon),
\end{cases}
\end{equation}
with $J_S(z)$ defined in \eqref{def:JS}.
\item [\rm{(c)}] As $z \to \infty$, we have
\begin{equation}\label{eq:asyR}
R(z) = I + \frac{R^{(1)}}{z} + \Boh (z^{-1}),
\end{equation}
where $R^{(1)}$ is independent of $z$.
\end{itemize}
\end{rhp}
Since the jump matrix $J_S(z)$ of $S$ given in \eqref{def:JS} tends to the identity matrix exponentially fast except for $z\in \Gamma_0^{(1)} \cup \Gamma_3^{(1)}$ as $s \to +\infty$, from the matching condition \eqref{eq:P-1:matching} and \eqref{eq:P1:matching}, we have
\begin{equation}
J_R(z) = I + \Boh (s^{-\frac 32}), \qquad s \to +\infty.
\end{equation}
By a standard argument \cite{Deift1999, Deift1993}, we conclude that
\begin{equation}\label{def:RR}
R(z) = I + \frac{R_1(z)}{s^{3/2}} + \Boh (s^{-3}), \quad s \to +\infty,
\end{equation}
uniformly for $z\in \mathbb{C}\setminus \Gamma_{R}$. Moreover, inserting the above expansion into \eqref{eq:Rjump}, it follows that function $R_1$ is analytic in $\mathbb{C} \setminus (\partial D(-1, \varepsilon) \cup \partial D(1, \varepsilon))$ with asymptotic behavior $\Boh(1/z)$ as $z\to \infty$, and satisfies
$$
R_{1,+}(z)-R_{1,-}(z)= \left\{
                         \begin{array}{ll}
                           J^{(-1)}_1(z), & \hbox{$z\in \partial D(-1,\varepsilon)$,} \\
                           J^{(1)}_1(z), & \hbox{$z\in \partial D(1,\varepsilon)$,}
                         \end{array}
                       \right.
$$
where the functions $J^{(-1)}_1(z)$ and $J^{(1)}_1(z)$ are given in \eqref{def:J-1} and \eqref{def:J1}, respectively.
By Cauchy's residue theorem, we have
\begin{align}
R_1(z) &= \frac{1}{2 \pi \ii} \oint_{\partial D(-1, \varepsilon)} \frac{J_1^{(-1)} (\zeta)}{z-\zeta} \ud \zeta + \frac{1}{2 \pi \ii} \oint_{\partial D(1, \varepsilon)} \frac{J_1^{(1)} (\zeta)}{z-\zeta} \ud \zeta \nonumber\\
&=\begin{cases}
\frac{\Res_{\zeta = -1} J^{(-1)}_1(\zeta)}{z+1} + \frac{\Res_{\zeta = 1} J^{(1)}_1(\zeta)}{z-1}, &\quad z \in \mathbb{C} \setminus \{D(-1, \varepsilon) \cup D(1, \varepsilon)\},\\
\frac{\Res_{\zeta = -1} J^{(-1)}_1(\zeta)}{z+1} + \frac{\Res_{\zeta = 1} J^{(1)}_1(\zeta)}{z-1} - J^{(-1)}_1(z), &\quad z \in D(-1, \varepsilon),\\
\frac{\Res_{\zeta = -1} J^{(-1)}_1(\zeta)}{z+1} + \frac{\Res_{\zeta = 1} J^{(1)}_1(\zeta)}{z-1} - J^{(1)}_1(z), &\quad z \in D(1, \varepsilon).
\end{cases}
\end{align}
In view of \eqref{J1--1} and \eqref{J1-1}, it follows from direct calculations that
\begin{equation}\label{eq:R1'-1}
R_1'(-1) = \frac{\ii}{32\left(r_1 - 2s_1/s \right)} E_{1,3} + \frac{\ii r_2^2}{18\left(r_2 - 2s_2/s\right)^3}E_{2,4} + \frac{\ii r_2}{4\left(r_2 - 2s_2/s\right)^2}E_{4,2}
\end{equation}
and
\begin{equation}\label{eq:R1'1}
R_1'(1) = \frac{\ii r_1^2}{18\left(r_1 - 2s_1/s\right)^3} E_{1,3} + \frac{\ii}{32\left(r_2 - 2s_2/s\right)}E_{2,4} + \frac{\ii r_1}{4\left(r_1 - 2s_1/s\right)^2}E_{3,1}.
\end{equation}

\section{Asymptotic analysis of the RH problem for $X$ as $s \to +\infty$ with $0<\gamma<1$}
\label{sec:AsyXgamma}

If $0< \gamma <1$, the matrix-valued function $X$ has a non-trivial jump
$\begin{pmatrix}
1&0&1-\gamma&1-\gamma\\0&1&1-\gamma&1-\gamma\\0&0&1&0\\0&0&0&1 \end{pmatrix}$ over the interval $(-s,s)$. This extra jump will lead to a completely different asymptotic analysis of the RH problem for $X$ as $s \to +\infty$, which will be carried out in this section.

\subsection{First transformation: $X \mapsto \what T$}
This transformation is a rescaling and normalization of the RH problem for $X$, and it is defined by
\begin{align}\label{def:XToTgamma}
\what T(z) &= \diag \left( s^{\frac14},s^{\frac14}, s^{-\frac14},s^{-\frac14} \right) X(sz)
\nonumber
\\
& \quad \times \diag \left(
e^{\theta_1(sz)-\tau s z}, e^{\theta_2(sz)+ \tau s z}, e^{-\theta_1(sz)-\tau
s z},e^{-\theta_2(sz) + \tau s z} \right),
\end{align}
where the functions $\theta_1$ and $\theta_2$ are given in \eqref{def:theta1} and \eqref{def:theta2}, respectively. In view of the facts that
\begin{equation}\label{eq:thetairelations}
\begin{aligned}
\theta_{1,+}(sx)+\theta_{1,-}(sx)& = 0, \qquad x>0,
\\
\theta_{2,+}(sx)+\theta_{2,-}(sx)&=0, \qquad x<0,
\end{aligned}
\end{equation}
and RH problem \ref{rhp:X} for $X$, it is readily seen that $\what T$ defined in \eqref{def:XToTgamma} satisfies the following RH problem.
\begin{rhp}\label{rhp:hatT}
\hfill
\begin{enumerate}
\item[\rm (a)] $\what T(z)$ is defined and analytic in $\mathbb{C} \setminus \Gamma_{\what T}$, where
\begin{equation}\label{def:gammahatT}
\Gamma_{\what T}:=\cup^5_{j=0}\Gamma_j^{(1)}\cup [-1,1],
\end{equation}
and where the contours $\Gamma_j^{(1)}$, $j=0,1,\ldots,5$, are defined in \eqref{def:Gammajs} with $s=1$.

\item[\rm (b)] For $z\in \Gamma_{\what T}$, $\what T(z)$ satisfies the jump condition
\begin{equation}\label{eq:hatT-jump}
 \what T_+(z)=\what T_-(z)J_{\what T}(z),
\end{equation}
where
\begin{equation}\label{def:JhatT}
J_{\what T}(z):=\left\{
 \begin{array}{ll}
          \begin{pmatrix}0&0&1&0\\0&1&0&0\\-1&0&0&0\\0&0&0&1 \end{pmatrix}, & \qquad \hbox{$z\in \Gamma_0^{(1)}$,} \\
          I-e^{\theta_1(sz)-\theta_2(sz)-2\tau s z}E_{2,1}+e^{2\theta_1(sz)}E_{3,1}
          \\
          ~ +e^{\theta_1(sz)-\theta_2(sz)+2\tau s z} E_{3,4},  & \qquad  \hbox{$z\in \Gamma_1^{(1)}$,} \\
          I- e^{-\theta_1(sz)+\theta_2(sz)+2\tau s z} E_{1,2}+ e^{2\theta_2(sz)} E_{4,2}
          \\
          ~ + e^{-\theta_1(sz)+\theta_2(sz)-2\tau s z} E_{4,3},  & \qquad \hbox{$z\in \Gamma_2^{(1)}$,} \\
          \begin{pmatrix}1&0&0&0\\0&0&0&1\\0&0&1&0\\0&-1&0&0 \end{pmatrix}, & \qquad  \hbox{$z\in \Gamma_3^{(1)}$,} \\
          I+ e^{-\theta_1(sz)+\theta_2(sz)+2\tau s z} E_{1,2}+ e^{2\theta_2(sz)}E_{4,2}
          \\
          ~ - e^{-\theta_1(sz)+\theta_2(sz)-2\tau s z}  E_{4,3}, & \qquad  \hbox{$z\in \Gamma_4^{(1)}$,} \\
          I+e^{\theta_1(sz)-\theta_2(sz)-2\tau s z} E_{2,1}+e^{2\theta_1(sz)} E_{3,1}
          \\
          ~ -e^{\theta_1(sz)-\theta_2(sz)+2\tau s z}E_{3,4}, & \qquad  \hbox{$z\in \Gamma_5^{(1)}$,} \\
         \msf J_{\msf L}(z), & \qquad  \hbox{$z \in (-1,0)$,}
          \\
         \msf J_{\msf R}(z), & \qquad  \hbox{$z \in (0,1)$,}
        \end{array}
      \right.
 \end{equation}
 with
 \begin{align}\label{def:msfJL}
  &\msf J_{\msf L}(z)
  \nonumber
  \\
  &=\begin{pmatrix}
          1&0&(1-\gamma)e^{-2\theta_1(sz)}&(1-\gamma)e^{-\theta_{1}(sz)-\theta_{2,+}(sz)+2\tau sz}
          \\0&e^{\theta_{2,+}(sz)-\theta_{2,-}(sz)}&(1-\gamma)e^{-\theta_1(sz)-\theta_{2,-}(sz)-2\tau sz}&1-\gamma
          \\0&0&1&0
          \\0&0&0&e^{\theta_{2,-}(sz)-\theta_{2,+}(sz)}
  \end{pmatrix},
 \end{align}
  and
\begin{align}\label{def:msfJR}
  &\msf J_{\msf R}(z)
  \nonumber
  \\
  &=\begin{pmatrix}e^{\theta_{1,+}(sz)-\theta_{1,-}(sz)}&0&1-\gamma&(1-\gamma)e^{-\theta_{1,-}(sz)-\theta_2(sz)+2\tau sz}
          \\0&1&(1-\gamma)e^{-\theta_{1,+}(sz)-\theta_2(sz)-2\tau sz}&(1-\gamma)e^{-2\theta_2(sz)}\\0&0&e^{\theta_{1,-}(sz)-\theta_{1,+}(sz)}&0\\0&0&0&1
          \end{pmatrix}.
 \end{align}

\item[\rm (c)]As $z \to \infty$ with $z\in \mathbb{C} \setminus \Gamma_{\what T}$, we have
\begin{align}\label{eq:asyhatT}
\what T(z)=\left( I+ \frac{\what T^{(1)}}{z} + \Boh(z^{-2}) \right) \diag \left((-z)^{-\frac14},z^{-\frac14},(-z)^{\frac14},z^{\frac14} \right)A,
\end{align}
where $\what T^{(1)}$ is independent of $z$ and $A$ is defined in \eqref{def:A}.

\item[\rm (d)]
As $z \to \pm 1$, we have $\what T(z)=\Boh(\ln(z \mp 1))$.
\end{enumerate}
\end{rhp}

\subsection{Second transformation: $\what T \mapsto \what S$}
On account of the definitions of $\theta_1(z)$ and $\theta_2(z)$ given in \eqref{def:theta1} and \eqref{def:theta2}, it is readily seen that $J_{\what T}(z)$ in \eqref{def:JhatT} tends to $I$ exponentially fast as $s\to +\infty$ for $z\in \Gamma_{\what T} \setminus \mathbb{R}$. Moreover, the $(2,2)$, $(4,4)$ entries of $\msf J_L$ in \eqref{def:msfJL} and the $(1,1)$, $(3,3)$ entries of $\msf J_R$ in \eqref{def:msfJR} are highly oscillatory for large positive $s$.

The second transformation then involves the so-called lens opening around the interval $(-1,0)\cup (0,1)$. The idea is to remove the highly oscillatory terms of $J_{\what T}$ with the cost of creating extra jumps that tend to the identity matrices on some new contours. To proceed, we observe from \eqref{def:msfJL}, \eqref{def:msfJR} and \eqref{eq:thetairelations} that
\begin{equation}
\msf J_{\msf L}(z)=\msf J_{1,-}(z)
\begin{pmatrix}
1&0&0&0
\\
0&0&0&1-\gamma
\\
0&0&1&0
\\0&\frac{1}{\gamma-1}&0&0 \end{pmatrix}\msf J_{2,+}(z), \qquad z\in(-1,0),
\end{equation}
where
\begin{align}
\msf J_{1}(z)&=
I+e^{-\theta_1(sz)+\theta_2(sz)+2\tau s z}E_{1,2}+\frac{e^{2\theta_2(sz)}}{1-\gamma}E_{4,2}-e^{-\theta_1(sz)+\theta_2(sz)-2 \tau sz}E_{4,3},
\label{def:msfJ1} \\
\msf J_2(z)&
= I - e^{-\theta_1(sz) + \theta_2(sz) + 2\tau s z}E_{1,2} + \frac{e^{2\theta_2(sz)}}{1-\gamma}E_{4,2}
+ e^{-\theta_1(sz)+\theta_2(sz)-2 \tau sz}E_{4,3},
\label{def:msfJ2}
\end{align}
and
\begin{equation}
\msf J_{\msf R}(z)=\msf J_{3,-}(z)
\begin{pmatrix}
0&0&1-\gamma&0
\\
0&1&0&0
\\
\frac{1}{\gamma-1}&0&0&0
\\0&0&0&1
 \end{pmatrix}\msf J_{4,+}(z), \qquad z\in(0,1),
\end{equation}
where
\begin{align}
\msf J_3(z)&=
I+e^{\theta_1(sz)-\theta_2(sz)-2\tau s z}E_{2,1}+\frac{e^{2\theta_1(sz)}}{1-\gamma}E_{3,1}-e^{\theta_1(sz)-\theta_2(sz)+2\tau sz}E_{3,4},
\label{def:msfJ3}
\\
\msf J_4(z)&= I- e^{\theta_1(sz)-\theta_2(sz)-2\tau s z}E_{2,1}+\frac{e^{2\theta_1(sz)}}{1-\gamma}E_{3,1}+e^{\theta_1(sz)-\theta_2(sz)+2\tau sz}E_{3,4}.
\label{def:msfJ4}
\end{align}
It is easily seen that, for $i=1,\ldots,4$, we have
\begin{equation}\label{eq:inverseJi}
\msf J_i(z)^{-1}= 2I - \msf J_i(z).
\end{equation}

Let $\Omega_{\msf L, \pm}$ and $\Omega_{\msf R, \pm}$ be the lenses on the $\pm$-side of $(-1,0)$ and $(0,1)$, as shown in Figure \ref{fig:lenses}. The second transformation is defined by
\begin{equation}\label{def:hatS}
\what S = \what T\left\{
                   \begin{array}{ll}
                     \msf J_2(z)^{-1}, & \hbox{$z\in \Omega_{\msf L,+}$,}
                     \\
                     \msf J_1(z), & \hbox{$z\in \Omega_{\msf L,-}$,}
                     \\
                     \msf J_4(z)^{-1}, & \hbox{$z\in \Omega_{\msf R,+}$,}
                     \\
                     \msf J_3(z), & \hbox{$z\in \Omega_{\msf R,-}$,}
                     \\
                     I, & \hbox{elsewhere.}
                   \end{array}
                 \right.
\end{equation}
It is then readily seen from RH problem \ref{rhp:hatT} for $\what T$ that $\what S$ satisfies the following RH problem.

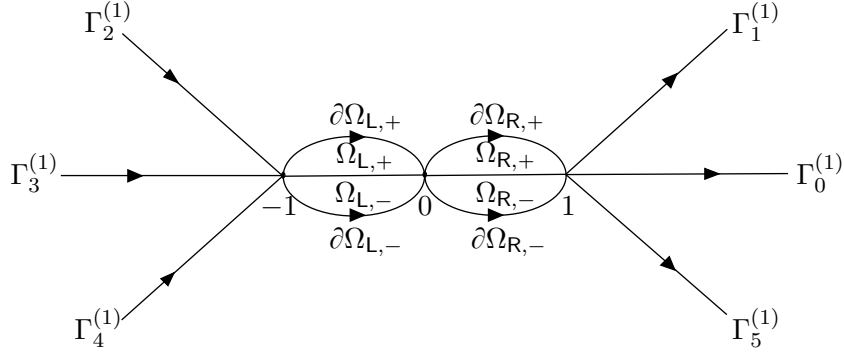
\begin{figure}[t]
\center

\tikzset{every picture/.style={line width=0.5pt}} 

\begin{tikzpicture}[x=0.75pt,y=0.75pt,yscale=-1,xscale=1]


\draw    (139.53,161.2) -- (250.53,161.2) ;
\draw    (170.2,89.76) -- (250.65,161.7) ;
\draw    (169.86,233.83) -- (250.53,161.2) ;
\draw  [fill={rgb, 255:red, 0; green, 0; blue, 0 }  ,fill opacity=1 ] (173,164.33) -- (180,161.33) -- (173,158.33) -- cycle ;
\draw  [fill={rgb, 255:red, 0; green, 0; blue, 0 }  ,fill opacity=1 ] (193,216.5) -- (196,210.13) -- (189.5,213.08)  -- cycle ;
\draw  [fill={rgb, 255:red, 0; green, 0; blue, 0 }  ,fill opacity=1 ] (191,111.59) -- (198,114.5) -- (194,108)  -- cycle ;

\draw  [fill={rgb, 255:red, 0; green, 0; blue, 0 }  ,fill opacity=1 ][line width=1.5]  (250.42,161.03) .. controls (250.42,160.76) and (250.47,160.53) .. (250.53,160.53) .. controls (250.6,160.53) and (250.65,160.76) .. (250.65,161.03) .. controls (250.65,161.31) and (250.6,161.53) .. (250.53,161.53) .. controls (250.47,161.53) and (250.42,161.31) .. (250.42,161.03) -- cycle ;
\draw  [fill={rgb, 255:red, 0; green, 0; blue, 0 }  ,fill opacity=1 ] (284,144.8) -- (291,141.8) -- (284,138.8) -- cycle ;
\draw  [fill={rgb, 255:red, 0; green, 0; blue, 0 }  ,fill opacity=1 ] (284,184.1) -- (291,181.1) -- (284,178.1) -- cycle ;
\draw  [fill={rgb, 255:red, 0; green, 0; blue, 0 }  ,fill opacity=1 ] (353,144.2) -- (360,141.2) -- (353,138.2) -- cycle ;
\draw  [fill={rgb, 255:red, 0; green, 0; blue, 0 }  ,fill opacity=1 ] (353,184.1) -- (360,181.1) -- (353,178.1) -- cycle ;

\draw    (391.42,160.7) -- (428.91,160.53) -- (502.53,160.2) ;
\draw    (391.65,160.7) -- (471.98,231.14) ;
\draw    (391.42,160.7) -- (472.21,87.57) ;
\draw  [fill={rgb, 255:red, 0; green, 0; blue, 0 }  ,fill opacity=1 ] (461,163) -- (468,160.33) -- (461,157.66) -- cycle ;
\draw  [fill={rgb, 255:red, 0; green, 0; blue, 0 }  ,fill opacity=1 ] (444,117) -- (447,110) -- (440,113) -- cycle ;
\draw  [fill={rgb, 255:red, 0; green, 0; blue, 0 }  ,fill opacity=1 ] (439.01,206.25) -- (446,208.7) -- (443.25,202.01) -- cycle ;

\draw  [fill={rgb, 255:red, 0; green, 0; blue, 0 }  ,fill opacity=1 ][line width=1.5]  (321.09,160.7) .. controls (321.09,160.42) and (321.14,160.2) .. (321.2,160.2) .. controls (321.26,160.2) and (321.31,160.42) .. (321.31,160.7) .. controls (321.31,160.98) and (321.26,161.2) .. (321.2,161.2) .. controls (321.14,161.2) and (321.09,160.98) .. (321.09,160.7) -- cycle ;
\draw    (250.53,161.53) -- (321.55,160.91) ;
\draw   (250.53,161.53) .. controls (250.53,150.49) and (266.33,141.53) .. (285.81,141.53) .. controls (305.29,141.53) and (321.09,150.49) .. (321.09,161.53) .. controls (321.09,172.58) and (305.29,181.53) .. (285.81,181.53) .. controls (266.33,181.53) and (250.53,172.58) .. (250.53,161.53) -- cycle ;
\draw    (321.2,161.2) -- (392.21,160.57) ;
\draw   (321.2,161.2) .. controls (321.2,150.15) and (336.99,141.2) .. (356.48,141.2) .. controls (375.96,141.2) and (391.75,150.15) .. (391.75,161.2) .. controls (391.75,172.25) and (375.96,181.2) .. (356.48,181.2) .. controls (336.99,181.2) and (321.2,172.25) .. (321.2,161.2) -- cycle ;

\draw (237.8,170.44) node [anchor=north west][inner sep=0.75pt]   [align=left] {$-1$};
\draw (150,72) node [anchor=north west][inner sep=0.75pt]   [align=left] {$\Gamma_2^{(1)}$};
\draw (113,148.44) node [anchor=north west][inner sep=0.75pt]   [align=left] {$\Gamma_3^{(1)}$};
\draw (145,226.44) node [anchor=north west][inner sep=0.75pt]   [align=left] {$\Gamma_4^{(1)}$};
\draw (316.47,170.44) node [anchor=north west][inner sep=0.75pt]   [align=left] {0};
\draw (473.33,226.44) node [anchor=north west][inner sep=0.75pt]   [align=left] {$\Gamma_5^{(1)}$};
\draw (473.33,72.44) node [anchor=north west][inner sep=0.75pt]   [align=left] {$\Gamma_1^{(1)}$};
\draw (505.33,148.44) node [anchor=north west][inner sep=0.75pt]   [align=left] {$\Gamma_0^{(1)}$};
\draw (275.48,143.6) node [anchor=north west][inner sep=0.75pt]    {$\Omega_{\msf L,+}$};
\draw (275,163.4) node [anchor=north west][inner sep=0.75pt]    {$\Omega_{\msf L,-}$};
\draw (272,123.73) node [anchor=north west][inner sep=0.75pt]    {$\partial \Omega_{\msf L,+}$};
\draw (272,186.07) node [anchor=north west][inner sep=0.75pt]    {$\partial \Omega_{\msf L,-}$};
\draw (387.93,170.44) node [anchor=north west][inner sep=0.75pt]   [align=left] {$1$};
\draw (345.14,143.27) node [anchor=north west][inner sep=0.75pt]    {$\Omega_{\msf R,+}$};
\draw (345.14,163.07) node [anchor=north west][inner sep=0.75pt]    {$\Omega_{\msf R,-}$};
\draw (342,123.73) node [anchor=north west][inner sep=0.75pt]    {$\partial \Omega_{\msf R,+}$};
\draw (342.67,186.07) node [anchor=north west][inner sep=0.75pt]    {$\partial \Omega_{\msf R,-}$};

\end{tikzpicture}
\caption{Regions $\Omega_{\msf L, \pm}$, $\Omega_{\msf R, \pm}$ and the jump contours of RH problem \ref{rhp:hatS} for $\what S$.}
\label{fig:lenses}
\end{figure}

\begin{rhp}\label{rhp:hatS}
\hfill
\begin{enumerate}
\item[\rm (a)] $\what S(z)$ is defined and analytic in $\mathbb{C} \setminus \Gamma_{\what S}$, where
\begin{equation}\label{def:gammahatS}
\Gamma_{\what S}:=\cup^5_{j=0}\Gamma_j^{(1)}\cup [-1,1] \cup \partial \Omega_{\msf L, \pm} \cup \partial \Omega_{\msf R, \pm};
\end{equation}
see Figure \ref{fig:lenses} for an illustration.

\item[\rm (b)] For $z\in \Gamma_{\what S}$, $\what S(z)$ satisfies the jump condition
\begin{equation}\label{eq:hatS-jump}
 \what S_+(z)=\what S_-(z)J_{\what S}(z),
\end{equation}
where
\begin{equation}\label{def:JhatS}
J_{\what S}(z):=\left\{
 \begin{array}{ll}
          J_{\what T}(z), & \qquad \hbox{$z\in \cup^5_{j=0}\Gamma_j^{(1)}$,} \\
          \msf J_2(z),  & \qquad  \hbox{$z\in \partial \Omega_{\msf L,+}$,} \\
          \msf J_1(z),  & \qquad \hbox{$z\in \partial \Omega_{\msf L,-}$,} \\
          \msf J_4(z), & \qquad  \hbox{$z\in \partial \Omega_{\msf R,+}$,} \\
          \msf J_3(z), & \qquad  \hbox{$z\in \partial \Omega_{\msf R,-}$,} \\
           \begin{pmatrix}
             1&0&0&0
             \\
            0&0&0&1-\gamma
             \\
            0&0&1&0
           \\0&\frac{1}{\gamma-1}&0&0
          \end{pmatrix}, & \qquad  \hbox{$z \in (-1,0)$,}
          \\
        \begin{pmatrix}
0&0&1-\gamma&0
\\
0&1&0&0
\\
\frac{1}{\gamma-1}&0&0&0
\\0&0&0&1
 \end{pmatrix}, & \qquad  \hbox{$z \in (0,1)$,}
        \end{array}
      \right.
 \end{equation}
where the functions $J_{\what T}(z)$ and $\msf J_i(z)$, $i=1,\ldots,4$, are defined in \eqref{def:JhatT}, \eqref{def:msfJ1}, \eqref{def:msfJ2}, \eqref{def:msfJ3} and \eqref{def:msfJ4}, respectively.

\item[\rm (c)]As $z \to \infty$ with $z\in \mathbb{C} \setminus \Gamma_{\what S}$, we have
\begin{align}\label{eq:asyhatS}
\what S(z)=\left(I+ \frac{\what T^{(1)}}{z} + \Boh(z^{-2}) \right) \diag \left((-z)^{-\frac14},z^{-\frac14},(-z)^{\frac14},z^{\frac14} \right)A,
\end{align}
where $\what T^{(1)}$ is given in \eqref{eq:asyhatT} and $A$ is defined in \eqref{def:A}.

\item[\rm (d)]
As $z \to \pm 1$, we have $\what S(z)=\Boh(\ln(z \mp 1))$.
\end{enumerate}
\end{rhp}

\subsection{Global parametrix}
As $s\to +\infty$, it is now readily seen that all the jump matrices of $\what S$ tend to the identity matrices exponentially
fast except for those along $\mathbb{R}$ and we are lead to consider the following global parametrix.

\begin{rhp}\label{rhp:hatN}
\hfill
\begin{enumerate}
\item[\rm (a)] $\what N(z)$ is defined and analytic in $\mathbb{C}\setminus \mathbb R$.

\item[\rm (b)] For $x\in \mathbb{R}$, $\what N(x)$ satisfies the jump condition
\begin{equation}\label{eq:hatN-jump}
 \what N_+(x)=\what N_-(x)\left\{ \begin{array}{ll}
 \begin{pmatrix}
1&0&0&0\\0&0&0&1\\0&0&1&0\\0&-1&0&0
\end{pmatrix}, & \quad \hbox{$x<-1$,}
\\
\begin{pmatrix}
             1&0&0&0
             \\
            0&0&0&1-\gamma
             \\
            0&0&1&0
           \\0&\frac{1}{\gamma-1}&0&0
\end{pmatrix}, & \quad  \hbox{$x \in (-1,0)$,}
\\
\begin{pmatrix}
0&0&1-\gamma&0
\\
0&1&0&0
\\
\frac{1}{\gamma-1}&0&0&0
\\0&0&0&1
 \end{pmatrix}, & \quad  \hbox{$x \in (0,1)$,}
\\
\begin{pmatrix}0&0&1&0\\0&1&0&0\\-1&0&0&0\\0&0&0&1 \end{pmatrix}, & \quad  \hbox{$x>1$.}
\end{array}
              \right.
\end{equation}

\item[\rm (c)]As $z \to \infty$,  we have
\begin{align}
\what N(z)=\left( I+ \Boh(z^{-2}) \right) \diag \left((-z)^{-\frac14},z^{-\frac14},(-z)^{\frac14},z^{\frac14} \right)A,
\end{align}
where $A$ is defined in \eqref{def:A}.

\end{enumerate}
\end{rhp}

To solve the above RH problem, we observe from the jump condition for $\what N$ that it is natural to expect $\what N$ should take a chessboard structure, i.e.,
$$
\what N=
\begin{pmatrix}
\star & 0 & \star & 0
\\
0 & \star & 0 & \star
\\
\star & 0 & \star & 0
\\
0 & \star & 0 & \star
\end{pmatrix},
$$
where $\star$ denotes a matrix entry to be specified. This sparsity pattern then implies that we could decompose $\what N$ into two $2\times 2$ RH problems. Indeed, let
\begin{equation}\label{def:hatN1}
\what N_1(z):=
\begin{pmatrix}
\what N_{22}(z) & \what N_{24}(z)
\\
\what N_{42}(z) & \what N_{44}(z)
\end{pmatrix}.
\end{equation}
It is readily seen that $\what N_1$ solves the following RH problem.
\begin{rhp}\label{rhp:hatN1}
\hfill
\begin{enumerate}
\item[\rm (a)] $\what N_1(z)$ is defined and analytic in $\mathbb{C}\setminus (-\infty,0]$.

\item[\rm (b)] For $x\in (-\infty,0)$, $\what N_1(x)$ satisfies the jump condition
\begin{equation}\label{eq:hatN1-jump}
 \what N_{1,+}(x)=\what N_{1,-}(x)\left\{ \begin{array}{ll}
 \begin{pmatrix}
0&1
\\
-1&0
\end{pmatrix}, & \quad \hbox{$x<-1$,}
\\
\begin{pmatrix}
0&1-\gamma
\\
\frac{1}{\gamma-1}&0
\end{pmatrix}, & \quad  \hbox{$x \in (-1,0)$.}
\end{array}
              \right.
\end{equation}

\item[\rm (c)]As $z \to \infty$,  we have
\begin{align}\label{eq:asyhatN1}
\what N_1(z)=\left( I+ \Boh(z^{-1}) \right) \frac{ z^{- \sigma_3/4}}{\sqrt 2}
\begin{pmatrix}
1 & \ii
\\
\ii & 1
\end{pmatrix}.
\end{align}
\end{enumerate}
\end{rhp}

Similarly,
\begin{equation}\label{def:hatN2}
\what N_2(z):=
\begin{pmatrix}
\what N_{11}(z) & \what N_{13}(z)
\\
\what N_{31}(z) & \what N_{33}(z)
\end{pmatrix}
\end{equation}
is a solution of the following RH problem.
\begin{rhp}\label{rhp:hatN2}
\hfill
\begin{enumerate}
\item[\rm (a)] $\what N_2(z)$ is defined and analytic in $\mathbb{C}\setminus [0,\infty)$.

\item[\rm (b)] For $x\in (0,\infty)$, $\what N_2(x)$ satisfies the jump condition
\begin{equation}\label{eq:hatN2-jump}
 \what N_{2,+}(x)=\what N_{2,-}(x)\left\{ \begin{array}{ll}
 \begin{pmatrix}
 0&1-\gamma
 \\
 \frac{1}{\gamma-1}&0
 \end{pmatrix}, & \quad \hbox{$x\in (0,1)$,}
 \\
 \begin{pmatrix}
 0&1
 \\
 -1&0
 \end{pmatrix}, & \quad  \hbox{$x > 1$.}
\end{array}
              \right.
\end{equation}

\item[\rm (c)]As $z \to \infty$,  we have
\begin{align}\label{eq:asyhatN2}
\what N_2(z)=\left( I+ \Boh(z^{-1}) \right)  \frac{(-z)^{- \sigma_3/4}}{\sqrt 2}
\begin{pmatrix}
1 & -\ii
\\
-\ii & 1
\end{pmatrix}.
\end{align}
\end{enumerate}
\end{rhp}
Since it is easily seen that
\begin{equation}\label{eq:hatN2inhatN1}
\what N_2(z)=\sigma_3 \what N_1(-z) \sigma_3,
\end{equation}
it suffices to solve RH problem \ref{rhp:hatN1} for $\what N_1$. For that purpose, we set
\begin{equation}\label{def:lambda}
\lambda(\zeta):= \left(\frac{\zeta-\ii}{\zeta+\ii}\right)^{\beta},\qquad \zeta \in \mathbb{C} \setminus [-\ii,\ii],
\end{equation}
where $\beta=\ln (1-\gamma)/(2\pi \ii)$ (see \eqref{def:beta}) and the branch cut is chosen such that $\lambda(\zeta) \to  1$ as $\zeta \to \infty$ with the orientation from $\ii$ to $-\ii$. With the aid of the function $\lambda$, we define
\begin{equation}\label{def:di}
d_1(z)=\lambda(z^{\frac12}), \qquad d_2(z)=\lambda(-z^{\frac12}).
\end{equation}
Some properties of $d_1$ and $d_2$ are collected in the proposition below.
\begin{proposition}\label{prop:di}
The functions $d_1(z)$ and $d_2(z)$ in \eqref{def:di} satisfy the following properties.
\begin{enumerate}
\item[\rm (i)] $d_1(z)$ and $d_2(z)$ are analytic in $\mathbb{C} \setminus (-\infty,0]$. Moreover, we have
\begin{align}
d_{1,\pm}(x)&=d_{2,\mp}(x), && x<-1, \label{eq:dijump1}
\\
d_{1,\pm}(x)&=d_{2,\mp}(x)e^{-2\beta \pi \ii},  && -1<x<0.
\end{align}
\item[\rm (ii)] As $z\to \infty$, we have
\begin{align}
d_1(z)&=1-\frac{2\beta \ii}{z^{1/2}}-\frac{2\beta^2}{z}+\Boh(z^{-2}),
\\
d_2(z)&=1+\frac{2\beta \ii}{z^{1/2}}-\frac{2\beta^2}{z}+\Boh(z^{-2}).
\end{align}
\item[\rm (iii)] As $z \to 0$, we have
\begin{align}
d_1(z) &= e^{-\beta \pi \ii}\left(1 + 2 \ii \beta \sqrt{z}- 2 \beta^2 z - \frac{2 \ii (\beta + 2 \beta^3)}{3} z^{\frac 32} + \Boh(z^2)\right), \label{eq:d10}
\\
d_2(z) &= e^{\beta \pi \ii}\left(1 - 2 \ii \beta \sqrt{z}- 2 \beta^2 z + \frac{2 \ii (\beta + 2 \beta^3)}{3} z^{\frac 32} + \Boh(z^2)\right). \label{eq:d20}
\end{align}
\item[\rm (iv)] As $ z \to -1$ and $\Im{z}>0$, we have
\begin{align}
d_1(z) & = e^{-\beta \pi \ii} 4^{-\beta} (z+1)^{\beta}\left(1+ \frac{\beta}{2} (z+1) + \Boh\left((z+1)^2\right)\right), \label{eq:d1-1}
\\
d_2(z) & =  e^{\beta \pi \ii} 4^{\beta} (z+1)^{-\beta}\left(1- \frac{\beta}{2} (z+1) + \Boh\left((z+1)^2\right)\right). \label{eq:d2-1}
\end{align}
\item[\rm (v)] As $z \to 1$, we have
\begin{align}
d_1(z) & = e^{-\frac{\beta \pi \ii}{2}}\left(1 + \frac{\beta \ii}{2}(z-1) + \Boh\left((z-1)^2\right)\right), \label{eq:d11}
\\
d_2(z) & = e^{\frac{\beta \pi \ii}{2}}\left(1 - \frac{\beta \ii}{2}(z-1) + \Boh\left((z-1)^2\right)\right). \label{eq:d21}
\end{align}
\end{enumerate}
\end{proposition}
\begin{proof}
From the definition of $\lambda$ given in \eqref{def:lambda}, it is readily seen that
\begin{equation}
\lambda_{+}(\zeta)=\lambda_{-}(\zeta)\left\{
                                       \begin{array}{ll}
                                         e^{-2\beta \pi \ii}, & \hbox{$\zeta \in (0,\ii)$,} \\
                                         e^{2\beta \pi \ii}, & \hbox{$\zeta \in (-\ii,0)$,}
                                       \end{array}
                                     \right.
\end{equation}
and
\begin{align}
\lambda(\zeta)&=1-\frac{2 \beta \ii}{\zeta}-\frac{2\beta^2}{\zeta^2}+\Boh(\zeta^{-3}), \quad \zeta \to \infty,\\
\lambda(\zeta)&=\begin{cases}
e^{-\beta \pi \ii}\left(1 + 2\ii \beta \zeta-2\beta^2 \zeta^2 -\frac{2 \ii (\beta + \beta^3)}{3} \zeta^3 + \Boh(\zeta^4)\right), & \quad \Re{\zeta}>0,\\
e^{\beta \pi \ii}\left(1 + 2\ii \beta \zeta-2\beta^2 \zeta^2 -\frac{2 \ii (\beta + \beta^3)}{3} \zeta^3 + \Boh(\zeta^4)\right), & \quad \Re{\zeta}<0,
\end{cases} \quad \zeta \to 0,\\
\lambda(\zeta)&= e^{-\frac{\beta \pi \ii}{2}} (\zeta-\ii)^{\beta}\left(2-\frac{\beta}{2}(\zeta-\ii)+\Boh\left((\zeta-\ii)^2\right)\right),\quad \zeta \to \ii,\\
\lambda(\zeta)&= e^{-\frac{\beta \pi \ii}{2}} \left(1 + \beta \ii (\zeta-1) + \Boh\left((\zeta-1)^2\right)\right), \quad \zeta \to 1.
\end{align}
This, together with \eqref{def:di}, gives us the claims in the proposition after straightforward calculations.
\end{proof}

We can now solve RH problem \ref{rhp:hatN1} for $\what N_1$ by using the functions $d_1$ and $d_2$.
\begin{proposition}\label{prop:solhatN1}
Let $d_1$ and $d_2$ be two functions defined in \eqref{def:di}. A solution of RH problem \ref{rhp:hatN1} is given by
\begin{equation}\label{eq:hatN1exp}
\what N_1(z)=
\begin{pmatrix}
1 & 0
\\
-2\beta & 1
\end{pmatrix}\frac{ z^{- \sigma_3/4}}{\sqrt 2}
\begin{pmatrix}
1 & \ii
\\
\ii & 1
\end{pmatrix}\diag (d_1(z),d_2(z)).
\end{equation}
\end{proposition}
\begin{proof}
By \eqref{eq:hatN1exp} and \eqref{eq:dijump1}, it follows that for $x<-1$
\begin{align*}
& \what N_{1,-}(z)^{-1}\what N_{1,+}(z)
\nonumber
\\
&=\frac12\diag\left(\frac{1}{d_{1,-}(x)},\frac{1}{d_{2,-}(x)}\right)
\begin{pmatrix}
1 & -\ii
\\
-\ii & 1
\end{pmatrix}\diag(-\ii,\ii)
\begin{pmatrix}
1 & \ii
\\
\ii & 1
\end{pmatrix}\diag (d_{1,+}(x),d_{2,+}(x))
\nonumber
\\
&= \begin{pmatrix}
0 & 1
\\
-1 & 0
\end{pmatrix},
\end{align*}
as required. The jump of $\what N_1$ on $(-1,0)$ can be checked in a similar manner and we omit the details here.

To show the large $z$ behavior of $\what N_1$ in \eqref{eq:hatN1exp}, we observe from item (ii) of Proposition \ref{prop:di} that, as $z\to \infty$,
\begin{align*}
\frac{z^{- \sigma_3/4}}{\sqrt 2}
\begin{pmatrix}
1 & \ii
\\
\ii & 1
\end{pmatrix}\diag (d_1(z),d_2(z))=\left(\begin{pmatrix}
1 & 0
\\
2\beta & 1
\end{pmatrix}+\Boh(z^{-1})\right)\frac{z^{- \sigma_3/4}}{\sqrt 2}
\begin{pmatrix}
1 & \ii
\\
\ii & 1
\end{pmatrix}.
\end{align*}
Thus, $\what N_1$ in \eqref{eq:hatN1exp} indeed satisfies the asymptotic condition \eqref{eq:asyhatN1}.

This completes the proof of Proposition \ref{prop:solhatN1}.
\end{proof}

In view of \eqref{def:hatN1}, \eqref{def:hatN2}, \eqref{eq:hatN2inhatN1} and Proposition \ref{prop:solhatN1}, the following lemma is immediate.
\begin{lemma}
A solution of RH problem \ref{rhp:hatN} is given by
\begin{align}\label{eq:hatNexp}
\what N(z)
&=
(I+2\beta E_{3,1}-2\beta E_{4,2})
\diag \left((-z)^{-\frac14},z^{-\frac14},(-z)^{\frac14},z^{\frac14} \right)A
\nonumber
\\
& \quad \times \diag \left(d_1(-z),d_1(z),d_2(-z),d_2(z) \right),
\end{align}
where the functions $A$, $d_1$ and $d_2$ are defined in \eqref{def:A} and \eqref{def:di}. Moreover, we have
\begin{align}\label{eq:asyhatN}
\what N(z)
&=
\left(I+\frac{\what N^{(1)}}{z}+\Boh(z^{-2})
\right)
\diag \left((-z)^{-\frac14},z^{-\frac14},(-z)^{\frac14},z^{\frac14} \right)A, \quad z \to \infty,
\end{align}
where $A$ is defined in \eqref{def:A} and
\begin{equation}\label{def:hatN1exp}
\what N^{(1)}=\begin{pmatrix}
2 \beta^2 & 0 & -2\beta & 0\\
0 & \ast & 0 & \ast\\
\ast & 0 & \ast & 0\\
0 & \ast & 0 & \ast
\end{pmatrix}.
\end{equation}
\end{lemma}

\subsection{Local parametrix near $z=0$}

Since the jump matrices for $\what S$ and $\what N$ are not uniformly close to each other near $z=0$ and $z=\pm1$, we need to construct the local parametrices near these points and start with the local parametrix near the origin.

\begin{rhp}\label{rhp:hatP0}
\hfill
\begin{enumerate}
\item[\rm (a)] $\what P^{(0)}(z)$ is defined and analytic in $D(0, \varepsilon)\setminus \Gamma_{\what S}$, where $\Gamma_{\what S}$ is defined in \eqref{def:gammahatS}.

\item[\rm (b)] For $z \in D(0, \varepsilon) \cap \Gamma_{\what S}$, $\what P^{(0)}(z)$ satisfies the jump condition
\begin{equation}\label{eq:hatP0-jump}
 \what P^{(0)}_+(z)=\what P^{(0)}_-(z)J_{\what S}(z),
\end{equation}
where $J_{\what S}(z)$ is given in \eqref{def:JhatS}.

\item[\rm (c)]As $s \to \infty$,  we have the matching condition
\begin{equation}\label{eq:matchhatP0}
\what P^{(0)}(z)=\left( I+ \Boh(s^{-\frac 12}) \right) \what N(z),\quad z \in \partial D(0, \varepsilon),
\end{equation}
where $\what N(z)$ is given in \eqref{eq:hatNexp}.
\end{enumerate}
\end{rhp}

This RH problem can be solved by using the solution $M(z)$ of the tacnode RH problem \ref{rhp:tac}. More precisely, we define
\begin{align}\label{def:hatP0}
\what P^{(0)}(z) & = \what E_0(z) M(sz) \diag \left((1-\gamma)^{-\frac 12}e^{\theta_1(sz) - \tau sz},(1-\gamma)^{-\frac 12}e^{\theta_2(sz) + \tau sz},(1-\gamma)^{\frac 12}e^{-\theta_1(sz) - \tau sz}\right.,\nonumber\\
&\qquad  \left.(1-\gamma)^{\frac 12}e^{-\theta_2(sz) + \tau sz}\right),
\end{align}
with
\begin{align}\label{def:hatE0}
\what E_0(z)& = \what N(z)\diag\left((1-\gamma)^{\frac 12},(1-\gamma)^{\frac 12},(1-\gamma)^{-\frac 12},(1-\gamma)^{-\frac 12}\right) A^{-1}\nonumber\\
&\quad \times \diag \left((-sz)^{\frac 14}, (sz)^{\frac 14},(-sz)^{-\frac 14},(sz)^{-\frac 14}\right),
\end{align}
where $A$ is defined in \eqref{def:A} and $\what N (z)$ is given in \eqref{eq:hatNexp}.

\begin{proposition}
The local parametrix $\what P^{(0)}(z)$ defined in \eqref{def:hatP0} solves RH problem \ref{rhp:hatP0}.
\end{proposition}
\begin{proof}
We first show the prefactor $\what E_0(z)$ is analytic in $D(0, \varepsilon)$. From its definition in \eqref{def:hatE0}, the only possible jump is on $(-\varepsilon, \varepsilon)$. For $x \in (-\varepsilon, 0)$, recalling the jump matrix of $\what{N}(z)$ in \eqref{eq:hatN-jump}, we have
\begin{align}
&\what E_{0,-}(x)^{-1} \what E_{0,+}(x)\nonumber\\
&= \diag \left((-sz)^{-\frac 14}, (sz)_-^{-\frac 14},(-sz)^{\frac 14},(sz)_-^{\frac 14}\right) A \diag\left((1-\gamma)^{-\frac 12},(1-\gamma)^{-\frac 12},(1-\gamma)^{\frac 12},(1-\gamma)^{\frac 12}\right)\nonumber\\
&\quad \times \begin{pmatrix}
             1&0&0&0
             \\
            0&0&0&1-\gamma
             \\
            0&0&1&0
           \\0&\frac{1}{\gamma-1}&0&0
\end{pmatrix} \diag\left((1-\gamma)^{\frac 12},(1-\gamma)^{\frac 12},(1-\gamma)^{-\frac 12},(1-\gamma)^{-\frac 12}\right) A^{-1}\nonumber\\
&\quad \times \diag \left((-sz)^{\frac 14}, (sz)_+^{\frac 14},(-sz)^{-\frac 14},(sz)_+^{-\frac 14}\right)\nonumber\\
&=\diag \left((-sz)^{-\frac 14}, (sz)_-^{-\frac 14},(-sz)^{\frac 14},(sz)_-^{\frac 14}\right) \diag(1,-\ii,1,\ii)\nonumber\\
&\quad \times \diag \left((-sz)^{\frac 14}, (sz)_+^{\frac 14},(-sz)^{-\frac 14},(sz)_+^{-\frac 14}\right)=I.
\end{align}
Similarly, for $x \in (0, \varepsilon)$, we have
\begin{align}
&\what E_{0,-}(x)^{-1} \what E_{0,+}(x)\nonumber\\
&= \diag \left((-sz)_-^{-\frac 14}, (sz)^{-\frac 14},(-sz)_-^{\frac 14},(sz)^{\frac 14}\right) A \diag\left((1-\gamma)^{-\frac 12},(1-\gamma)^{-\frac 12},(1-\gamma)^{\frac 12},(1-\gamma)^{\frac 12}\right)\nonumber\\
&\quad \times \begin{pmatrix}
             1&0&1-\gamma&0
             \\
            0&1&0&0
             \\
            \frac{1}{\gamma-1}&0&0&0
           \\0&0&0&1
\end{pmatrix} \diag\left((1-\gamma)^{\frac 12},(1-\gamma)^{\frac 12},(1-\gamma)^{-\frac 12},(1-\gamma)^{-\frac 12}\right) A^{-1}\nonumber\\
&\quad \times \diag \left((-sz)_+^{\frac 14}, (sz)^{\frac 14},(-sz)_+^{-\frac 14},(sz)^{-\frac 14}\right)\nonumber\\
&=\diag \left((-sz)_-^{-\frac 14}, (sz)^{-\frac 14},(-sz)_-^{\frac 14},(sz)^{\frac 14}\right)\diag{(\ii,1,-\ii,1)}\nonumber\\
&\quad \times \diag \left((-sz)_+^{\frac 14}, (sz)^{\frac 14},(-sz)_+^{-\frac 14},(sz)^{-\frac 14}\right)=I.
\end{align}
Moreover, as $z \to 0$, one has
\begin{equation}
\what E_0(z) = \what E_0(0)\left(I + \what E_0(0)^{-1}\what E_0'(0)z + \Boh(z^{-2})\right)\diag \left(s^{\frac 14}, s^{\frac 14},s^{-\frac 14},s^{-\frac 14}\right),
\end{equation}
where
\begin{equation}
\what E_0(0) = \begin{pmatrix}
1 & 0 & -2\beta & 0\\
0 & 1 & 0 & 2 \beta\\
2 \beta & 0 & 1- 4 \beta^2 & 0\\
0 & -2 \beta & 0 & 1-4 \beta^2
\end{pmatrix}
\end{equation}
and
\begin{equation}
\what E_0(0)^{-1}\what E_0'(0) = \begin{pmatrix}
-2 \beta^2 & 0 & \frac{2\beta (4\beta^2-1)}{3} & 0\\
0 & 2 \beta^2 & 0 & \frac{2\beta (4\beta^2-1)}{3}\\
-2 \beta & 0 & 2 \beta^2 & 0\\
0 & -2 \beta & 0 & -2\beta^2
\end{pmatrix}.
\end{equation}
Thus, $\what E_0(z)$ is indeed analytic in $D(0, \varepsilon)$ and the jump condition \eqref{eq:hatP0-jump} can be verified easily from this fact, \eqref{eq:thetairelations} and
the jump condition of $M$ given in \eqref{jumps:M}.

For the matching condition, it follows from the asymptotic behavior of $M$ at infinity given in \eqref{eq:asy:M} that, for $z \in \partial D(0, \varepsilon)$,
\begin{equation}\label{eq:match01}
\what P^{(0)}(z)\what N(z)^{-1} = I + \frac{\what J^{(0)}_1(z)}{s^{1/2}} +  \frac{\what J^{(0)}_2(z)}{s}+\Boh(s^{-\frac 32}),\quad s \to +\infty,
\end{equation}
where
\begin{equation}\label{def:hatJ01}
\what J^{(0)}_1(z) = \frac{1}{z} E_0(z) \begin{pmatrix}
0 & 0 &M^{(1)}_{13} &M^{(1)}_{14}\\
0 & 0 &M^{(1)}_{23} &M^{(1)}_{24}\\
0 & 0 & 0 & 0\\
0 & 0 & 0 &0
\end{pmatrix}E_0(z)^{-1}
\end{equation}
and
\begin{align}\label{def:hatJ02}
\what J^{(0)}_2(z)=
 \frac{1}{z} E_0(z) \begin{pmatrix}
M^{(1)}_{11} & M^{(1)}_{12} &0 &0\\
M^{(1)}_{21} & M^{(1)}_{22} &0 & 0\\
0 & 0 & M^{(1)}_{33} & M^{(1)}_{34}\\
0 & 0 & M^{(1)}_{43} &M^{(1)}_{44}
\end{pmatrix}E_0(z)^{-1}
\end{align}
with $E_0(z) := \what E_0(z) \diag \left(s^{-\frac 14}, s^{-\frac 14},s^{\frac 14},s^{\frac 14}\right)$ and $M^{(1)}$ given in \eqref{eq:asy:M}.
\end{proof}

\subsection{Local parametrix near $z=-1$}
Near $z=-1$, we intend to find an RH problem as follows.
\begin{rhp}\label{rhp:hatP-1}
\hfill
\begin{enumerate}
\item[\rm (a)] $\what P^{(-1)}(z)$ is defined and analytic in $D(-1, \varepsilon)\setminus \Gamma_{\what S}$, where $\Gamma_{\what S}$ is defined in \eqref{def:gammahatS}.

\item[\rm (b)] For $z \in D(-1, \varepsilon) \cap \Gamma_{\what S}$, $\what P^{(-1)}(z)$ satisfies the jump condition
\begin{equation}\label{eq:hatP-1-jump}
 \what P^{(-1)}_+(z)=\what P^{(-1)}_-(z)J_{\what S}(z),
\end{equation}
where $J_{\what S}(z)$ is given in \eqref{def:JhatS}.

\item[\rm (c)]As $s \to +\infty$,  we have the matching condition
\begin{equation}\label{eq:matchhatP-1}
\what P^{(-1)}(z)=\left( I+ \Boh(s^{-\frac 32}) \right) \what N(z), \quad z \in \partial D(-1, \varepsilon),
\end{equation}
where $\what N(z)$ is given in \eqref{eq:hatNexp}.
\end{enumerate}
\end{rhp}
This local parametrix can be constructed by using the confluent hypergeometric parametrix $\Phi^{(\CHF)}$ introduced in Appendix \ref{sec:CHF}. To proceed, we introduce the function
\begin{equation}\label{def:hatf-1}
\what f_{-1}(z) = -2 \ii s^{-\frac 32} \begin{cases}
\theta_2(sz) - \theta_{2,+}(-s), & \quad\Im{z} > 0,\\
-\theta_2(sz) + \theta_{2,-}(-s), & \quad\Im{z} < 0.
\end{cases}
\end{equation}
By \eqref{def:theta2}, it is easily seen that
\begin{equation}\label{eq:hatf-1}
\what f_{-1}(z) = 2\left(r_2 - \frac{s_2}{s}\right)(z+1) + \Boh\left((z+1)^2\right), \quad z \to -1.
\end{equation}
Thus, $\what f_{-1}(z)$ is a conformal mapping near $z=-1$ for large positive $s$.

We set
\begin{align}\label{def:hatP-1}
&\what P^{(-1)}(z)
\nonumber \\
& = \what E_{-1}(z) \begin{pmatrix}
1 & 0 & 0 &0\\
0 & \Phi^{(\CHF)}_{11}(s^{\frac 32} \what f_{-1}(z); -\beta) & 0 & \Phi^{(\CHF)}_{12}(s^{\frac 32} \what f_{-1}(z); -\beta)\\
0 & 0 &1&0\\
0& \Phi^{(\CHF)}_{21}(s^{\frac 32} \what f_{-1}(z); -\beta) & 0 & \Phi^{(\CHF)}_{22}(s^{\frac 32} \what f_{-1}(z); -\beta)
\end{pmatrix}\nonumber\\
&\quad \times \diag \left(1, e^{\theta_2(sz) - \frac{\beta \pi \ii}{2}}, 1, e^{-\theta_2(sz) + \frac{\beta \pi \ii}{2}}\right)\nonumber\\
&\quad \times \begin{cases}
I-e^{-\theta_1(sz)+\theta_2(sz)+2 \tau sz} E_{1,2}+e^{-\theta_1(sz)+\theta_2(sz)-2 \tau sz}E_{4,3}, &\quad z \in (\Omega_2\cup\Omega_5)\cap D(-1, \varepsilon) \setminus \partial \Omega_{L,\pm},\\
I, & \quad z \in (\Omega_3\cup\Omega_4)\cap D(-1, \varepsilon),
\end{cases}
\end{align}
where $\beta$ is given in \eqref{def:beta}, $\what f_{-1}$ is defined in \eqref{def:hatf-1} and
\begin{align}\label{eq:hatE-1}
&\what E_{-1}(z) = \what N(z)\nonumber\\
 &\times \begin{cases}
\diag \left(1, e^{-\theta_{2,+}(-s) + \frac{\beta \pi \ii}{2}}s^{-\frac{3\beta}{2}} \what f_{-1}(z)^{-\beta}, 1, e^{\theta_{2,+}(-s) -\frac{\beta \pi \ii}{2}}s^{\frac{3\beta}{2}} \what f_{-1}(z)^{\beta}\right),
\hfill \Im{z}>0,\\
\begin{pmatrix}
1 & 0 & 0 & 0\\
0 & 0 & 0 & 1\\
0 & 0 & 1 & 0\\
0 & -1 & 0 & 0
\end{pmatrix}\diag \left(1, e^{-\theta_{2,+}(-s) - \frac{3\beta \pi \ii}{2}}s^{-\frac{3\beta}{2}} \what f_{-1}(z)^{-\beta}, 1, e^{\theta_{2,+}(-s) +\frac{3\beta \pi \ii}{2}}s^{\frac{3\beta}{2}} \what f_{-1}(z)^{\beta}\right), \\
\hfill \Im{z}<0.
\end{cases}
\end{align}
\begin{proposition}\label{pro:hatP-1}
The local parametrix $\what P^{(-1)}(z)$ defined in \eqref{def:hatP-1} solves RH problem \ref{rhp:hatP-1}.
\end{proposition}

\begin{proof}
We first show the analyticity of $\what E_{-1}(z)$. From its definition in \eqref{eq:hatE-1}, the only possible jump is on $(-1-\varepsilon, -1+\varepsilon)$. For $x \in (-1-\varepsilon, -1)$, we have $\what f_{-1,+}(x)^{\beta} = e^{2 \beta \pi \ii}\what f_{-1,-}(x)^{\beta}$, it then follows from \eqref{eq:hatN-jump} that
\begin{equation}
\what E_{-1,-}(x)^{-1}\what E_{-1,+}(x) = I.
\end{equation}
Similarly, for $x \in (-1, -1+\varepsilon)$, we also have
\begin{equation}
\what E_{-1,-}(x)^{-1}\what E_{-1,+}(x) = I.
\end{equation}
Moreover, as $z \to -1$, applying \eqref{eq:d1-1}, \eqref{eq:d2-1} and \eqref{eq:hatf-1}, we have
\begin{equation}
\what E_{-1}(z) = \what E_{-1}(-1)\left(I+ \what E_{-1}(-1)^{-1} \what E_{-1}'(-1)(z+1) +\Boh\left((z+1)^2\right)\right),
\end{equation}
where
\begin{equation}\label{eq:hatE-1-1}
\what E_{-1}(-1) = \frac{1}{\sqrt{2}} (I + 2 \beta E_{3,1}-2 \beta E_{4,2}) \begin{pmatrix}
e^{-\frac{\beta \pi \ii}{2}} & 0 & -\ii e^{\frac{\beta \pi \ii}{2}} & 0\\
0 & e^{-\frac{\pi \ii}{4}}\textsf{b} & 0 & e^{\frac{\pi \ii}{4}}\textsf{b}^{-1}\\
-\ii e^{-\frac{\beta \pi \ii}{2}} & 0 & e^{\frac{\beta \pi \ii}{2}} & 0\\
0 & -e^{-\frac{\pi \ii}{4}}\textsf{b} & 0 & e^{\frac{\pi \ii}{4}}\textsf{b}^{-1}
\end{pmatrix}
\end{equation}
and
\begin{equation}\label{eq:hatE-1'-1}
\what E_{-1}(-1)^{-1}\what E_{-1}'(-1)=\begin{pmatrix}
-\frac{\ii \beta}{2} & 0 & -\frac{\ii}{4}e^{\beta \pi \ii} & 0\\
0 & \frac{\beta}{4}\left(2 + \frac{r_2+s_2/s}{r_2-s_2/s}\right) & 0 & \frac{\ii}{4\textsf{b}^2}\\
\frac{\ii}{4}e^{-\beta \pi \ii} & 0 & \frac{\ii \beta}{2} & 0\\
0 & -\frac{\ii \textsf{b}^2}{4} & 0 & -\frac{\beta}{4}\left(2 + \frac{r_2+s_2/s}{r_2-s_2/s}\right)
\end{pmatrix}
\end{equation}
with $\textsf{b} = e^{-\theta_{2,+}(-s)-\pi \ii \beta/2}\what f_{-1}'(-1)^{-\beta} 4^{-\beta} s^{-3\beta/2}$. Thus, the prefactor $\what E_{-1}(z)$ is indeed analytic in $D(-1, \varepsilon)$, the jump condition \eqref{eq:hatP-1-jump} is satisfied due to this fact and the jump condition of $\Phi^{(\CHF)}$ given in \eqref{HJumps}.

From the definitions of $\theta_1(z)$ and $\theta_2(z)$ in \eqref{def:theta1} and \eqref{def:theta2}, it is clear that functions $e^{-\theta_1(sz)+\theta_2(sz)+2 \tau sz}$ and $e^{-\theta_1(sz)+\theta_2(sz)-2 \tau sz}$ in \eqref{def:hatP-1} are exponentially small for $z \in D(-1, \varepsilon)$ with large positive $s$. As $ s \to +\infty$, the matching condition \eqref{eq:matchhatP-1} follows directly from \eqref{def:hatP-1}, \eqref{eq:hatE-1} and the asymptotic behavior of the confluent hypergeometric parametrix $\Phi^{(\CHF)}(z)$ at infinity in \eqref{H at infinity}.

This completes the proof of Proposition \ref{pro:hatP-1}.
\end{proof}

\subsection{Local parametrix near $z=1$}
Near the endpoint $z=1$, the local parametrix $\what P^{(1)}(z)$ reads as follows.

\begin{rhp}\label{rhp:hatP1}
\hfill
\begin{enumerate}
\item[\rm (a)] $\what P^{(1)}(z)$ is defined and analytic in $D(1, \varepsilon)\setminus \Gamma_{\what S}$, where $\Gamma_{\what S}$ is defined in \eqref{def:gammahatS}.

\item[\rm (b)] For $z \in D(1, \varepsilon) \cap \Gamma_{\what S}$, $\what P^{(1)}(z)$ satisfies the jump condition
\begin{equation}\label{eq:hatP1-jump}
 \what P^{(1)}_+(z)=\what P^{(1)}_-(z)J_{\what S}(z),
\end{equation}
where $J_{\what S}(z)$ is given in \eqref{def:JhatS}.

\item[\rm (c)]As $s \to \infty$,  we have the matching condition
\begin{equation}\label{eq:matchhatP1}
\what P^{(1)}(z)=\left( I+ \Boh(s^{-\frac 32}) \right) \what N(z), \quad z \in \partial D(1, \varepsilon),
\end{equation}
where $\what N(z)$ is given in \eqref{eq:hatNexp}.
\end{enumerate}
\end{rhp}

Again, the above RH problem can be solved by using the confluent hypergeometric parametrix $\Phi^{(\CHF)}$. To do this, we introduce the function
\begin{equation}\label{def:hatf1}
\what f_{1}(z) = -2 \ii s^{-\frac 32} \begin{cases}
\theta_1(sz) - \theta_{1,+}(s), & \quad\Im{z} > 0,\\
-\theta_1(sz) + \theta_{1,-}(s), & \quad\Im{z} < 0.
\end{cases}
\end{equation}
By \eqref{def:theta1}, it is easily seen that
\begin{equation}\label{eq:hatf1}
\what f_{1}(z) = 2\left(r_1 - \frac{s_1}{s}\right)(z-1) + \Boh((z-1)^2), \quad z \to 1.
\end{equation}
Hence, $\what f_{1}$ is a conformal mapping near $z=1$ for large positive $s$.

We then define
\begin{align}\label{def:hatP1}
&\what P^{(1)}(z) \nonumber
\\
& = \what E_{1}(z) \begin{pmatrix}
\Phi^{(\CHF)}_{11}(s^{\frac 32} \what f_{1}(z); \beta) & 0 & \Phi^{(\CHF)}_{12}(s^{\frac 32} \what f_{1}(z); \beta) & 0\\
0 & 1 &0&0\\
\Phi^{(\CHF)}_{21}(s^{\frac 32} \what f_{1}(z); \beta) & 0 & \Phi^{(\CHF)}_{22}(s^{\frac 32} \what f_{1}(z); \beta) & 0\\
0 & 0 & 0 & 1
\end{pmatrix}\nonumber\\
&\quad \times \diag \left(e^{\theta_1(sz) - \frac{\beta \pi \ii}{2}}, 1, e^{-\theta_1(sz) + \frac{\beta \pi \ii}{2}}, 1\right)\nonumber\\
&\quad \times \begin{cases}
I-e^{\theta_1(sz)-\theta_2(sz)-2 \tau sz} E_{2,1}+e^{\theta_1(sz)-\theta_2(sz)+2 \tau sz}E_{3,4}, &\quad z \in (\Omega_2\cup\Omega_5)\cap D(1, \varepsilon) \setminus \partial \Omega_{R,\pm},\\
I, & \quad z \in (\Omega_1\cup\Omega_6)\cap D(1, \varepsilon),
\end{cases}
\end{align}
where $\beta$ is given in \eqref{def:beta}, $\what f_{1}$ is defined in \eqref{def:hatf1} and
\begin{align}\label{def:hatE1}
&\what E_{1}(z) = \what N(z)\nonumber\\
 &\times \begin{cases}
\diag \left(e^{-\theta_{1,+}(s) + \frac{\beta \pi \ii}{2}}s^{\frac{3\beta}{2}} \what f_{1}(z)^{\beta}, 1, e^{\theta_{1,+}(s) -\frac{\beta \pi \ii}{2}}s^{-\frac{3\beta}{2}} \what f_{1}(z)^{-\beta}, 1\right), \hfill \Im{z}>0,\\
\begin{pmatrix}
0 & 0 & 1 & 0\\
0 & 1 & 0 & 0\\
-1 & 0 & 0 & 0\\
0 & 0 & 0 & 1
\end{pmatrix}\diag \left(e^{\theta_{1,-}(s) +\frac{\beta \pi \ii}{2}}s^{\frac{3\beta}{2}} \what f_{1}(z)^{\beta}, 1, e^{-\theta_{1,-}(s) -\frac{\beta \pi \ii}{2}}s^{-\frac{3\beta}{2}} \what f_{1}(z)^{-\beta},1\right), \\
\hfill \Im{z}<0.
\end{cases}
\end{align}

The following proposition can be proved in a manner similar to that of Proposition \ref{pro:hatP-1}, and we omit the details here.
\begin{proposition}
The local parametrix $\what P^{(1)}(z)$ defined in \eqref{def:hatP1} solves RH problem \ref{rhp:hatP1}.
\end{proposition}
For later use, we include the following asymptotic behavior of $\what E_{1}(z)$ near $z=1$ by applying \eqref{eq:d11}, \eqref{eq:d21} and \eqref{eq:hatf1}:
\begin{equation}
\what E_{1}(z) = \what E_{1}(1)\left(I + \what E_{1}(1)^{-1}\what E_{1}'(1)(z-1) +\Boh\left((z-1)^2\right)\right), \qquad z\to 1,
\end{equation}
where
\begin{equation}\label{eq:hatE11}
\what E_{1}(1) =  \frac{1}{\sqrt{2}} (I + 2 \beta E_{3,1}-2 \beta E_{4,2}) \begin{pmatrix}
e^{\frac{\pi \ii}{4}}\textsf{a} & 0 & e^{\frac{-\pi \ii}{4}}\textsf{a}^{-1}& 0\\
0 & e^{-\frac{\beta \pi \ii}{2}} & 0 & \ii e^{\frac{\beta \pi \ii}{2}}\\
-e^{\frac{\pi \ii}{4}}\textsf{a} & 0 & e^{\frac{-\pi \ii}{4}}\textsf{a}^{-1} & 0\\
0 & \ii e^{-\frac{\beta \pi \ii}{2}} & 0 &  e^{\frac{\beta \pi \ii}{2}}
\end{pmatrix}
\end{equation}
and
\begin{equation}\label{eq:hatE1'1}
\what E_{1}(1)^{-1}\what E_{1}'(1)=\begin{pmatrix}
\frac{\beta}{4}\left(2 + \frac{r_1+s_1/s}{r_1-s_1/s}\right) & 0 & \frac{\ii}{4\textsf{a}^2} & 0\\
0 & \frac{\ii \beta}{2} & 0 & -\frac{\ii}{4}e^{\beta \pi \ii}\\
-\frac{\ii \textsf{a}^2}{4} & 0 & -\frac{\beta}{4}\left(2 + \frac{r_1+s_1/s}{r_1-s_1/s}\right) & 0\\
0 & \frac{\ii}{4}e^{-\beta \pi \ii} & 0 & -\frac{\ii \beta}{2}
\end{pmatrix}
\end{equation}
with
\begin{align}\label{def:a}
\textsf{a} = e^{-\theta_{1,+}(s)+\pi \ii \beta/2}\what f_{1}'(1)^{-\beta} 4^{\beta} s^{3 \beta /2}.
\end{align}
\subsection{Final transformation}

We define the following final transformation
\begin{equation}\label{def:hatR}
\what R(z) = \begin{cases}
\what S(z) \what P^{(0)}(z)^{-1}, & \quad z \in D(0, \varepsilon),\\
\what S(z) \what P^{(-1)}(z)^{-1}, & \quad z \in D(-1, \varepsilon),\\
\what S(z) \what P^{(1)}(z)^{-1}, & \quad z \in D(1, \varepsilon),\\
\what S(z) \what N(z)^{-1}, & \quad \textrm{elsewhere}.
\end{cases}
\end{equation}
From the RH problems for $\what S$, $\what N$, $\what P^{(0)}$ and $\what P^{(\pm 1)}$, it follows that $\what R$ satisfies the following RH problem.
\begin{figure}[t]
\begin{center}

\tikzset{every picture/.style={line width=0.75pt}} 

\begin{tikzpicture}[x=0.75pt,y=0.75pt,yscale=-1,xscale=1]

\draw   (101,146) .. controls (101,132.19) and (112.19,121) .. (126,121) .. controls (139.81,121) and (151,132.19) .. (151,146) .. controls (151,159.81) and (139.81,171) .. (126,171) .. controls (112.19,171) and (101,159.81) .. (101,146) -- cycle ;
\draw  [fill={rgb, 255:red, 0; green, 0; blue, 0 }  ,fill opacity=1 ] (125,145) .. controls (125,144.45) and (125.45,144) .. (126,144) .. controls (126.55,144) and (127,144.45) .. (127,145) .. controls (127,145.55) and (126.55,146) .. (126,146) .. controls (125.45,146) and (125,145.55) .. (125,145) -- cycle ;

\draw   (190,146) .. controls (190,132.19) and (201.19,121) .. (215,121) .. controls (228.81,121) and (240,132.19) .. (240,146) .. controls (240,159.81) and (228.81,171) .. (215,171) .. controls (201.19,171) and (190,159.81) .. (190,146) -- cycle ;
\draw  [fill={rgb, 255:red, 0; green, 0; blue, 0 }  ,fill opacity=1 ] (214,145) .. controls (214,144.45) and (214.45,144) .. (215,144) .. controls (215.55,144) and (216,144.45) .. (216,145) .. controls (216,145.55) and (215.55,146) .. (215,146) .. controls (214.45,146) and (214,145.55) .. (214,145) -- cycle ;

\draw   (280,145) .. controls (280,131.19) and (291.19,120) .. (305,120) .. controls (318.81,120) and (330,131.19) .. (330,145) .. controls (330,158.81) and (318.81,170) .. (305,170) .. controls (291.19,170) and (280,158.81) .. (280,145) -- cycle ;
\draw  [fill={rgb, 255:red, 0; green, 0; blue, 0 }  ,fill opacity=1 ] (304,144) .. controls (304,143.45) and (304.45,143) .. (305,143) .. controls (305.55,143) and (306,143.45) .. (306,144) .. controls (306,144.55) and (305.55,145) .. (305,145) .. controls (304.45,145) and (304,144.55) .. (304,144) -- cycle ;

\draw    (40.58,58.92) -- (110,126) ;
\draw    (320.58,164.92) -- (390,232) ;
\draw    (40,231) -- (110,166) ;
\draw    (320,125) -- (390,60) ;
\draw  [fill={rgb, 255:red, 0; green, 0; blue, 0 }  ,fill opacity=1 ] (81.11,97.8) -- (75.5,96.45) -- (79.43,92.28) -- cycle ;
\draw  [fill={rgb, 255:red, 0; green, 0; blue, 0 }  ,fill opacity=1 ] (81.19,192.46) -- (79.24,197.88) -- (75.52,193.52) -- cycle ;
\draw  [fill={rgb, 255:red, 0; green, 0; blue, 0 }  ,fill opacity=1 ] (355.29,198.46) -- (349.68,197.11) -- (353.61,192.94) -- cycle ;
\draw  [fill={rgb, 255:red, 0; green, 0; blue, 0 }  ,fill opacity=1 ] (355,92.5) -- (353.05,97.92) -- (349.33,93.56) -- cycle ;
\draw  [draw opacity=0] (136.42,123.12) .. controls (140.68,114.43) and (154.1,108.08) .. (170,108.08) .. controls (185.9,108.08) and (199.32,114.43) .. (203.58,123.12) -- (170,129.04) -- cycle ; \draw   (136.42,123.12) .. controls (140.68,114.43) and (154.1,108.08) .. (170,108.08) .. controls (185.9,108.08) and (199.32,114.43) .. (203.58,123.12) ;
\draw  [fill={rgb, 255:red, 0; green, 0; blue, 0 }  ,fill opacity=1 ] (171.79,108.08) -- (166.78,110.95) -- (166.78,105.21) -- cycle ;

\draw  [draw opacity=0] (226.42,123.12) .. controls (230.68,114.43) and (244.1,108.08) .. (260,108.08) .. controls (275.9,108.08) and (289.32,114.43) .. (293.58,123.12) -- (260,129.04) -- cycle ; \draw   (226.42,123.12) .. controls (230.68,114.43) and (244.1,108.08) .. (260,108.08) .. controls (275.9,108.08) and (289.32,114.43) .. (293.58,123.12) ;
\draw  [fill={rgb, 255:red, 0; green, 0; blue, 0 }  ,fill opacity=1 ] (261.79,108.08) -- (256.78,110.95) -- (256.78,105.21) -- cycle ;

\draw  [draw opacity=0] (203.58,168.1) .. controls (199.32,176.79) and (185.9,183.13) .. (170,183.13) .. controls (154.1,183.13) and (140.68,176.79) .. (136.42,168.1) -- (170,162.17) -- cycle ; \draw   (203.58,168.1) .. controls (199.32,176.79) and (185.9,183.13) .. (170,183.13) .. controls (154.1,183.13) and (140.68,176.79) .. (136.42,168.1) ;
\draw  [fill={rgb, 255:red, 0; green, 0; blue, 0 }  ,fill opacity=1 ] (173.22,183.13) -- (168.21,186) -- (168.21,180.27) -- cycle ;
\draw  [draw opacity=0] (294.58,168.1) .. controls (290.32,176.79) and (276.9,183.13) .. (261,183.13) .. controls (245.1,183.13) and (231.68,176.79) .. (227.42,168.1) -- (261,162.17) -- cycle ; \draw   (294.58,168.1) .. controls (290.32,176.79) and (276.9,183.13) .. (261,183.13) .. controls (245.1,183.13) and (231.68,176.79) .. (227.42,168.1) ;
\draw  [fill={rgb, 255:red, 0; green, 0; blue, 0 }  ,fill opacity=1 ] (264.22,183.13) -- (259.21,186) -- (259.21,180.27) -- cycle ;

\draw  [fill={rgb, 255:red, 0; green, 0; blue, 0 }  ,fill opacity=1 ] (128.5,121) -- (123.5,123.87) -- (123.5,118.13) -- cycle ;
\draw  [fill={rgb, 255:red, 0; green, 0; blue, 0 }  ,fill opacity=1 ] (217.5,121) -- (212.5,123.87) -- (212.5,118.13) -- cycle ;
\draw  [fill={rgb, 255:red, 0; green, 0; blue, 0 }  ,fill opacity=1 ] (308,120) -- (303,122.87) -- (303,117.13) -- cycle ;

\draw (116,145) node [anchor=north west][inner sep=0.75pt]   [align=left] {$-1$};
\draw (210,148) node [anchor=north west][inner sep=0.75pt]   [align=left] {$0$};
\draw (301,147) node [anchor=north west][inner sep=0.75pt]   [align=left] {$1$};
\draw (15,49) node [anchor=north west][inner sep=0.75pt]   [align=left] {$\Gamma_2^{(1)}$};
\draw (15,225) node [anchor=north west][inner sep=0.75pt]   [align=left] {$\Gamma_4^{(1)}$};
\draw (392,49) node [anchor=north west][inner sep=0.75pt]   [align=left] {$\Gamma_1^{(1)}$};
\draw (391,225) node [anchor=north west][inner sep=0.75pt]   [align=left] {$\Gamma_5^{(1)}$};
\draw (155,90) node [anchor=north west][inner sep=0.75pt]   [align=left] {$\partial \Omega_{L,+}$};
\draw (155,190) node [anchor=north west][inner sep=0.75pt]   [align=left] {$\partial \Omega_{L,-}$};
\draw (250,90) node [anchor=north west][inner sep=0.75pt]   [align=left] {$\partial \Omega_{R,+}$};
\draw (250,190) node [anchor=north west][inner sep=0.75pt]   [align=left] {$\partial \Omega_{R,-}$};
\draw (35,140) node [anchor=north west][inner sep=0.75pt]   [align=left] {$\partial D(-1,\varepsilon)$};
\draw (190,177) node [anchor=north west][inner sep=0.75pt]   [align=left] {$\partial D(0,\varepsilon)$};
\draw (335,140) node [anchor=north west][inner sep=0.75pt]   [align=left] {$\partial D(1,\varepsilon)$};

\end{tikzpicture}

   \caption{The contour $\Gamma_{\what R}$ in the RH problem for $\what R$.}
   \label{fig:hatR}
\end{center}
\end{figure}
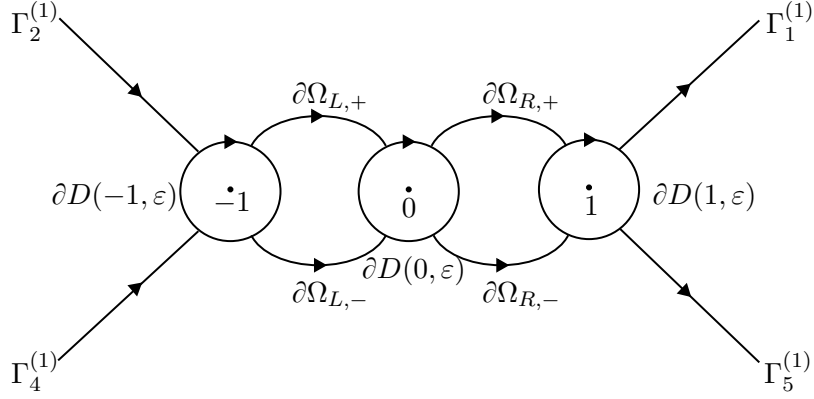

\begin{rhp}
\hfill
\begin{itemize}
\item [\rm{(a)}] $\what R(z)$ is defined and analytic in $\mathbb{C} \setminus \Gamma_{\what R}$, where
\begin{equation}
\Sigma_{\what R}:=\Gamma_{\what S} \cup \partial D(-1,\varepsilon) \cup \partial D(0,\varepsilon) \cup \partial D(1,\varepsilon) \setminus \{\mathbb{R}
\cup D(-1,\varepsilon) \cup D(0,\varepsilon) \cup D(1,\varepsilon) \};
\end{equation}
see Figure \ref{fig:hatR} for an illustration.
\item [\rm{(b)}] For $z \in \Gamma_{\what R}$, we have
\begin{equation}\label{eq:hatRjump}
\what R_+(z) = \what R_-(z) J_{\what R} (z),
\end{equation}
 where
\begin{equation}\label{def:JhatR}
J_{\what R}(z) = \begin{cases}
\what P^{(0)}(z) \what N(z)^{-1}, & \quad z \in \partial D(0, \varepsilon),\\
\what P^{(-1)}(z) \what N(z)^{-1}, & \quad z \in \partial D(-1, \varepsilon),\\
\what P^{(1)}(z) \what N(z)^{-1}, & \quad z \in \partial D(1, \varepsilon),\\
\what N(z) J_{\what S}(z) \what N(z)^{-1}, & \quad \Gamma_{\what R} \setminus \left\{\partial D(0, \varepsilon) \cup \partial D(\pm 1, \varepsilon) \right\},
\end{cases}
\end{equation}
with $J_{\what S}(z)$ defined in \eqref{def:JhatS}.
\item [\rm{(c)}] As $z \to \infty$, we have
\begin{equation}\label{eq:asy:hatR}
\what R(z) = I + \frac{\what R^{(1)}}{z} + \Boh (z^{-2}),
\end{equation}
where $\what R^{(1)}$ is independent of $z$.
\end{itemize}
\end{rhp}
As $s\to +\infty$, we have the following estimate of $J_{\what R}(z)$ in \eqref{def:JhatR}. For $z \in \Gamma_{\what R} \setminus \left\{\partial D(0, \varepsilon) \cup \partial D(\pm 1, \varepsilon) \right\}$, it is readily seen from \eqref{def:JhatS} and \eqref{eq:hatNexp} that there exists a positive constant $c$ such that
\begin{equation}\label{eq:estJhatR1}
J_{\what R}(z) = I + \Boh\left(e^{-c s^{3/2}}\right),
\end{equation}
for $z \in \partial D(\pm 1, \varepsilon)$, it follows from \eqref{eq:matchhatP-1} and \eqref{eq:matchhatP1} that
\begin{equation}
J_{\what R}(z) = I + \Boh\left(s^{-\frac{3}{2}}\right),
\end{equation}
and for $z \in D(0, \varepsilon)$, it follows from \eqref{eq:match01} that
\begin{equation}\label{eq:estJhatR3}
J_{\what R}(z) = I + \frac{\what J^{(0)}_1(z)}{s^{1/2}} + \frac{\what J^{(0)}_2(z)}{s}+\Boh\left(s^{-\frac{3}{2}}\right),
\end{equation}
where $\what J^{(0)}_1(z)$ and $\what J^{(0)}_2(z)$ are given in \eqref{def:hatJ01} and \eqref{def:hatJ02}.

By \cite{Deift1999, Deift1993}, the estimates \eqref{eq:estJhatR1}--\eqref{eq:estJhatR3} imply that
\begin{equation}\label{eq:hatRexpansion}
\what R(z) = I + \frac{\what R_1(z)}{s^{1/2}} + \frac{\what R_2(z)}{s}+\Boh (s^{-\frac{3}{2}}), \qquad s \to +\infty,
\end{equation}
uniformly for $z \in \mathbb{C} \setminus \Gamma_{\what R}$. Moreover, by inserting \eqref{eq:hatRexpansion} into \eqref{eq:hatRexpansion}, it follows from \eqref{eq:estJhatR1}--\eqref{eq:estJhatR3} that $\what R_1$ satisfies the following RH problem.
\begin{rhp}
\hfill
\begin{itemize}
\item [\rm{(a)}] $\what R_1(z)$ is analytic in $\mathbb{C} \setminus \partial D(0, \varepsilon)$.
\item [\rm{(b)}] For $z \in \partial D(0, \varepsilon)$, we have
\begin{equation}
\what R_{1,+}(z)-\what R_{1,-}(z)=\what J^{(0)}_1(z),
\end{equation}
 where $\what J^{(0)}_1(z)$ is given in \eqref{def:hatJ01}.
\item [\rm{(c)}] As $z \to \infty$, we have
\begin{equation}
\what R_1(z) = \Boh (z^{-1}).
\end{equation}
\end{itemize}
\end{rhp}
By Cauchy's residue theorem, we have
\begin{align}\label{eq:hatR1}
\what R_1(z) = \frac{1}{2 \pi \ii} \int_{\partial D(0, \varepsilon)} \frac{\what J^{(0)}_1(\zeta)}{\zeta -z} \ud \zeta
=\begin{cases}
\frac{\Res_{\zeta =0}\what J^{(0)}_1(\zeta)}{z} - \what J^{(0)}_1(z), & \quad z \in D(0, \varepsilon),\\
\frac{\Res_{\zeta =0}\what J^{(0)}_1(\zeta)}{z}, & \quad \textrm{elsewhere}.
\end{cases}
\end{align}
Similarly, we have that $\what R_2$ in \eqref{eq:hatRexpansion} satisfies the following RH problem.
\begin{rhp}
\hfill
\begin{itemize}
\item [\rm{(a)}] $\what R_2(z)$ is analytic in $\mathbb{C} \setminus \partial D(0, \varepsilon)$.
\item [\rm{(b)}] For $z \in \partial D(0, \varepsilon)$, we have
\begin{equation}
\what R_{2,+}(z)-\what R_{2,-}(z)=\what R_{1,-}(z) \what J^{(0)}_1(z)+\what J^{(0)}_2(z),
\end{equation}
 where $\what J^{(0)}_1(z)$ and $\what J^{(0)}_2(z)$ are given in \eqref{def:hatJ01} and \eqref{def:hatJ02}, respectively.
\item [\rm{(c)}] As $z \to \infty$, we have
\begin{equation}
\what R_2(z) = \Boh (z^{-1}).
\end{equation}
\end{itemize}
\end{rhp}
By Cauchy's residue theorem, we have
\begin{align}\label{eq:hatR2}
\what R_2(z) &= \frac{1}{2 \pi \ii} \int_{\partial D(0, \varepsilon)} \frac{\what R_{1,-}(\zeta) \what J^{(0)}_1(\zeta)+\what J^{(0)}_2(\zeta)}{\zeta -z} \ud \zeta
\nonumber
\\
&=\begin{cases}
\frac{\Res_{\zeta =0}(\what R_{1,-}(\zeta) \what J^{(0)}_1(\zeta)+\what J^{(0)}_2(\zeta))}{z} - \what R_{1,-}(z) \what J^{(0)}_1(z)-\what J^{(0)}_2(z), & \quad z \in D(0, \varepsilon),\\
\frac{\Res_{\zeta =0}(\what R_{1,-}(\zeta) \what J^{(0)}_1(\zeta)+\what J^{(0)}_2(\zeta))}{z}, & \quad \textrm{elsewhere}.
\end{cases}
\end{align}

\section{Asymptotic analysis of the RH problem for $X$ as $s \to 0^+$}\label{sec:AsyX0}

In this section, we analyze the asymptotics for $X$ as $s \to 0^+$, which is relatively simpler than the case when $s \to +\infty$. Throughout this section, it is assumed that $0 < \gamma \le 1$.

\subsection{Global parametrix}
As $s \to 0^+$, the interval $(-s, s)$ vanishes, it is then easily seen that the RH problem for $X$ is approximated by following global parametrix $\widecheck{N}$ for $|z|>\delta > s$:
\begin{equation}\label{eq:tildeN}
\widecheck{N}(z) = M(z) \begin{cases}
J_1(z), \qquad & \textrm{$\arg z < \varphi$ and $\arg (z-s) > \varphi$,} \\
J_5(z)^{-1}, \qquad & \textrm{$\arg z >- \varphi$ and $\arg (z-s) <- \varphi$,} \\
J_2(z), \qquad & \textrm{$\arg z > \pi - \varphi$ and $\arg (z+s) < \pi -\varphi$,}\\
J_4(z)^{-1}, \qquad & \textrm{$\arg z <  \varphi-\pi$ and $\arg (z+s) >\varphi-\pi$,}\\
I, \qquad & \textrm{elsewhere,}
\end{cases}
\end{equation}
where $M$ is the solution of the tacnode RH problem \ref{rhp:tac}, $\varphi$ is given in \eqref{phi} and $J_k$, $k=0,\ldots, 5$, denotes the jump matrix of $M$ on the ray $\Gamma_k$ as shown in Figure \ref{fig:tacnode}.

\subsection{Local parametrix}
For $|z|<\delta$, which particularly includes a neighborhood of $(-s, s)$, we approximate $X$ by the following local parametrix.
\begin{rhp}\label{rhp:tildeP0}
\hfill
\begin{itemize}
\item [\rm{(a)}] $\widecheck P^{(0)}(z)$ is defined and analytic in $D(0, \delta) \setminus \Gamma_{X}$, where $\Gamma_{X}$ is defined in \eqref{def:gammaX}.
\item [\rm{(b)}] For $z \in D(0, \delta) \cap \Gamma_{X}$, $\widecheck P^{(0)}$ satisfies the jump condition
\begin{equation}\label{eq:tildeP0-jump}
\widecheck P^{(0)}_+(z) = \widecheck P^{(0)}_-(z) J_{X}(z),
\end{equation}
where $J_{X}(z)$ is given in \eqref{def:JX}.
\item [\rm{(c)}] As $s \to 0^+$, we have the matching condition
\begin{equation}\label{eq:tildeP0:matching}
\widecheck P^{(0)}(z) = (I+ \Boh(s)) \widecheck{N}(z), \qquad z \in \partial D(0, \delta),
\end{equation}
where $\widecheck{N}(z)$ is given in \eqref{eq:tildeN}.
\end{itemize}
\end{rhp}
Recall that $\widehat M $ is the analytic continuation of the restriction of $M$ in the sector bounded by the rays $\Gamma_1$ and $\Gamma_2$ to the whole complex plane and $\Omega_k^{(s)}$, $k=1,\ldots,6$, are six regions as shown in Figure \ref{fig:X}, we look for a solution to the above RH problem of the following form:
\begin{equation}\label{eq:tildeP0}
\widecheck P^{(0)}(z) = \what M(z) \begin{pmatrix}
1 & 0 & \eta \left(\frac{z}{s}\right) & \eta \left(\frac{z}{s}\right)\\
0 & 1 & \eta \left(\frac{z}{s}\right) & \eta \left(\frac{z}{s}\right)\\
0 & 0 & 1 & 0\\
0 & 0 & 0 & 1
\end{pmatrix}
\begin{cases}
J_1(z)^{-1}, \quad & z \in \Omega_1^{(s)},\\
I, \quad & z \in \Omega_2^{(s)},\\
J_2(z)^{-1}, \quad & z \in \Omega_3^{(s)},\\
J_1(z)^{-1}J_0(z)^{-1}J_5(z)^{-1}J_4(z), \quad & z \in \Omega_4^{(s)},\\
J_1(z)^{-1}J_0(z)^{-1}J_5(z)^{-1}, \quad & z \in \Omega_5^{(s)},\\
J_1^{-1}(z)J_0^{-1}(z), \quad & z \in \Omega_6^{(s)},
\end{cases}
\end{equation}
where $\eta $ is a function to be determined later. In view of RH problem \ref{rhp:tildeP0}, it follows that $\eta$ solves the following scalar RH problem.
\begin{rhp}
\hfill
\begin{itemize}
\item [\rm{(a)}] $\eta (z)$ is defined and analytic in $\mathbb{C} \setminus [-1, 1]$.
\item [\rm{(b)}] For $x \in (-1, 1)$, we have $\eta_+(x) = \eta_-(x) -\gamma$.
\item [\rm{(c)}] As $z \to \infty$, we have
\begin{equation}
\eta(z) = \Boh (z^{-1}).
\end{equation}
\end{itemize}
\end{rhp}
By the Sokhotske-Plemelj formula, it is easy to find
\begin{equation}\label{def:eta}
\eta (z) = - \frac{\gamma}{2 \pi \ii} \ln{\left(\frac{z-1}{z+1}\right)}.
\end{equation}

Since $\eta(z/s) = \Boh(s)$ as $ s \to 0^+$ for $z \in \partial D(0, \delta)$, we deduce the matching condition \eqref{eq:tildeP0:matching} from \eqref{eq:tildeN}, \eqref{eq:tildeP0} and the fact that $\what M$ is bounded near the origin.

\subsection{Final transformation}
We define the final transformation as
\begin{equation}\label{def:tildeR}
\widecheck R(z) = \begin{cases}
X(z) \widecheck{N}(z)^{-1}, \qquad & z \in \mathbb{C} \setminus D(0, \delta),\\
X(z) \widecheck P^{(0)}(z)^{-1}, \qquad & z \in D(0, \delta).
\end{cases}
\end{equation}
From the RH problems for $X$ and $\widecheck P^{(0)}$, it is readily seen $\widecheck R$ is analytic in $D(0, \delta) \setminus \{-s, s\}$.
On account of the fact that
\begin{equation}
\eta (z/s) = \begin{cases}
\Boh(\ln{(z+s)}), & \quad z \to -s\\
\Boh(\ln{(z-s)}), & \quad z \to s,
\end{cases}
\end{equation}
we conclude from \eqref{eq:X-near-s}, \eqref{eq:X-near--s} and \eqref{eq:tildeP0} that both $s$ and $-s$ are removable singularities. Moreover, it is readily seen that $\widecheck R$ solves the following RH problem.
\begin{rhp}
\hfill
\begin{itemize}
\item [\rm{(a)}] $\widecheck R(z)$ is defined and analytic in $\mathbb{C} \setminus \partial D(0, \delta)$.
\item [\rm{(b)}] For $z \in \partial D(0, \delta)$, we have
\begin{equation}
\widecheck R_+(z) = \widecheck R_+(z) J_{\widetilde{R}}(z),
\end{equation}
 where
    \begin{equation}
        J_{\widecheck{R}}(z) = \widecheck P^{(0)}(z) \widecheck{N}(z)^{-1}
    \end{equation}
\item [\rm{(c)}] As $z \to \infty$, we have
\begin{equation}
\widecheck R(z) = I + \Boh(z^{-1}).
\end{equation}
\end{itemize}
\end{rhp}

According to \eqref{eq:tildeP0:matching}, it is easily seen that $J_{\widetilde{R}}(z) =I+ \Boh(s)$, as $s \to 0^+$. Therefore, we have
\begin{equation}\label{eq:asytildeR}
\widecheck R(z)=I+ \Boh(s), \qquad \frac{\ud}{\ud z}\widecheck R(z)=\Boh(s), \qquad s \to 0^+,
\end{equation}
uniformly for $z \in \mathbb{C} \setminus \partial D(0, \delta)$.


\section{Asymptotics of $p_k(s)$ and $q_k(s)$ for large and small $s$} \label{sec:pq}
We have defined the functions $p_k(s)$ and $q_k(s)$, $k=1\ldots,6$, in \eqref{def:p56}, \eqref{def:q56} and \eqref{def:qkpk}, which satisfies the equations  \eqref{def:pq's} and \eqref{eq:sumpq}. It is the aim of this section to derive large and small $s$ asymptotics of these functions for
$0< \gamma < 1$. As we will see later, these asymptotic formulas are essential in the proof of large gap asymptotics of $F(s;\gamma)$.


\begin{proposition}\label{th:pq}
For the purely imaginary parameter $\beta$ given in \eqref{def:beta}, there exist a family of special solutions to the system of differential equations \eqref{def:pq's} and \eqref{eq:sumpq} with the following asymptotic behaviors:  as $s \to +\infty$,
\begin{align}
p_1(s) &= e^{-\tau s}\frac{\gamma |\Gamma(1-\beta)| s^{1/4}}{\sqrt{2} \pi}e^{-\frac{\beta \pi \ii}{2}}\left( \sin(\vartheta(s)-\frac{\pi}{4}) - 2 \ii \beta \cos(\vartheta(s)-\frac{\pi}{4})\right)(1+\Boh(s^{-\frac 12})),\label{p1}\\
p_2(s) &= e^{-\tau s}\frac{\sqrt{2} \gamma \beta |\Gamma(1-\beta)| M^{(1)}_{14}}{\pi s^{1/4}} e^{-\frac{\beta \pi \ii}{2}}\sin(\vartheta(s)-\frac{\pi}{4})(1+\Boh(s^{-\frac 12})),\label{p2}\\
p_3(s) &= e^{-\tau s}\frac{\ii \gamma |\Gamma(1-\beta)| }{\sqrt{2} \pi s^{1/4}} e^{-\frac{\beta \pi \ii}{2}}\cos(\vartheta(s)-\frac{\pi}{4})(1+\Boh(s^{-\frac 12})),\label{p3}\\
p_4(s) &= e^{-\tau s} \times \Boh(s^{-\frac 34}),\label{p4}\\
p_5(s) &= \ii r_1\left(-2 \beta s^{\frac 12} + M^{(1)}_{13}+2\beta s^{-\frac 12} \left(M^{(1)}_{11} -\left(M^{(1)}_{13}\right)^2-M^{(1)}_{14}M^{(1)}_{23}-M^{(1)}_{33}\right)\right)
\nonumber
\\
&\quad +\Boh(s^{-1}),\label{p5}\\
p_6(s)&=\ii r_1 \left(2 \beta \dot M^{(1)}_{14} s^{\frac 12} + \dot M^{(1)}_{12}-4\beta^2\left(\dot M^{(1)}_{12}+\dot M^{(1)}_{13}\dot M^{(1)}_{14}+\dot M^{(1)}_{14}\dot M^{(1)}_{24}+\dot M^{(1)}_{34}\right)\right)\nonumber\\
&\quad+ \ii r_2 \left(2 \beta \widetilde M^{(1)}_{14} s^{\frac 12} + \widetilde M^{(1)}_{12}-4\beta^2\left(\widetilde M^{(1)}_{12}+\widetilde M^{(1)}_{13}\widetilde M^{(1)}_{14}+\widetilde M^{(1)}_{14}\widetilde M^{(1)}_{24}+\widetilde M^{(1)}_{34}\right)\right)
\nonumber
\\
&\quad +\Boh(s^{-\frac 12}),\label{p6}\\
q_1(s) &=\sqrt{2} e^{\tau s}|\Gamma(1-\beta)|e^{-\frac{\beta \pi \ii}{2}}\cos(\vartheta(s)-\frac{\pi}{4})s^{-\frac 14}(1+\Boh(s^{-\frac 12})),\label{q1}\\
q_2(s) & = e^{\tau s}\times \Boh(s^{-\frac 34}),\label{q2}\\
q_3(s) & = \sqrt{2} e^{\tau s}|\Gamma(1-\beta)|e^{-\frac{\beta \pi \ii}{2}}s^{\frac 14}\left(2\beta\cos(\vartheta(s)-\frac{\pi}{4})+\ii\sin(\vartheta(s)-\frac{\pi}{4})\right)(1+\Boh(s^{-\frac 12})),\label{q3}\\
q_4(s) & = -2\sqrt{2} \ii \beta e^{\tau s}|\Gamma(1-\beta)|M^{(1)}_{23}e^{-\frac{\beta \pi \ii}{2}}\sin(\vartheta(s)-\frac{\pi}{4})s^{-\frac 14}(1+\Boh(s^{-\frac 12})),\label{q4}\\
q_5(s) & =-4\beta^2 s + 2\beta M^{(1)}_{13} s^{\frac 12}- M^{(1)}_{11} + 4\beta^2\left(M^{(1)}_{11} -\left(M^{(1)}_{13}\right)^2-M^{(1)}_{14}M^{(1)}_{23}-M^{(1)}_{33}\right)\nonumber\\
&\quad +2\beta \dot M^{(1)}_{13} s^{\frac 12}- \dot M^{(1)}_{11} + 4\beta^2\left(\dot M^{(1)}_{11} -\left(\dot M^{(1)}_{13}\right)^2-\dot M^{(1)}_{14}\dot M^{(1)}_{23}-\dot M^{(1)}_{33}\right)+\Boh(s^{-\frac 12}),\label{q5}\\
q_6(s) &=M^{(1)}_{14}-2\beta s^{-\frac 12} \left(M^{(1)}_{12}+M^{(1)}_{13}M^{(1)}_{14}+M^{(1)}_{14}M^{(1)}_{24}+M^{(1)}_{34}\right)+\Boh(s^{-1});\label{q6}
\end{align}
as $s \to 0^+$
\begin{align}
p_1(s)&=\Boh(1), \quad p_2(s)=\Boh(1), \quad p_3(s)=\Boh(1), \quad p_4(s)=\Boh(1),\label{eq:pk0}\\
p_5(s)&=\ii r_1 M^{(1)}_{13}+\Boh(s), \label{eq:p50}\\
p_6(s)&= \ii r_1 \dot M^{(1)}_{12} + \ii r_2 \widetilde M^{(1)}_{12} +\Boh(s),\label{eq:p60}\\
q_1(s)&=\Boh(1), \quad q_2(s)=\Boh(1),\quad q_3(s)=\Boh(1),\quad q_4(s)=\Boh(1),\label{eq:qk0}\\
q_5(s)&=-\dot M^{(1)}_{11}-M^{(1)}_{11}+\Boh(s),\label{eq:q50}\\
q_6(s)&= M^{(1)}_{14}+\Boh(s)\label{eq:q60},
\end{align}
where $\Gamma (z)$ is Euler's gamma function, $\vartheta (s)$ and $M^{(1)}$ are given in \eqref{def:theta} and \eqref{eq:asy:M}, $\widetilde M^{(1)}$ and $\dot M^{(1)}$ are defined through \eqref{def:tildeX} and \eqref{def:dotX}.
\end{proposition}
\begin{proof}
We split the proof into two parts, which deal with the large and small $s$ asymptotics, respectively.

\subsubsection*{Asymptotics of $p_k(s)$ and $q_k(s)$ as $s \to +\infty$}
We first consider the asymptotics of $p_k$ and $q_k$ for $k=1, \ldots 4$.
Recall that
\begin{equation}\label{pq2}
\begin{pmatrix}
q_1(s)\\q_2(s)\\q_3(s)\\q_4(s)
\end{pmatrix}=X_{R,0}(s) \begin{pmatrix}
1\\1\\0\\0
\end{pmatrix} \quad \textrm{and} \quad \begin{pmatrix}
p_1(s)\\p_2(s)\\p_3(s)\\p_4(s)
\end{pmatrix}=-\frac{\gamma}{2 \pi \ii}X_{R,0}(s)^{-\msf T} \begin{pmatrix}
0\\0\\1\\1
\end{pmatrix},
\end{equation}
where
\begin{align}\label{eq:XR0s}
X_{R,0}(s) = \lim_{z \to 1, z \in \Omega_2^{(1)}} X(sz)
\begin{pmatrix}
1 & 0 & \frac{\gamma}{2 \pi \ii} \ln(sz-s) & \frac{\gamma}{2 \pi \ii} \ln(sz-s)\\
0 & 1 & \frac{\gamma}{2 \pi \ii} \ln(sz-s) & \frac{\gamma}{2 \pi \ii} \ln(sz-s)\\
0 & 0 & 1 & 0\\
0 & 0 & 0 & 1
\end{pmatrix}.
\end{align}

Tracing back the transformations $X \to \what{T} \to \what S \to \what R$ in \eqref{def:XToTgamma}, \eqref{def:hatS} and \eqref{def:hatR}, it follows that, for $z \in \Omega_2^{(1)} \setminus \Omega_{R,+}$,
\begin{align}
X(sz) & = \diag \left(s^{-\frac 14}, s^{-\frac 14}, s^{\frac 14}, s^{\frac 14} \right)\what{T}(z) \diag \left(e^{-\theta_1(sz)+\tau sz}, e^{-\theta_2(sz)-\tau sz}, e^{\theta_1(sz)+\tau sz}, e^{\theta_2(sz)-\tau sz} \right)\nonumber\\
&=\diag \left(s^{-\frac 14}, s^{-\frac 14}, s^{\frac 14}, s^{\frac 14} \right)\what{S}(z) \diag \left(e^{-\theta_1(sz)+\tau sz}, e^{-\theta_2(sz)-\tau sz}, e^{\theta_1(sz)+\tau sz}, e^{\theta_2(sz)-\tau sz} \right)\nonumber\\
&=\diag \left(s^{-\frac 14}, s^{-\frac 14}, s^{\frac 14}, s^{\frac 14} \right)\what{R}(z) \what P^{(1)}(z)
\nonumber\\
&\quad \times \diag \left(e^{-\theta_1(sz)+\tau sz}, e^{-\theta_2(sz)-\tau sz}, e^{\theta_1(sz)+\tau sz}, e^{\theta_2(sz)-\tau sz} \right).
\end{align}
This, together with \eqref{eq:XR0s} and \eqref{def:hatP1}, implies that
\begin{align}
X_{R,0}(s) & = \diag \left(s^{-\frac 14}, s^{-\frac 14}, s^{\frac 14}, s^{\frac 14} \right)\what{R}(1) \what E_1(1)\nonumber\\
&\quad \times\lim_{z \to 1, z \in \Omega_2^{(1)}} \left[\begin{pmatrix}
\Phi^{(\CHF)}_{11}(s^{\frac 32} \what f_{1}(z); \beta) & 0 & \Phi^{(\CHF)}_{12}(s^{\frac 32} \what f_{1}(z); \beta) & 0\\
0 & 1 &0&0\\
\Phi^{(\CHF)}_{21}(s^{\frac 32} \what f_{1}(z); \beta) & 0 & \Phi^{(\CHF)}_{22}(s^{\frac 32} \what f_{1}(z); \beta) & 0\\
0 & 0 & 0 & 1
\end{pmatrix}\nonumber\right.\\
&\quad \times \diag \left(e^{\theta_1(sz) - \frac{\beta \pi \ii}{2}}, 1, e^{-\theta_1(sz) + \frac{\beta \pi \ii}{2}}, 1\right)\begin{pmatrix}
1 & 0 & 0 & 0\\
-e^{\theta_1(sz)-\theta_2(sz)-2 \tau sz} & 1 & 0 & 0\\
0 & 0 & 1 & e^{\theta_1(sz)-\theta_2(sz)+2 \tau sz}\\
0 & 0 & 0 & 1
\end{pmatrix}\nonumber\\
&\quad \times  \diag \left(e^{-\theta_1(sz)+\tau sz}, e^{-\theta_2(sz)-\tau sz}, e^{\theta_1(sz)+\tau sz}, e^{\theta_2(sz)-\tau sz} \right)\nonumber\\
&\quad \times \left.\begin{pmatrix}
1 & 0 & \frac{\gamma}{2 \pi \ii} \ln(sz-s) & \frac{\gamma}{2 \pi \ii} \ln(sz-s)\\
0 & 1 & \frac{\gamma}{2 \pi \ii} \ln(sz-s) & \frac{\gamma}{2 \pi \ii} \ln(sz-s)\\
0 & 0 & 1 & 0\\
0 & 0 & 0 & 1
\end{pmatrix}\right],
\end{align}
where $ \what E_1$ and $\what f_{1}$ are defined in \eqref{def:hatE1} and \eqref{def:hatf1}, respectively.  On account of the behavior of $\Phi^{(\CHF)}$ near the origin given in \eqref{eq:H-expand-2}, it is readily seen that
\begin{align}\label{eq:XR0}
&X_{R,0}(s) = \diag \left(s^{-\frac 14}, s^{-\frac 14}, s^{\frac 14}, s^{\frac 14} \right)\what{R}(1) \what E_1(1)\what{\Upsilon}_0\nonumber\\
& \times \begin{pmatrix}
e^{\tau s} & 0 & e^{\tau s} \frac{\gamma}{2 \pi \ii}\left(\ln s - \ln (e^{-\frac{\pi \ii}{2}} s^{\frac 32}\what f_{1}'(1))\right) & e^{\tau s} \frac{\gamma}{2 \pi \ii}\left(\ln s - \ln (e^{-\frac{\pi \ii}{2}} s^{\frac 32}\what f_{1}'(1))\right)\\
-e^{-\theta_2(s)-\tau s} & e^{-\theta_2(s)-\tau s} & 0 & 0\\
0 & 0 & e^{\tau s} & e^{\tau s}\\
0 & 0 & 0 & e^{\theta_2(s)-\tau s}
\end{pmatrix},
\end{align}
where
\begin{equation}\label{def:whatUpsilon}
\what{\Upsilon}_0 = \begin{pmatrix}
\left(\Upsilon_0\right)_{11} & 0 & \left(\Upsilon_0\right)_{12} & 0\\
0 & 1 & 0 & 0\\
\left(\Upsilon_0\right)_{21} & 0 & \left(\Upsilon_0\right)_{22} & 0\\
0 & 0 & 0 & 1
\end{pmatrix}
\end{equation}
with $\Upsilon_0$ defined in \eqref{eq:H-expand-coeff-0}.

Inserting \eqref{eq:XR0} into the first equation of \eqref{pq2}, it follows that
\begin{align}\label{def:q1234}
\begin{pmatrix}
q_1(s)\\q_2(s)\\q_3(s)\\q_4(s)
\end{pmatrix}=e^{\tau s}\diag \left(s^{-\frac 14}, s^{-\frac 14}, s^{\frac 14}, s^{\frac 14} \right)\what{R}(1) \what E_1(1)\what{\Upsilon}_0\begin{pmatrix}
1 \\0\\0\\0
\end{pmatrix}.
\end{align}
We note from \eqref{eq:hatRexpansion}, \eqref{eq:hatR1} and \eqref{def:hatJ01} that
\begin{align}\label{eq:hatR-1}
\what R(1) = I + s^{-\frac 12} (I+2\beta E_{3,1}-2\beta E_{4,2}) \begin{pmatrix}
0 & 0 &M^{(1)}_{13} &M^{(1)}_{14}\\
0 & 0 &M^{(1)}_{23} &M^{(1)}_{24}\\
0 & 0 & 0 & 0\\
0 & 0 & 0 &0
\end{pmatrix}(I+2\beta E_{3,1}-2\beta E_{4,2})^{-1} + \Boh(s^{-1}).
\end{align}
A combination of \eqref{eq:hatE11}, \eqref{def:whatUpsilon}, \eqref{eq:hatR-1} and \eqref{def:q1234} shows
\begin{align}\label{eq:q1intermi}
q_1(s) =\frac{1}{\sqrt{2}} e^{\tau s} s^{-\frac 14} \left(\textsf{a} e^{\frac{\pi \ii}{4}- \beta \pi \ii} \Gamma (1- \beta)+ \textsf{a}^{-1} e^{-\frac{\pi \ii}{4}} \Gamma (1+ \beta)\right)\left(I+\Boh(s^{-\frac 12})\right)
\end{align}
with $\textsf{a}$ defined in \eqref{def:a}. Since $\Re{\beta} =0$, the two term inside the first bracket of \eqref{eq:q1intermi} are complex conjugate of each other, which implies that
\begin{align}
q_1(s) = \sqrt{2} e^{\tau s} s^{-\frac 14} \Re\left(\textsf{a} e^{\frac{\pi \ii}{4}- \beta \pi \ii} \Gamma (1- \beta)\right)\left(I+\Boh(s^{-\frac 12})\right).
\end{align}
The asymptotic formula \eqref{q1} then follows directly from the above equation. The asymptotics of $q_k(s)$ for $k=2, 3, 4$, in \eqref{q2}--\eqref{q4} can be derived in a similar way.

Similarly, by inserting \eqref{eq:XR0} into the second equation of \eqref{pq2}, we have
\begin{align}
\begin{pmatrix}
p_1(s)\\p_2(s)\\p_3(s)\\p_4(s)
\end{pmatrix}=-\frac{\gamma}{2 \pi \ii} e^{-\tau s}\diag \left(s^{\frac 14}, s^{\frac 14}, s^{-\frac 14}, s^{-\frac 14} \right)\what{R}(1)^{-\msf T}\what E_1(1)^{-\msf T}\what{\Upsilon}_0^{-\msf T}\begin{pmatrix}
0\\0\\1\\0
\end{pmatrix}.
\end{align}
The asymptotic formulas \eqref{p1}--\eqref{p4} of $p_k(s)$, $k=1, 2, 3, 4$, then follows from \eqref{eq:hatE11}, \eqref{def:whatUpsilon}, \eqref{eq:hatR-1} and direct calculations.

To derive the asymptotics of $p_5(s)$, $p_6(s)$, $q_5(s)$ and $q_6(s)$ defined in \eqref{def:p56}--\eqref{def:q56}, we trace back the transformations $X \to \what{T} \to \what S \to \what R$ in \eqref{def:XToTgamma}, \eqref{def:hatS}, \eqref{def:hatR}, and obtain that for $z\in \mathbb{C} \setminus \{D(0,\varepsilon) \cup D(\pm 1,\varepsilon)\}$,
\begin{multline}
X(sz)=\diag{\left(s^{-\frac 14}, s^{-\frac 14}, s^{\frac 14}, s^{\frac 14}\right)} \what R(z) \what N(z)
\\
\times \diag{\left(e^{- \theta_1(sz) + \tau s z},e^{- \theta_2(sz) - \tau s z},e^{ \theta_1(sz) + \tau s z},e^{\theta_2(sz) - \tau s z}\right)}.
\end{multline}
Taking $z \to \infty$ and comparing the coefficient of $\Boh (1/z)$ term on both sides of the above formula, we have
\begin{equation}\label{def:X11}
X^{(1)} = s \diag{\left(s^{-\frac 14}, s^{-\frac 14}, s^{\frac 14}, s^{\frac 14}\right)} \left(\what R^{(1)}+\what N^{(1)}\right)\diag{\left(s^{\frac 14}, s^{\frac 14}, s^{-\frac 14}, s^{-\frac 14}\right)},
\end{equation}
where $\what N^{(1)}$ and $\what R^{(1)}$ are given in \eqref{eq:asyhatN} and \eqref{eq:asy:hatR}, respectively. In view of \eqref{def:hatN1exp}, it is readily seen that the first row of $\what N^{(1)}$ is $\left(2 \beta^2, 0, -2 \beta, 0\right)$. By \eqref{eq:asy:hatR}, \eqref{eq:hatRexpansion}, \eqref{eq:hatR1} and \eqref{eq:hatR2}, one has
\begin{align}
\what R^{(1)} = s^{-\frac 12} \Res_{\zeta =0}\what J^{(0)}_1(\zeta) + s^{-1} \Res_{\zeta =0}(\what R_{1,-}(\zeta) \what J^{(0)}_1(\zeta)+\what J^{(0)}_2(\zeta)) +\Boh(s^{-\frac 32}).
\end{align}
Substituting \eqref{eq:hatR1}, the expressions of $\what J^{(0)}_1$ and $\what J^{(0)}_2$ in \eqref{def:hatJ01} and \eqref{def:hatJ02} into the above formula gives us
\begin{align}
\what R^{(1)}_{11} & = -2\beta M^{(1)}_{13} s^{-\frac 12} + s^{-1}\left(M^{(1)}_{11}-4\beta^2\left(M^{(1)}_{11}-\left(M^{(1)}_{13}\right)^2-M^{(1)}_{14}M^{(1)}_{23}-M^{(1)}_{33}\right)\right)+\Boh (s^{-\frac 32} ),\label{eq:hatR111asy}\\
\what R^{(1)}_{12} & = 2\beta M^{(1)}_{14} s^{-\frac 12} + s^{-1}\left(M^{(1)}_{12}-4\beta^2\left(M^{(1)}_{12}+M^{(1)}_{13}M^{(1)}_{14}+M^{(1)}_{14}M^{(1)}_{24}+M^{(1)}_{34}\right)\right)+\Boh (s^{-\frac 32} ),\\
\what R^{(1)}_{13} & =  M^{(1)}_{13} s^{-\frac 12} + 2 \beta s^{-1}\left(M^{(1)}_{11}-\left(M^{(1)}_{13}\right)^2-M^{(1)}_{14}M^{(1)}_{23}-M^{(1)}_{33}\right)+\Boh (s^{-\frac 32} ),\\
\what R^{(1)}_{14} & = M^{(1)}_{14} s^{-\frac 12}  -2 \beta s^{-1}\left(M^{(1)}_{12}+M^{(1)}_{13}M^{(1)}_{14}+M^{(1)}_{14}M^{(1)}_{24}+M^{(1)}_{34}\right)+\Boh (s^{-\frac 32} ),\label{eq:hatR114asy}
\end{align}
where $M^{(1)}$ is given in \eqref{eq:asy:M}. Combining \eqref{def:X11} and \eqref{eq:hatR111asy}--\eqref{eq:hatR114asy}, it follows that
\begin{align}
X^{(1)}_{11} & = 2 \beta^2 s - 2 \beta M^{(1)}_{13} s^{\frac 12} + M^{(1)}_{11} - 4 \beta^2 \left(M^{(1)}_{11} -\left(M^{(1)}_{13}\right)^2-M^{(1)}_{14}M^{(1)}_{23}-M^{(1)}_{33}\right) + \Boh(s^{-\frac 12}), \label{eq:asyX111}\\
X^{(1)}_{12} & =  2 \beta M^{(1)}_{14} s^{\frac 12} + M^{(1)}_{12} - 4 \beta^2 \left(M^{(1)}_{12} +\left(M^{(1)}_{13}\right)^2+M^{(1)}_{14}M^{(1)}_{24}+M^{(1)}_{34}\right) + \Boh(s^{-\frac 12}),\\
X^{(1)}_{13} & = -2 \beta s^{\frac12} +M^{(1)}_{13}+ \Boh(s^{-\frac 12}),\label{eq:asyX113}\\
X^{(1)}_{14} & = M^{(1)}_{14}+ \Boh(s^{-\frac 12}).
\end{align}
Substituting the above equations into \eqref{def:p56} and \eqref{def:q56} yields the asymptotics of $p_5(s)$, $p_6(s)$, $q_5(s)$ and $q_6(s)$ shown in \eqref{p5},\eqref{p6},\eqref{q5} and \eqref{q6}.

\subsubsection*{Asymptotics of $p_k(s)$ and $q_k(s)$ as $s \to 0^+$}
The small $s$  asymptotics of $p_k(s)$ and $q_k(s)$, $k=1, \ldots, 6$, are outcomes of the asymptotic analysis performed in Section \ref{sec:AsyX0}. For $|z|>\delta$, it follows from \eqref{def:tildeR} and \eqref{eq:asytildeR} that
\begin{align}
X(z)=\widecheck R(z) \widecheck{N}(z)=\left(I+\Boh(s)\right)\widecheck{N}(z).
\end{align}
Thus, on account of \eqref{eq:tildeN}, \eqref{eq:asy:M} and \eqref{eq:asyX}, we have
\begin{equation}
X^{(1)}=M^{(1)}\left(I+\Boh(s)\right).
\end{equation}
It is then straightforward to obtain asymptotics of $p_5(s), p_6(s), q_5(s), q_6(s)$ in \eqref{eq:p50}, \eqref{eq:p60}, \eqref{eq:q50} and \eqref{eq:q60} by using \eqref{def:p56}, \eqref{def:q56} and the symmetric relation of $M^{(1)}$ established in \eqref{eq:symmM11} and \eqref{eq:symmM12}.

If $|z| < \delta$, it follows from \eqref{eq:tildeP0}, \eqref{def:eta} and \eqref{def:tildeR} that
\begin{align}
X(z)=\widecheck R(z) \widecheck P^{(0)}(z)
=\widecheck R(z) \what M(z) \left(I+\frac{\gamma}{2 \pi \ii} \ln \left(\frac{z+s}{z-s}\right) \begin{pmatrix}
0 & 0 & 1 & 1\\
0 & 0 & 1 & 1\\
0 & 0 & 0 & 0\\
0 & 0 & 0 & 0
\end{pmatrix}\right), \qquad z \in \Omega_2^{(s)}.
\end{align}
This, together with \eqref{eq:asytildeR}, \eqref{eq:X-near-s} and \eqref{eq: Phi-expand-s}, implies that
\begin{align}
X_{R,0}(s)&=\widecheck R(s) \what M(s) \left(I+\frac{\gamma}{2 \pi \ii} \ln (2s) \begin{pmatrix}
0 & 0 & 1 & 1\\
0 & 0 & 1 & 1\\
0 & 0 & 0 & 0\\
0 & 0 & 0 & 0
\end{pmatrix}\right)\nonumber\\
&=\left(\what M(0)+\Boh(s)\right)\left(I+\frac{\gamma}{2 \pi \ii} \ln (2s) \begin{pmatrix}
0 & 0 & 1 & 1\\
0 & 0 & 1 & 1\\
0 & 0 & 0 & 0\\
0 & 0 & 0 & 0
\end{pmatrix}\right), \qquad s \to 0^+.
\end{align}
As a consequence, it follows from  the definitions of $p_k(s)$ and $q_k(s)$, $k=1,\ldots,4$, in \eqref{def:qkpk} that
\begin{align}\label{eq:pkzero}
\begin{pmatrix}
p_1(s) \\ p_2(s)\\ p_3(s)\\p_4(s)
\end{pmatrix} = -\frac{\gamma}{2 \pi \ii} \left(\what M(0)^{-\msf T} + \Boh(s)\right)\begin{pmatrix}
0\\0\\1\\1
\end{pmatrix}=\Boh(1)
\end{align}
and
\begin{align}\label{eq:qkzero}
\begin{pmatrix}
q_1(s) \\ q_2(s)\\ q_3(s)\\q_4(s)
\end{pmatrix} =  \left(\what M(0) + \Boh(s)\right)\begin{pmatrix}
1\\1\\0\\0
\end{pmatrix}=\Boh(1),
\end{align}
as claimed in \eqref{eq:pk0} and \eqref{eq:qk0}.

This completes the proof of Proposition \ref{th:pq}.
\end{proof}

\begin{remark}
Substituting large $s$ asymptotics of $X^{(1)}_{11}$ and $X^{(1)}_{13}$ given in \eqref{eq:asyX111} and \eqref{eq:asyX113} into \eqref{eq:derivativeFs12}--\eqref{eq:derivativeFtau}, it is readily seen that, as $s \to +\infty$,
\begin{align}
\frac{\partial}{\partial s_1} F(s;\gamma, r_1, r_2,s_1,s_2, \tau) &= - 4 \ii \beta s^{\frac 12} + \Boh(s^{-\frac 12}),\\
\frac{\partial}{\partial s_2} F(s;\gamma, r_1, r_2,s_1,s_2, \tau) &= - 4 \ii \beta s^{\frac 12} + \Boh(s^{-\frac 12}),\\
\frac{\partial}{\partial \tau} F(s;\gamma, r_1, r_2,s_1,s_2, \tau) &= \Boh (s^{-\frac 12}).\label{eq:F-tau}
\end{align}
The above estimates particularly imply that the the non-trivial constant term in the asymptotics of $F(s;\gamma)$, $\gamma\in[0,1)$ is independent of the parameters $s_1, s_2$ and $\tau$.
\end{remark}

\section{Proofs of main results}\label{sec:proof}

\subsection{Proof of Theorem \ref{th:F-H}}
On account of \eqref{eq:derivativeinsX-2} and \eqref{eq:JMU}, it is easily seen that
\begin{equation}
\frac{\ud}{\ud s} F(s;\gamma)=H(s),
\end{equation}
which leads to the integral representation of $F$ claimed in \eqref{eq:F-H2} after integrating with respect to $t$.
It then remains to establish asymptotics of the Hamiltonian $H$, which will be discussed in what follows.

\subsubsection*{Asymptotics of $H(s)$ as $s \to +\infty$ for $\gamma=1$}
We make use of the representation of $H$ in \eqref{eq:JMU}, which reads
\begin{align}
H(s)=-\frac{1}{2 \pi \ii} \sum_{i=3}^4 \sum_{j=1}^2 \left(X_{R,1}(s)-X_{L,1}(s)\right)_{ij},
\end{align}
where $X_{R,1}(s)$ is given in the local behavior of $X$ near $z = s$ in \eqref{eq: Phi-expand-s} and $X_{L,1}(s)$ is given in the local behavior of $X$ near $z = -s$ in \eqref{eq: Phi-expand--s}. Therefore, from \eqref{eq:X-near-s} and \eqref{eq:X-near--s} we obtain
\begin{align}
H(s)=-\frac{1}{2 \pi \ii}\left[ \lim_{z \to s} \sum_{i=3}^4\sum_{j=1}^2 \left(X(z)^{-1}X'(z)\right)_{ij}+ \lim_{z \to -s} \sum_{i=3}^4\sum_{j=1}^2 \left(X(z)^{-1} X'(z)\right)_{ij} \right], \quad z \in \Omega_2^{(s)}.
\end{align}
By inverting the transformations $X \to T \to S$ in \eqref{def:PhiToT} and \eqref{def:TtoS}, it follows from the above formula that
\begin{align}\label{ds-F-1}
H(s) &=-\frac{1}{2 \pi \ii s}\left[ \lim_{z \to 1} \sum_{i=3}^4\sum_{j=1}^2 \left(T(z)^{-1}T'(z)\right)_{ij}+ \lim_{z \to -1} \sum_{i=3}^4\sum_{j=1}^2 \left(T(z)^{-1} T'(z)\right)_{ij} \right]\nonumber\\
&= -\frac{1}{2 \pi \ii s}\left[ \lim_{z \to 1} \sum_{i=3}^4\sum_{j=1}^2 \left(\diag \left(e^{s^{\frac 32}g_1(z)-\tau s z}, e^{s^{\frac 32}g_2(z)+ \tau s z}, e^{-s^{\frac 32}g_1(z)-\tau s z}, e^{-s^{\frac 32}g_2(z) + \tau s z} \right)\right.\right.\nonumber\\
&\quad\left. \times S(z)^{-1}S'(z)\diag \left(e^{s^{-\frac 32}g_1(z)+\tau s z}, e^{-s^{\frac 32}g_2(z)- \tau s z}, e^{s^{\frac 32}g_1(z)+\tau s z}, e^{s^{\frac 32}g_2(z) -\tau s z} \right)\right)_{ij}\nonumber\\
&\quad+\lim_{z \to -1} \sum_{i=3}^4\sum_{j=1}^2 \left(\diag \left(e^{s^{\frac 32}g_1(z)-\tau s z}, e^{s^{\frac 32}g_2(z)+ \tau s z}, e^{-s^{\frac 32}g_1(z)-\tau s z}, e^{-s^{\frac 32}g_2(z) + \tau s z} \right)\right.\nonumber\\
&\quad\left. \left.\times S(z)^{-1}S'(z)\diag \left(e^{s^{-\frac 32}g_1(z)+\tau s z}, e^{-s^{\frac 32}g_2(z)- \tau s z}, e^{s^{\frac 32}g_1(z)+\tau s z}, e^{s^{\frac 32}g_2(z) -\tau s z} \right)\right)_{ij} \right],
\end{align}
where the limits are taken from $\Omega_2^{(1)}$.

We next calculate the limits of $(S(z)^{-1}S'(z))_{ij}$, $i=3,4$, $j=1,2$, as $z\to \pm 1$. If $z$ is close to $-1$, it follows from \eqref{def:R} and \eqref{def:P-1} that
\begin{align}\label{eq:SSinver}
&S(z)^{-1}S'(z) \nonumber\\
&= P^{(-1)}(z)^{-1} R(z)^{-1} R'(z) P^{(-1)}(z) + P^{(-1)}(z)^{-1} (P^{(-1)})'(z)\nonumber\\
&=\mathscr{A}_{-1}(z)^{-1}\mathscr{B}_{-1}(s^3f_{-1}(z))^{-1}E_{-1}(z)^{-1} R(z)^{-1} R'(z) E_{-1}(z)\mathscr{B}_{-1}(s^3f_{-1}(z))\mathscr{A}_{-1}(z)\nonumber\\
&\quad +\mathscr{A}_{-1}(z)^{-1}\mathscr{B}_{-1}(s^3f_{-1}(z))^{-1}E_{-1}(z)^{-1}E_{-1}'(z)\mathscr{B}_{-1}(s^3f_{-1}(z)\mathscr{A}_{-1}(z) \nonumber\\
&\quad+s^3f_{-1}'(z)\mathscr{A}_{-1}(z)^{-1}\mathscr{B}_{-1}(s^3f_{-1}(z))^{-1}\mathscr{B}_{-1}'(s^3f_{-1}(z))\mathscr{A}_{-1}(z) + \mathscr{A}_{-1}(z)^{-1}\mathscr{A}_{-1}'(z),
\end{align}
where $E_{-1}$ is given in \eqref{def:E-1},
\begin{equation}
\mathscr{A}_{-1}(z) =
\begin{pmatrix}
1 & -e^{-s^{\frac32 }(g_1(z)-g_2(z)) + 2 \tau s z} & 0 & 0\\
0 & e^{s^{\frac32}g_2(z)} & 0 & 0\\
0 & 0 & 1 & 0\\
0 & 0 & e^{-s^{\frac32}g_1(z) - 2 \tau s z} & e^{-s^{\frac32}g_2(z)}
\end{pmatrix}
\end{equation}
and
\begin{equation}\label{def:B-1}
\mathscr{B}_{-1}(z) =  \begin{pmatrix}
1 & 0 & 0 & 0\\
0 & \Phi^{(\Bes)}_{11}(z) & 0 & \Phi^{(\Bes)}_{12}(z)\\
0 & 0 & 1 & 0\\
0 & \Phi^{(\Bes)}_{21}(z) & 0 & \Phi^{(\Bes)}_{22}(z)
\end{pmatrix}.
\end{equation}
Recalling the definitions of $g_1(z)$ and $g_2(z)$ in \eqref{def:g1} and \eqref{def:g2}, we have
\begin{equation}\label{g--1}
g_1(-1) = \frac{\sqrt{2}}{3} r_1 + \frac{2^{3/2}s_1}{s}, \qquad g_2(-1) = 0.
\end{equation}
Thus, one can check directly that for an arbitrary $4 \times 4$ matrix $\msf M=(m_{ij})_{i,j=1}^4$, as $s \to +\infty$,
\begin{align}
\lim_{z \to -1} \left(\mathscr{A}_{-1}(z)^{-1} \msf M \mathscr{A}_{-1}(z)\right)_{31}&=m_{31},\label{eq:fact1}\\
\lim_{z \to -1} \left(\mathscr{A}_{-1}(z)^{-1} \msf M \mathscr{A}_{-1}(z)\right)_{32}&=m_{32} + \Boh(e^{-\frac{\sqrt{2}}{3} r_1 s^{\frac 32}}), \\
\lim_{z \to -1} \left(\mathscr{A}_{-1}(z)^{-1} \msf M \mathscr{A}_{-1}(z)\right)_{41}&=m_{41} + \Boh(e^{-\frac{\sqrt{2}}{3} r_1 s^{\frac 32}}),\\
\lim_{z \to -1} \left(\mathscr{A}_{-1}(z)^{-1} \msf M \mathscr{A}_{-1}(z)\right)_{42}&=m_{42} + \Boh(e^{-\frac{\sqrt{2}}{3} r_1 s^{\frac 32}}),
\end{align}
and
\begin{equation}
\lim_{z \to 0} \left(\mathscr{B}_{-1}(z)^{-1} \msf M \mathscr{B}_{-1}(z)\right)_{ij} = m_{ij}, \quad \textrm{for $i = 3,4$, and $j=1,2$.}
\end{equation}
Recall the following properties of the modified Bessel functions $I_0$ and $K_0$ (cf. \cite[Chapter 10]{DLMF}):
\begin{align}
I_0(z) & = \sum_{k=0}^{\infty} \frac{(z/2)^{2k}}{(k!)^2},\label{eq:I0}\\
K_0(z)&=-\left(\ln \left(\frac{z}{2}\right) + \gamma_E\right)I_0(z) + \Boh(z^2), \qquad z \to 0,\label{eq:K0}
\end{align}
where $\gamma_E$ is the Euler's constant, we obtain from \eqref{def:B-1} and \eqref{def:Bespara} that
\begin{equation}
\lim_{z \to 0} \left(\mathscr{B}_{-1}(z)^{-1} \mathscr{B}_{-1}'(z)\right)_{31} = \lim_{z \to 0} \left(\mathscr{B}_{-1}(z)^{-1} \mathscr{B}_{-1}'(z)\right)_{32}=\lim_{z \to 0} \left(\mathscr{B}_{-1}(z)^{-1} \mathscr{B}_{-1}'(z)\right)_{41}=0,
\end{equation}
\begin{equation}
\lim_{z \to 0} \left(\mathscr{B}_{-1}(z)^{-1} \mathscr{B}_{-1}'(z)\right)_{42} = \frac{\pi \ii}{2}.
\end{equation}
With the aid of local behavior of $E_{-1}(z)$ near $z=-1$ in \eqref{def:localE-1} and the explicit expression of $R_1'(-1)$ in \eqref{eq:R1'-1}, we have
\begin{align}
\lim_{z \to -1} E_{-1}(z)^{-1}R(z)^{-1} R'(z) E_{-1}(z) &= \lim_{z \to -1} E_{-1}(z)^{-1}\left(\frac{R_1'(z)}{s^{3/2}}\right) E_{-1}(z)\nonumber\\
& =\frac{\pi \ii r_2}{4\left(r_2 - 2s_2/s\right)}E_{4,2}
 + \Boh(s^{-\frac 32}),
\end{align}
and
\begin{equation}
\lim_{z \to -1} E_{-1}(z)^{-1} E_{-1}'(z) = \begin{pmatrix}
0 & 0 & -\frac{\ii}{8} & 0\\
0 & - \frac{r_2}{3\left(r_2 - 2s_2/s\right)} & 0 & 0\\
\frac{\ii}{8} & 0 & 0 & 0\\
0 & 0 & 0 & \frac{r_2}{3\left(r_2 - 2s_2/s\right)}
\end{pmatrix} + \Boh(s^{-\frac 32}).\label{eq:EEinvlimit}
\end{equation}
A combination of \eqref{eq:SSinver} and \eqref{eq:fact1}--\eqref{eq:EEinvlimit} gives us
\begin{align}
\lim_{z \to -1} \left(S(z)^{-1}S'(z)\right)_{31} &= \frac{\ii}{8} + \Boh(s^{-\frac 32}),\label{eq:SS'-11}\\
\lim_{z \to -1} \left(S(z)^{-1}S'(z)\right)_{32} &= \Boh(e^{-\frac{\sqrt{2}}{3}r_1 s^{\frac 32}}),\\
\lim_{z \to -1} \left(S(z)^{-1}S'(z)\right)_{41} &=\Boh(e^{-\frac{\sqrt{2}}{3}r_1 s^{\frac 32}}),\\
\lim_{z \to -1} \left(S(z)^{-1}S'(z)\right)_{42} &= \frac{\pi \ii}{2}\left(r_2 - \frac{2s_2}{s}\right)^2 s^3 + \frac{\pi \ii}{4} + \Boh(s^{-1}).\label{eq:SS'-14}
\end{align}
Similarly, if $z$ is close to $1$, we use \eqref{def:R} to obtain
\begin{align}
& S(z)^{-1}S'(z)\nonumber\\
&= P^{(1)}(z)^{-1} R(z)^{-1} R'(z) P^{(1)}(z) + P^{(1)}(z)^{-1} (P^{(1)})'(z) \nonumber\\
&=\mathscr{A}_{1}(z)^{-1}\mathscr{B}_{1}(s^3f_1(z))^{-1}E_1(z)^{-1} R(z)^{-1} R'(z) E_1(z)\mathscr{B}_{1}(s^3f_1(z))\mathscr{A}_{1}(z)\nonumber\\
&\quad +\mathscr{A}_{1}(z)^{-1}\mathscr{B}_{1}(s^3f_1(z))^{-1}E_1(z)^{-1} E_1'(z)\mathscr{B}_{1}(s^3f_1(z)\mathscr{A}_{1}(z) \nonumber\\
&\quad+s^3f_1'(z)\mathscr{A}_{1}(z)^{-1}\mathscr{B}_{1}(s^3f_1(z))^{-1}\mathscr{B}_{1}'(s^3f_1(z))\mathscr{A}_{1}(z) + \mathscr{A}_{1}(z)^{-1}\mathscr{A}_{1}'(z),
\end{align}
where $E_1$ is defined in \eqref{def:E1},
\begin{equation}
\mathscr{A}_{1}(z) = \begin{pmatrix}
e^{s^{\frac32}g_1(z)} & 0 & 0 & 0\\
-e^{s^{\frac32}(g_1(z)-g_2(z)) - 2 \tau s z} &1 & 0 & 0\\
0 & 0 & e^{-s^{\frac32}g_1(z)} & e^{-s^{\frac32}g_2(z) + 2 \tau s z}\\
0 & 0 & 0 & 1
\end{pmatrix}
\end{equation}
and
\begin{equation}\label{def:B1}
\mathscr{B}_{1}(z)=
\begin{pmatrix}
\Phi^{(\Bes)}_{11}(z) & 0 & -\Phi^{(\Bes)}_{12}(z) & 0\\
0 & 1 & 0 & 0\\
-\Phi^{(\Bes)}_{21}(z) & 0 & \Phi^{(\Bes)}_{22}(z)\\
0 & 0 & 0 & 1
\end{pmatrix}.
\end{equation}
As $z \to 1$, from \eqref{def:g1} and \eqref{def:g2}, it follows that
\begin{equation}\label{g-1}
g_1(1) =0, \qquad g_2(1) = \frac{\sqrt{2}}{3} r_2 + \frac{2^{3/2}s_2}{s}.
\end{equation}
For an arbitrary $4 \times 4$ matrix $\msf M=(m_{ij})_{i,j=1}^4$, as $s \to +\infty$, we have
\begin{align}
\lim_{z \to 1} \left(\mathscr{A}_{1}(z)^{-1} \msf M \mathscr{A}_{1}(z)\right)_{31} &= m_{31} + \Boh(e^{-\frac{\sqrt{2}}{3}r_2 s^{\frac 32}}),\\
\lim_{z \to 1} \left(\mathscr{A}_{1}(z)^{-1} \msf M \mathscr{A}_{1}(z)\right)_{32} &= m_{32} + \Boh(e^{-\frac{\sqrt{2}}{3}r_2 s^{\frac 32}}),\\
\lim_{z \to 1} \left(\mathscr{A}_{1}(z)^{-1} \msf M \mathscr{A}_{1}(z)\right)_{41} &= m_{41} + \Boh(e^{-\frac{\sqrt{2}}{3}r_2 s^{\frac 32}}),\\
\lim_{z \to 1} \left(\mathscr{A}_{1}(z)^{-1} \msf M \mathscr{A}_{1}(z)\right)_{31} &= m_{42}.
\end{align}
and
\begin{equation}
\lim_{z \to 0} \left(\mathscr{B}_{1}(z)^{-1} \msf M \mathscr{B}_{1}(z)\right)_{ij}  = m_{ij}, \quad \textrm{for $i=3,4$ and $j = 1, 2.$}
\end{equation}
From the properties of the modified Bessel functions in \eqref{eq:I0} and \eqref{eq:K0}, we see from \eqref{def:B1} and \eqref{def:Bespara} that
\begin{equation}
\lim_{z \to 0} \left(\mathscr{B}_{1}(z)^{-1} \mathscr{B}_{1}'(z)\right)_{32} = \lim_{z \to 0} \left(\mathscr{B}_{1}(z)^{-1} \mathscr{B}_{1}'(z)\right)_{41}=\lim_{z \to 0} \left(\mathscr{B}_{1}(z)^{-1} \mathscr{B}_{1}'(z)\right)_{42}=0,
\end{equation}
\begin{equation}
\lim_{z \to 0} \left(\mathscr{B}_{1}(z)^{-1} \mathscr{B}_{1}'(z)\right)_{31} = -\frac{\pi \ii}{2}.
\end{equation}
By using the local behavior of $E_1(z)$ in \eqref{def:localE1} and the explicit expression of $R_1'(1)$ in \eqref{eq:R1'1}, we have
\begin{align}
\lim_{z \to 1} E_1(z)^{-1} R(z)^{-1} R'(z) E_1(z) &= \lim_{z \to 1} E_1(z)^{-1} \left(\frac{R_1'(z)}{s^{3/2}} + \Boh(s^{-3})\right) E_1(z) \nonumber\\
&= \frac{\pi \ii r_1}{4\left(r_1 - 2s_1/s\right)}E_{3,1}
 + \Boh(s^{-\frac 32}),
\end{align}
and
\begin{equation}
\lim_{z \to 1} E_1(z)^{-1} E_1'(z) = \begin{pmatrix}
\frac{r_1}{3\left(r_1 - 2s_1/s\right)} & 0 & 0 & 0\\
0 & 0 & 0 & -\frac{\ii}{8}\\
0 & 0 & -\frac{r_1}{3\left(r_1 - 2s_1/s\right)} & 0\\
0 & \frac{\ii}{8} & 0 & 0
\end{pmatrix} + \Boh(s^{-\frac 32}).
\end{equation}
Thus, we obtain after a direct calculation that
\begin{align}
\lim_{z \to 1} \left(S(z)^{-1}S'(z)\right)_{31} &= \frac{\pi \ii}{2}\left(r_1 - \frac{2s_1}{s}\right)^2 s^3 + \frac{\pi \ii}{4} + \Boh(s^{-1}),\label{eq:SS'11}\\
\lim_{z \to 1} \left(S(z)^{-1}S'(z)\right)_{32} &= \Boh(e^{-\frac{\sqrt{2}}{3}r_2 s^{\frac 32}}),\\
\lim_{z \to 1} \left(S(z)^{-1}S'(z)\right)_{41} &=\Boh(e^{-\frac{\sqrt{2}}{3}r_2 s^{\frac 32}}),\\
\lim_{z \to 1} \left(S(z)^{-1}S'(z)\right)_{42} &= \frac{\ii}{8} + \Boh(s^{-\frac 32}).\label{eq:SS'14}
\end{align}
Substituting equations \eqref{eq:SS'-11}--\eqref{eq:SS'-14} and \eqref{eq:SS'11}--\eqref{eq:SS'14} into \eqref{ds-F-1}, it follows from  \eqref{g--1} and \eqref{g-1} that
\begin{equation}\label{eq:H-asy1}
H(s) = -\frac{r_1^2 + r_2^2}{4}s^2 + (r_1s_1 + r_2s_2)s - s_1^2-s_2^2 - \frac{1}{4s} + \Boh(s^{-2}), \quad s \to +\infty,
\end{equation}
as shown in \eqref{asy:H}.

\subsubsection*{Asymptotics of $H(s)$ as $s \to +\infty$ for $0\leq \gamma<1$}
If $\gamma \in [0, 1)$, one can theoretically obtain asymptotics of $H(s)$ by combining \eqref{def:H} and the large $s$ asymptotics of $p_k(s)$ and $q_k(s)$, $k=1,\ldots,6$, established in Theorem \ref{th:pq}. This approach, however, is too complicated. Alternatively, we again
turn to use the relation \eqref{eq:JMU}, which can be written as
\begin{equation}\label{eq:HXRL}
H(s) = -\frac{\gamma}{2 \pi \ii} \begin{pmatrix}
0 & 0 & 1 & 1
\end{pmatrix} \left( X_{R,1}(s) - X_{L,1}(s)\right)\begin{pmatrix}
1\\1\\0\\0
\end{pmatrix},
\end{equation}
where $X_{R,1}(s)$ and $X_{L,1}(s)$ are given in \eqref{eq: Phi-expand-s} and \eqref{eq: Phi-expand--s}.

From \eqref{eq:X-near-s} and \eqref{eq: Phi-expand-s}, it follows that
\begin{equation}
X_{R,1}(s) = \frac{X_{R,0}(s)^{-1}}{s} \lim_{z \to 1, z \in \Omega_2^{(1)}} \left[X(sz) \begin{pmatrix}
1 & 0 & \frac{\gamma}{2 \pi \ii} \ln{(sz-s)} & \frac{\gamma}{2 \pi \ii} \ln{(sz-s)}\\
0 & 1 & \frac{\gamma}{2 \pi \ii} \ln{(sz-s)} & \frac{\gamma}{2 \pi \ii} \ln{(sz-s)}\\
0 & 0 & 1 & 0\\
0 & 0 & 0 & 1
\end{pmatrix}\right]'
\end{equation}
Note that (see \eqref{eq:XR0})
\begin{equation}
\begin{pmatrix}
0 & 0 & 1 & 1
\end{pmatrix} X_{R,0}(s)^{-1} = \begin{pmatrix}
0 & 0 & e^{-\tau s} & 0
\end{pmatrix} \what{\Upsilon}_0^{-1}\what E_1(1)^{-1} \what R(1)^{-1} \diag (s^{\frac 14}, s^{\frac 14}, s^{-\frac 14}, s^{-\frac 14}),
\end{equation}
and by tracing back the transformations $X \to \what{T} \to \what S \to \what R$ in \eqref{def:XToTgamma}, \eqref{def:hatS} and \eqref{def:hatR},
\begin{align}
&\lim_{z \to 1, z \in \Omega_2^{(1)}} \left[X(sz) \begin{pmatrix}
1 & 0 & \frac{\gamma}{2 \pi \ii} \ln{(z-s)} & \frac{\gamma}{2 \pi \ii} \ln{(z-s)}\\
0 & 1 & \frac{\gamma}{2 \pi \ii} \ln{(z-s)} & \frac{\gamma}{2 \pi \ii} \ln{(z-s)}\\
0 & 0 & 1 & 0\\
0 & 0 & 0 & 1
\end{pmatrix}\right]'\begin{pmatrix}
1\\1\\0\\0
\end{pmatrix} \nonumber\\
&= \diag (s^{-\frac 14}, s^{-\frac 14}, s^{\frac 14}, s^{\frac 14})\nonumber\\
& \quad \times \lim_{z \to 1, z \in \Omega_2} \left[ \what R(z) \what P^{(1)} (z) \diag(e^{- \theta_1(sz) + \tau s z},e^{- \theta_2(sz) - \tau s z},e^{ \theta_1(sz) + \tau s z},e^{\theta_2(sz) - \tau s z}) \right]'\begin{pmatrix}
1\\1\\0\\0
\end{pmatrix}\nonumber\\
&=\diag (s^{-\frac 14}, s^{-\frac 14}, s^{\frac 14}, s^{\frac 14}) \left(\what R'(1) \what E_1 (1)\what{\Upsilon}_0 e^{\tau s} + \what R(1) \what E_1' (1)\what{\Upsilon}_0 e^{\tau s}\right.\nonumber\\
& \quad \left. +s^{\frac 32} \what f_1'(1) \what R(1) \what E_1(1)\what{\Upsilon}_0 \what{\Upsilon}_1 e^{\tau s} + \tau s \what R(1) \what E_1(1)\what{\Upsilon}_0e^{\tau s}\right)\begin{pmatrix}
1\\0\\0\\0
\end{pmatrix},
\end{align}
where $\what{\Upsilon}_0$ is defined in \eqref{def:whatUpsilon} and $\what{\Upsilon}_1$ is a $4\times 4$ matrix related to $\Upsilon_1$ in \eqref{eq:H-expand-2} with $(\what{\Upsilon}_1)_{31}=\frac{ \beta \pi \ii e^{-\beta \pi \ii} }{\sin(\beta \pi )}$. It follows from the above three equations that
\begin{align}\label{HR1}
-\frac{\gamma}{2 \pi \ii} \begin{pmatrix}
0 & 0 & 1 & 1
\end{pmatrix} X_{R,1}(s) \begin{pmatrix}
1\\1\\0\\0
\end{pmatrix}&=-\frac{\gamma}{2 \pi \ii s} \left(\what{\Upsilon}_0^{-1}\what E_1(1)^{-1} \what R(1)^{-1}\what R'(1) \what E_1 (1)\what{\Upsilon}_0\right.\nonumber\\
&\quad\left.+\what{\Upsilon}_0^{-1}\what E_1(1)^{-1}\what E_1' (1)\what{\Upsilon}_0+s^{\frac 32} \what f_1'(1) \what{\Upsilon}_1 + \tau s I \right)_{31}.
\end{align}
To calculate the first term in the bracket, we substitute \eqref{eq:hatE11}, \eqref{def:whatUpsilon} and \eqref{eq:hatR-1} into the above equation and obtain
\begin{align}\label{term1}
\left(\what{\Upsilon}_0^{-1}\what E_1(1)^{-1} \what R(1)^{-1}\what R'(1) \what E_1 (1)\what{\Upsilon}_0\right)_{31}=\Boh(s^{-\frac 12}).
\end{align}
From the explicit expression of $\what{\Upsilon}_0$ in \eqref{def:whatUpsilon} and $\what E_1(1)^{-1}\what E_1' (1)$ in \eqref{eq:hatE1'1}, one has
\begin{align}
\left(\what{\Upsilon}_0^{-1}\what E_1(1)^{-1}\what E_1' (1)\what{\Upsilon}_0\right)_{31}&=-\frac{3 \beta e^{-\beta \pi \ii}}{2} \Gamma (1+\beta) \Gamma (1-\beta) -\frac{\ii}{4} \left(\frac{\Gamma(1+\beta)^2}{\textsf{a}^2} + \textsf{a}^2 \Gamma(1-\beta)^2 e^{-2 \beta \pi \ii} \right)\nonumber\\
&\quad+ \Boh(s^{-1}),
\end{align}
where $\textsf{a}$ is given in \eqref{def:a}. Moreover, from \cite[Chapter 5]{DLMF}, we have
\begin{align}
\Gamma (\beta) \Gamma (1-\beta) = \frac{\pi}{\sin (\beta \pi)},
\end{align}
and hence
\begin{align}
\Gamma (1+\beta) \Gamma (1-\beta) = |\Gamma (1+\beta)|^2=\frac{\beta \pi}{\sin (\beta \pi)}.
\end{align}
Therefore, we get
\begin{align}\label{term2}
\left(\what{\Upsilon}_0^{-1}\what E_1(1)^{-1}\what E_1' (1)\what{\Upsilon}_0\right)_{31}=-\frac{3 \beta^2 \pi e^{-\beta \pi \ii}}{2 \sin (\beta \pi)} -\frac{\beta \pi \ii e^{-\beta \pi \ii}}{2\sin (\beta \pi)} \cos (2 \vartheta (s))+ \Boh(s^{-1}).
\end{align}
where $\vartheta (s)$ is given in \eqref{def:theta}. We see from \eqref{eq:hatf1} that
\begin{align}\label{term3}
\left(s^{\frac 32} \what f_1'(1) \what{\Upsilon}_1\right)_{31} = \frac{2 \beta \pi \ii e^{-\beta \pi \ii}}{\sin (\beta \pi)}\left(r_1 s^{\frac 32}-s_1 s^{\frac 12}\right)
\end{align}
and
\begin{align}\label{term4}
\left(\tau s I \right)_{31} = 0.
\end{align}
It is then straightforward to calculate the following result by substituting \eqref{term1}, \eqref{term2} \eqref{term3}, \eqref{term4} and \eqref{def:beta} into \eqref{HR1}
\begin{multline}\label{eq:asy:H1}
-\frac{\gamma}{2 \pi \ii} \begin{pmatrix}
0 & 0 & 1 & 1
\end{pmatrix} X_{R,1}(s) \begin{pmatrix}
1\\1\\0\\0
\end{pmatrix} =2 \ii \beta r_1 s^{\frac 12} - 2 \ii \beta s_1 s^{-\frac 12} - \frac{3 \beta^2}{2} s^{-1}
\\
- \frac {\ii \beta}{2} \cos {(2 \vartheta (s))} s^{-1} + \Boh(s^{-\frac 32}),
\end{multline}
where $\vartheta (s)$ is given in \eqref{def:theta}.
Similarly, from \eqref{eq:X-near--s} and \eqref{eq: Phi-expand--s}, it follows that
\begin{align}
X_{L,1}(s) &= \frac{X_{L,0}(s)^{-1}}{s}
\nonumber \\
&\quad \times  \lim_{z \to -1, z \in \Omega_2^{(1)}} \left[X(sz)\begin{pmatrix}
0 & 1 & 0 & 0\\
1 & 0 & 0 & 0\\
0 & 0 & 0 & -1\\
0 & 0 & -1 & 0
\end{pmatrix} \begin{pmatrix}
1 & 0 & \frac{\gamma}{2 \pi \ii} \ln{(-sz-s)} & \frac{\gamma}{2 \pi \ii} \ln{(-sz-s)}\\
0 & 1 & \frac{\gamma}{2 \pi \ii} \ln{(-sz-s)} & \frac{\gamma}{2 \pi \ii} \ln{(-sz-s)}\\
0 & 0 & 1 & 0\\
0 & 0 & 0 & 1
\end{pmatrix}\right]'
\end{align}
with
\begin{equation}
\begin{pmatrix}
0 & 0 & 1 & 1
\end{pmatrix} X_{L,0}(s)^{-1} = \begin{pmatrix}
0 & 0 & 0 & e^{-\tau s-\beta \pi \ii}
\end{pmatrix} \widecheck{\Upsilon}_0^{-1}\what E_{-1}(-1)^{-1} \what R(-1)^{-1} \diag (s^{\frac 14}, s^{\frac 14}, s^{-\frac 14}, s^{-\frac 14}),
\end{equation}
where
\begin{align}\label{eq:tildeupsilon}
\widecheck{\Upsilon}_0=\begin{pmatrix}
1 & 0 & 0 & 0\\
0&\left(\Upsilon_0\right)_{11} & 0 & \left(\Upsilon_0\right)_{12}\\
0& 0 & 1 & 0 \\
0 & \left(\Upsilon_0\right)_{21} & 0 & \left(\Upsilon_0\right)_{22}
\end{pmatrix}.
\end{align}
Tracing back the transformations $X \to \what{T} \to \what S \to \what R$ in \eqref{def:XToTgamma}, \eqref{def:hatS} and \eqref{def:hatR} yields
\begin{align}
& \lim_{z \to -1, z \in \Omega_2^{(1)}} \left[X(sz) \begin{pmatrix}
0 & 1 & \frac{\gamma}{2 \pi \ii} \ln{(z-s)} & \frac{\gamma}{2 \pi \ii} \ln{(z-s)}\\
1 & 0 & \frac{\gamma}{2 \pi \ii} \ln{(z-s)} & \frac{\gamma}{2 \pi \ii} \ln{(z-s)}\\
0 & 0 & 0 & -1\\
0 & 0 & -1 & 0
\end{pmatrix}\right]'\begin{pmatrix}
1\\1\\0\\0
\end{pmatrix} \nonumber\\
&= \diag (s^{-\frac 14}, s^{-\frac 14}, s^{\frac 14}, s^{\frac 14})\nonumber\\
&\quad \times \lim_{z \to -1, z \in \Omega_2} \left[ \what R(z) \what P^{(-1)} (z) \diag(e^{- \theta_1(sz) + \tau s z},e^{- \theta_2(sz) - \tau s z},e^{ \theta_1(sz) + \tau s z},e^{\theta_2(zs) - \tau s z}) \right]'\begin{pmatrix}
1\\1\\0\\0
\end{pmatrix}\nonumber\\
&=\diag (s^{-\frac 14}, s^{-\frac 14}, s^{\frac 14}, s^{\frac 14}) \left(\what R'(-1) \what E_{-1} (-1)\widecheck{\Upsilon}_0 e^{\tau s -\beta \pi \ii} + \what R(-1) \what E_{-1}' (-1)\widecheck{\Upsilon}_0 e^{\tau s -\beta \pi \ii}\right.\nonumber\\
&\quad  \left. +s^{\frac 32} \what f_{-1}'(-1) \what R(-1) \what E_{-1}(-1)\widecheck{\Upsilon}_0 \widecheck{\Upsilon}_1 e^{\tau s -\beta \pi \ii} - \tau s \what R(-1) \what E_{-1}(-1)\widecheck{\Upsilon}_0e^{\tau s - \beta \pi \ii}\right)\begin{pmatrix}
0\\1\\0\\0
\end{pmatrix},\label{eq:XL-1}
\end{align}
where $\widecheck{\Upsilon}_0$ is given in \eqref{eq:tildeupsilon} and $\widecheck{\Upsilon}_1$ is a $4\times 4$ matrix related to $\Upsilon_1$ in \eqref{eq:H-expand-2} with $(\widecheck{\Upsilon}_1)_{42}=\frac{ \beta \pi \ii  e^{-\beta \pi \ii} }{\sin(\beta \pi )}$. With the aid of the explicit expressions of $\what E_{-1}(-1)$ and $\what E_{-1}'(-1)$ in \eqref{eq:hatE-1-1} and \eqref{eq:hatE-1'-1}, we obtain from \eqref{eq:hatRexpansion} and \eqref{eq:hatf-1} that
\begin{equation}\label{eq:asy:H2}
\frac{\gamma}{2 \pi \ii} \begin{pmatrix}
0 & 0 & 1 & 1
\end{pmatrix} X_{L,1}(s) \begin{pmatrix}
1\\1\\0\\0
\end{pmatrix} = 2 \ii \beta r_2 s^{\frac 12} - 2 \ii \beta s_2 s^{-\frac 12} - \frac{3 \beta^2}{2} s^{-1} - \frac {\ii \beta}{2} \cos {(2 \widetilde{\vartheta} (s))} s^{-1}+ \Boh(s^{-\frac 32}),
\end{equation}
where $\widetilde \vartheta (s)$ is given in \eqref{def:tildetheta}.

Inserting \eqref{eq:asy:H1} and \eqref{eq:asy:H2} into \eqref{eq:HXRL}, we arrive at
\begin{multline}
H(s) = 2 \ii \beta(r_1 + r_2) s^{\frac 12} - 2 \ii \beta (s_1+s_2) s^{-\frac 12} \\
- (3 \beta^2 + \frac{\ii \beta}{2} \cos{(2 \vartheta(s))}+\frac{\ii \beta}{2} \cos{(2 \widetilde{\vartheta}(s))})s^{-1} + \Boh(s^{-\frac 32}),
\end{multline}
as required.

\subsubsection*{Asymptotics of $H(s)$ as $s \to 0^+$}
Recall the symmetric relation of $M$ in \eqref{eq:symmM1}, it is easily seen from \eqref{eq:pkzero} and \eqref{eq:qkzero} that, as $s \to 0^+$,
\begin{align}
p_1(s) &= -\widetilde p_2(s)+ \Boh(s),\label{p1p2}\\
p_3(s) &=\widetilde p_4(s)+ \Boh(s),\\
q_1(s) &= \widetilde q_2(s) + \Boh(s),\label{q1q2}\\
q_3(s) & = -\widetilde q_4(s)+ \Boh(s).
\end{align}
This, together with \eqref{eq:pk0}, \eqref{eq:qk0} and the relation \eqref{eq:sumpq}, implies that
\begin{align}
&p_2(s)\widetilde q_1(s)+p_1(s) \widetilde q_2(s)-p_4(s)\widetilde q_3(s)-p_3(s) \widetilde q_4(s)\nonumber\\
& =p_2(s)q_2(s)+p_1(s) q_1(s)+p_4(s) q_4(s)+p_3(s) q_3(s)+ \Boh(s)=\Boh(s),
\end{align}
or equivalently,
\begin{align}
\widetilde p_2(s) q_1(s) +\widetilde p_1(s) q_2(s)-\widetilde p_4(s)q_3(s)
 -\widetilde p_3(s)q_4(s)=\Boh(s).
\end{align}
Inserting the above two estimates into the expression of $H(s)$ in \eqref{def:H2}, it then follows from the small $s$ asymptotics of $p_k$ and $q_k$, $k=1\ldots,6$, established in Theorem \ref{th:pq} that
\begin{equation}\label{eq:asyH0}
H(s) = \Boh(1), \qquad s\to 0^+.
\end{equation}

This finishes the proof of Theorem \ref{th:F-H}.
\qed

\subsection{Proof of Theorem \ref{th:1}}
Large $s$ asymptotics of $F(s; 1)$ in \eqref{ds-F} follows directly from \eqref{eq:H-asy1} and the integral representation of $F$ in \eqref{eq:F-H2}.

To establish large $s$ asymptotics of $F(s; \gamma)$ for $\gamma \in [0, 1)$, we note from \eqref{eq:F-tau} that the first few terms till the constant term in the expansion are independent of $\tau$. We thus simply take $\tau=0$ and obtain from \eqref{eq:differentialH1} that
\begin{align}\label{eq:int0sH}
\int_0^s H(t) \ud t & =\int_0^s  \left(\sum_{k=1}^6\left(p_k(t)q_k'(t)+\widetilde p_k(t)\widetilde q_k'(t)\right) -H(t)\right) \ud t + \frac{1}{3} [2tH(t) + p_1(t)q_1(t) \nonumber\\
&\quad + p_2(t)q_2(t)+\widetilde p_1(t)\widetilde q_1(t) + \widetilde p_2(t)\widetilde q_2(t) -2p_5(t)q_5(t)-2\widetilde p_5(t) \widetilde q_5(t)-p_6(t)q_6(t)\nonumber\\
&\quad-\widetilde p_6(t) \widetilde q_6(t)+2\frac{s_1}{r_1} p_5(t) + 2 \frac{s_2}{r_2} \widetilde p_5(t)]_{t=0}^s.
\end{align}

The integral on the right hand side of the above equation can be evaluated with the aid of \eqref{eq:differential:gamma}, where also holds if we replace $\gamma$ by $\beta$. By integrating both sides of \eqref{eq:differential:gamma} with respect to $s$, it is readily seen that
\begin{align}
&\frac{\partial}{\partial \beta} \int_0^s  \left(\sum_{k=1}^6\left(p_k(t)q_k'(t)+\widetilde p_k(t)\widetilde q_k'(t)\right) -H(t)\right) \ud t \nonumber\\
&= \sum_{k=1}^6\left(p_k(s)\frac{\partial}{\partial \beta}q_k(s)+  \widetilde p_k(s)\frac{\partial}{\partial \beta}  \widetilde q_k(s)-p_k(0)\frac{\partial}{\partial \beta}q_k(0)-  \widetilde p_k(0)\frac{\partial}{\partial \beta}  \widetilde q_k(0)\right)\nonumber\\
&=\sum_{k=1}^6\left(p_k(s)\frac{\partial}{\partial \beta}q_k(s)+  \widetilde p_k(s)\frac{\partial}{\partial \beta}  \widetilde q_k(s)\right),
\end{align}
where the second equality follows from small $s$ asymptotics of $p_k$ and $q_k$, $k=1,\ldots,6$, established in Theorem \ref{th:pq}.
Inserting the above formula into \eqref{eq:int0sH}, with the aids of the relations \eqref{p1p2}, \eqref{q1q2} and the small $s$ asymptotics of $p_5$, $p_6$, $q_5$, $q_6$ in \eqref{eq:p50}, \eqref{eq:p60}, \eqref{eq:q50}, \eqref{eq:q60}, we arrive at
\begin{align}\label{eq:asy:H3}
\int_0^s H(t) \ud t &= \int_0^{\beta}\sum_{k=1}^6\left(p_k(s)\frac{\partial}{\partial \beta'}q_k(s)+  \widetilde p_k(s)\frac{\partial}{\partial \beta'}  \widetilde q_k(s)\right) \ud \beta'+\frac{1}{3} [2tH(t) + p_1(t)q_1(t) \nonumber\\
&\quad + p_2(t)q_2(t)+\widetilde p_1(t)\widetilde q_1(t) + \widetilde p_2(t)\widetilde q_2(t) -2p_5(t)q_5(t)-2\widetilde p_5(t) \widetilde q_5(t)-p_6(t)q_6(t)\nonumber\\
&\quad-\widetilde p_6(t) \widetilde q_6(t)+2\frac{s_1}{r_1} p_5(t) + 2 \frac{s_2}{r_2} \widetilde p_5(t)]_{t=0}^s\nonumber\\
&=\int_0^{\beta}\sum_{k=1}^6\left(p_k(s)\frac{\partial}{\partial \beta'}q_k(s)+  \widetilde p_k(s)\frac{\partial}{\partial \beta'}  \widetilde q_k(s)\right) \ud \beta' + \frac{1}{3} [2sH(s) + p_1(s)q_1(s) \nonumber\\
&\quad + p_2(s)q_2(s)+\widetilde p_1(s)\widetilde q_1(s) + \widetilde p_2(s)\widetilde q_2(s) -2p_5(s)q_5(s)-2\widetilde p_5(s) \widetilde q_5(s)-p_6(s)q_6(s)\nonumber\\
&\quad-\widetilde p_6(s) \widetilde q_6(s)+2\frac{s_1}{r_1} p_5(s) + 2 \frac{s_2}{r_2} \widetilde p_5(s)-4\ii r_1 M_{11}^{(1)} M_{13}^{(1)}-4\ii r_2 \widetilde M_{11}^{(1)} \widetilde M_{13}^{(1)}\nonumber\\
&\quad + ( M_{14}^{(1)}+ \widetilde M_{14}^{(1)})(\ii r_1 M_{12}^{(1)} +\ii r_2 \widetilde M_{12}^{(1)})-2\ii s_1 M_{13}^{(1)}-2\ii s_2 \widetilde M_{13}^{(1)}].
\end{align}

We next estimate the terms $p_k(s) \frac{\partial}{\partial \beta}q_k(s)$, $k=1,\ldots,6$, as $s\to +\infty$ by using Theorem \ref{th:pq}.
Substituting \eqref{p1} and \eqref{q1} into $p_1(s) \frac{\partial}{\partial \beta}q_1(s)$, it follows that
\begin{align}
&p_1(s) \frac{\partial}{\partial \beta}q_1(s)=p_1(s) q_1(s)\frac{\partial}{\partial \beta}\ln q_1(s)\nonumber\\
&=\ii \beta \left(\cos (2\vartheta(s))+2 \ii \beta \sin(2\vartheta(s))+2\ii \beta\right)\left(-\frac{\pi \ii}{2} + \frac{\partial}{\partial \beta}\ln |\Gamma(1-\beta)|-\tan(\vartheta(s)-\frac{\pi}{4})\frac{\partial}{\partial \beta}\vartheta(s)\right)\nonumber\\
&\quad+\Boh\left(\frac{\ln s}{s^{1/2}}\right).
\end{align}
Similarly, we have
\begin{align}
p_2(s) \frac{\partial}{\partial \beta}q_2(s) = \Boh\left(\frac{\ln s}{s}\right)
\end{align}
from \eqref{p2} and \eqref{q2}. With the help of \eqref{p3} and \eqref{q3}, we obtain
\begin{align}
&p_3(s) \frac{\partial}{\partial \beta}q_3(s)=p_3(s) q_3(s)\frac{\partial}{\partial \beta}\ln q_3(s)\nonumber\\
&=2 \beta \left(2 \cos^2(\vartheta(s)-\frac{\pi}{4})\left(1-\frac{\beta \pi \ii}{2} + \beta \frac{\partial}{\partial \beta}\ln |\Gamma(1-\beta)|-\beta\tan(\vartheta(s)-\frac{\pi}{4})\frac{\partial}{\partial \beta}\vartheta(s)\right)\right.\nonumber\\
&\quad\left.+\frac{1}{2} \sin(2 \vartheta(s)-\frac{\pi}{2})\left(\frac{\pi}{2}+\ii \frac{\partial}{\partial \beta}\ln |\Gamma(1-\beta)|+\ii \cot(\vartheta(s)-\frac{\pi}{4})\frac{\partial}{\partial \beta}\vartheta(s)\right)\right)+\Boh\left(\frac{\ln s}{s^{1/2}}\right),
\end{align}
and by \eqref{p4} and \eqref{q4},
\begin{align}
p_4(s) \frac{\partial}{\partial \beta}q_4(s) = \Boh\left(\frac{\ln s}{s}\right).
\end{align}
Adding the above four formulas together, it follows from a direct calculation that
\begin{align}\label{eq:asypqk}
\sum_{k=1}^4 p_k(s) \frac{\partial}{\partial \beta}q_k(s)& = -\ii \beta \cos (2\vartheta(s))\tan(\vartheta(s)-\frac{\pi}{4})\frac{\partial}{\partial \beta}\vartheta(s)+2\beta(\sin(2 \vartheta(s))+1)\nonumber\\
&\quad-\ii \beta \cos (2\vartheta(s))\cot(\vartheta(s)-\frac{\pi}{4})\frac{\partial}{\partial \beta}\vartheta(s)+\Boh(s^{-\frac 12})\nonumber\\
&=2\ii \beta\frac{\partial}{\partial \beta}\vartheta(s) + 2 \beta(\sin(2 \vartheta(s))+1)+\Boh\left(\frac{\ln s}{s^{1/2}}\right).
\end{align}
To estimate $p_5(s) \frac{\partial}{\partial \beta}q_5(s)$, we refer to \eqref{def:q52} and rewrite it as
\begin{align}
& p_5(s) \frac{\partial}{\partial \beta}q_5(s)\nonumber\\
&=\frac{1}{\ii r_1}p_5(s)\frac{\partial}{\partial \beta}\left(\ii r_2 q_6(s)\widetilde q_6(s) - \frac{p_5(s)^2}{\ii r_1} - p_3(s)q_1(s)+\widetilde p_4(s) \widetilde q_2(s)-\ii s_1\right)\nonumber\\
&=\frac{r_2}{r_1} p_5(s)\frac{\partial}{\partial \beta}(q_6(s)\widetilde q_6(s))+\frac{2}{3r_1^2}\frac{\partial}{\partial \beta}p_5(s)^3-\frac{1}{\ii r_1} p_5(s)\frac{\partial}{\partial \beta}(p_3(s)q_1(s))+\Boh\left(\frac{\ln s}{s^{3/2}}\right).
\end{align}
Note that, with the aid of \eqref{p5} and \eqref{q6},
\begin{align}
&\quad\frac{r_2}{r_1} \int_0^{\beta}p_5(s)\frac{\partial}{\partial \beta'}(q_6(s)\widetilde q_6(s)) \ud \beta'\nonumber\\
&=2\ii r_2 \beta^2\left(\widetilde M^{(1)}_{14}\left(M^{(1)}_{12} +M^{(1)}_{13} M^{(1)}_{14}+M^{(1)}_{14}M^{(1)}_{24}+M^{(1)}_{34}\right)+M^{(1)}_{14} \left(\widetilde M^{(1)}_{12}+\widetilde M^{(1)}_{13}\widetilde M^{(1)}_{14}\right.\right.\nonumber\\
&\quad\left.\left.+\widetilde M^{(1)}_{14}\widetilde M^{(1)}_{24}+\widetilde M^{(1)}_{34}\right)\right) + \Boh\left(\frac{\ln s}{s^{1/2}}\right),
\end{align}
and
\begin{align}
\frac{2}{3r_1^2}\int_0^{\beta}\frac{\partial}{\partial \beta'}p_5(s)^3 \ud \beta' = \frac{2}{3r_1^2}p_5(s)^3+\frac{\ii r_1}{3} (M^{(1)}_{13})^3 + \Boh\left(\frac{\ln s}{s^{1/2}}\right).
\end{align}
Moreover, we obtain from integration by parts and the asymptotic formulas in \eqref{p5}, \eqref{p3} and \eqref{q1} that
\begin{multline}
-\frac{1}{\ii r_1} \int_0^{\beta}p_5(s)\frac{\partial}{\partial \beta'}(p_3(s)q_1(s)) \ud \beta' =-\frac{1}{\ii r_1}p_5(s)p_3(s)q_1(s)
\\
-\int_0^{\beta}4\beta'\cos^2(\vartheta(s)-\frac{\pi}{4}) \ud \beta' + \Boh\left(\frac{\ln s}{s^{1/2}}\right).
\end{multline}
Therefore, it is readily seen from the above four formulas that
\begin{align}\label{eq:asypq5}
& \int_0^{\beta}p_5(s) \frac{\partial}{\partial \beta'}q_5(s)\ud \beta' \nonumber\\
 &=2 \ii r_2 \beta^2\left(\widetilde M^{(1)}_{14}\left(M^{(1)}_{12} +M^{(1)}_{13} M^{(1)}_{14}+M^{(1)}_{14}M^{(1)}_{24}+M^{(1)}_{34}\right)+M^{(1)}_{14} \left(\widetilde M^{(1)}_{12}+\widetilde M^{(1)}_{13}\widetilde M^{(1)}_{14}+\widetilde M^{(1)}_{14}\widetilde M^{(1)}_{24}\right.\right.\nonumber\\
&\quad\left.\left.+\widetilde M^{(1)}_{34}\right)\right)+ \frac{2}{3r_1^2}p_5(s)^3+\frac{\ii r_1}{3} (M^{(1)}_{13})^3-\frac{1}{\ii r_1}p_5(s)p_3(s)q_1(s)-\int_0^{\beta}4\beta'\cos^2(\vartheta(s)-\frac{\pi}{4}) \ud \beta' \nonumber\\
&\quad+ \Boh\left(\frac{\ln s}{s^{1/2}}\right).
\end{align}
At last, we observe from \eqref{p6} and \eqref{q6} that
\begin{align}\label{eq:asypq6}
 \int_0^{\beta}p_6(s) \frac{\partial}{\partial \beta'}q_6(s)\ud \beta' &=-2\ii \beta^2 r_1 M^{(1)}_{14}\left(M^{(1)}_{12} +M^{(1)}_{13} M^{(1)}_{14}+M^{(1)}_{14}M^{(1)}_{24}+M^{(1)}_{34}\right)\nonumber\\
 &\quad-2\ii \beta^2 r_2\widetilde M^{(1)}_{14}\left(M^{(1)}_{12} +M^{(1)}_{13} M^{(1)}_{14}+M^{(1)}_{14}M^{(1)}_{24}+M^{(1)}_{34}\right) + \Boh\left(\frac{\ln s}{s^{1/2}}\right).
\end{align}
A combination of \eqref{M14}, \eqref{eq:asypqk}, \eqref{eq:asypq5}, \eqref{eq:asypq6} and Theorem \ref{th:pq} gives us that
\begin{multline}\label{eq:asy:H4}
\int_0^s H(t) \ud t =\int_0^{\beta}2 \ii \beta' \frac{\partial}{\partial \beta'}\left(\vartheta(s) + \widetilde \vartheta(s)\right)\ud \beta' + \frac{4}{3} \ii \beta (r_1+r_2)s^{\frac{3}{2}}
\\
-4 \ii \beta (s_1+s_2)s^{\frac{1}{2}}-2\beta^2+\Boh\left(\frac{\ln s}{s^{1/2}}\right).
\end{multline}
Recall the definition of $\vartheta(s)$ in \eqref{def:theta}, we have
\begin{align}\label{eq:intetheta}
\int_0^{\beta}2 \ii \beta' \frac{\partial}{\partial \beta'}\vartheta(s) \ud \beta'&=\int_0^{\beta}2 \ii \beta' \left(\frac{\partial}{\partial \beta'}\arg \Gamma(1+\beta) + \ii \frac{3}{2} \ln s + \ii \ln\left(8r_1-\frac{s_1}{s}\right)\right) \ud \beta'\nonumber\\
&=\ln G(1+\beta)G(1-\beta)-\frac{3}{2}\beta^2 \ln s+\beta^2 -\beta^2 \ln \left(8r_1\right)+\Boh(s^{-1}).
\end{align}
Inserting \eqref{eq:intetheta} into \eqref{eq:asy:H4}, we obtain the large gap asymptotic formula \eqref{ds-F} for $\gamma \in [0, 1)$ with the error term $\Boh\left(\frac{\ln s}{s^{1/2}}\right)$. In fact, integrating on the both sides of \eqref{asy:H}, one can find the error term in \eqref{ds-F} is $\Boh(s^{-\frac 12})$ instead of $\Boh\left(\frac{\ln s}{s^{1/2}}\right)$.

This completes the proof of Theorem \ref{th:1}.
\qed

\subsection{Proof of Corollary \ref{coro1}}
It is readily seen that, as $\nu \to 0$,
\begin{align}
\mathbb{E} \left(e^{-2\pi \nu N(s)}\right) = 1-2 \pi \mathbb{E}(N(s)) \nu + 2 \pi^2 \mathbb{E}(N(s)^2) \nu^2 + \Boh(\nu^3).
\end{align}
Then we have
\begin{align}\label{eq:ENs}
F(s; 1-e^{-2\pi \nu}) &=\ln  \mathbb{E} \left(e^{-2\pi \nu N(s)}\right)\nonumber\\
&=-2 \pi \mathbb{E}(N(s)) \nu + 2 \pi^2 \Var(N(s)) \nu^2 + \Boh(\nu^3), \qquad \nu \to 0,
\end{align}
where $F$ ie defined in \eqref{def:Fnotation}. In view of \eqref{ds-F}, we have
\begin{align}\label{eq:FENs}
F(s; 1-e^{-2\pi \nu}) &= -2 \pi \mu (s) \nu + 2 \pi^2\left(\sigma (s)^2 + \frac{\ln (64r_1 r_2)}{2 \pi^2}\right)\nu^2 \nonumber\\
&\quad+2 \ln (G(1+\ii \nu)G(1-\ii \nu)) + \Boh(s^{-\frac 12}), \quad s\to +\infty,
\end{align}
where the functions $\mu (s)$ and $\sigma (s)^2$ are defined in \eqref{def:mu-sigma}.

Note that
\begin{align}
G(1+z) = 1 + \frac{\ln (2\pi)-1}{2} z+\left(\frac{(\ln (2\pi)-1)^2}{8}-\frac{1 + \gamma_E}{2}\right) z^2 + \Boh(z^3), \qquad z \to 0,
\end{align}
where $\gamma_E$ is Euler's constant, we then obtain \eqref{def:EN} and \eqref{def:VarN} from \eqref{eq:ENs} and \eqref{eq:FENs}. Note that the additional factors $\ln s$ and $(\ln s)^2$ appear in the error terms due to the derivative with respect to $\nu$.

Finally, since $\sigma (s)^2 \to +\infty$ for large positive $s$, it is easily shown that
\begin{align}
\mathbb{E}\left(e^{t \cdot \frac{N(s) - \mu (s)}{\sqrt{\sigma (s)^2}}}\right) \to e^{\frac{t^2}{2}}, \qquad s \to +\infty,
\end{align}
which implies the convergence of $\frac{N(s) - \mu (s)}{\sqrt{\sigma (s)^2}}$ in distribution to the normal law $\mathcal{N} (0,1)$. The upper bound \eqref{ub} follows directly from a combination of \cite[Lemma 2.1 and Theorem 1.2]{CC21}, \eqref{eq:ENs} and \eqref{eq:FENs}.

This completes the proof of Corollary \ref{coro1}.
\qed

\begin{appendices}

\section{The Bessel parametrix}\label{sec:Bessel}
Let I, II, III, be the three regions shown in Figure \ref{fig:jumps:Phi-B}, the Bessel parametrix $\Phi^{(\Bes)}$ is defined as follows:
\begin{equation}\label{def:Bespara}
\Phi^{(\Bes)}(z) = \begin{cases}
\begin{pmatrix}
I_0(z^{\frac12}) & \frac{\ii}{\pi} K_0(z^{\frac12})\\
\pi \ii z^{\frac12}I_0'(z^{\frac12}) & -z^{1/2}K_0'(z^{\frac12})
\end{pmatrix}, &\quad z \in \textrm{I},\\
\begin{pmatrix}
I_0(z^{\frac12}) & \frac{\ii}{\pi} K_0(z^{\frac12})\\
\pi \ii z^{1/2}I_0'(z^{\frac12}) & -z^{1/2}K_0'(z^{\frac12})
\end{pmatrix}\begin{pmatrix}
1 & 0\\
-1 & 1
\end{pmatrix},& \quad z \in \textrm{II},\\
\begin{pmatrix}
I_0(z^{\frac12}) & \frac{\ii}{\pi} K_0(z^{\frac12})\\
\pi \ii z^{\frac12}I_0'(z^{\frac12}) & -z^{1/2}K_0'(z^{\frac12})
\end{pmatrix}\begin{pmatrix}
1 & 0\\
1 & 1
\end{pmatrix}, &\quad z \in \textrm{III},
\end{cases}
\end{equation}
where $I_0(z)$ and $K_0(z)$ denote the modified Bessel function of order $0$ (cf. \cite[Chapter 10]{DLMF}) and the principle branch is taken for $z^{1/2}$. By \cite{KMVV04}, $\Phi^{(\Bes)}$ is a solution of the following RH problem.
\subsection*{RH problem for $\Phi^{(\Bes)}$}
\begin{description}
\item(a) $\Phi^{(\Bes)}(z)$ is analytic in $\mathbb{C} \setminus \{\cup_{j=1}^3\Sigma_j\cup\{0\}\}$; where the contours $\Sigma_j$, $j=1,\ldots,3$, are indicated in Figure \ref{fig:jumps:Phi-B}.
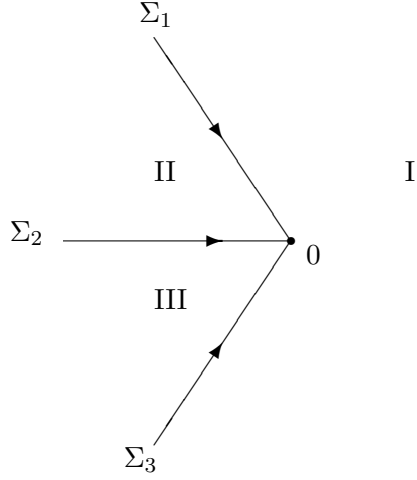
\begin{figure}[h]
\begin{center}
   \setlength{\unitlength}{1truemm}
   \begin{picture}(100,70)(-5,2)
       \put(40,40){\line(-2,-3){18}}
       \put(40,40){\line(-2,3){18}}
       \put(40,40){\line(-1,0){30}}

       \put(30,55){\thicklines\vector(2,-3){1}}
       \put(30,40){\thicklines\vector(1,0){1}}
       \put(30,25){\thicklines\vector(2,3){1}}


       \put(42,36.9){$0$}
             \put(20,69){$\Sigma_1$}
              \put(3,40){$\Sigma_2$}
       \put(18,10){$\Sigma_3$}

         \put(55,48){I}
            \put(22,48){II}
        \put(22,31){III}

       \put(40,40){\thicklines\circle*{1}}
\end{picture}
  \caption{The jump contour for the RH problem for $\Phi^{(\Bes)}(z)$.}
  \label{fig:jumps:Phi-B}
\end{center}
\end{figure}
\item (b) For $z \in \Sigma_j, j=1, 2, 3$,  $\Phi^{(\Bes)}(z)$ satisfies the jump condition
\begin{equation}\label{eq:jump:Bessel}
 \Phi_{+}^{(\Bes)}(z) =  \Phi_{ -}^{(\Bes)}(z) \begin{cases}
 \begin{pmatrix}
 1 & 0\\
1 & 1
 \end{pmatrix}, &\quad z \in \Sigma_1,\\
  \begin{pmatrix}
 0 & 1\\
 -1 & 0
 \end{pmatrix}, &\quad z \in \Sigma_2,\\
  \begin{pmatrix}
 1 & 0\\
 1 & 1
 \end{pmatrix}, &\quad z \in \Sigma_3.
 \end{cases}
\end{equation}
\item (c) As $z \to \infty$, $\Phi^{(\Bes)}(z)$ satisfies the following asymptotic behavior:
\begin{equation}\label{eq:infty:Bessel}
\Phi^{(\Bes)}(z) = \frac{(\pi^2 z)^{-\sigma_3/4}}{\sqrt{2}} \begin{pmatrix} 1 & \ii\\ \ii & 1 \end{pmatrix} \left(I + \frac{1}{8 z^
{1/2}} \begin{pmatrix} -1 & -2\ii \\ -2\ii & 1  \end{pmatrix} + \Boh(z^{-1})\right) e^{z^{\frac12} \sigma_3}.
\end{equation}
\item (d) As $z \to 0$, we have $\Phi^{(\Bes)}(z)=\Boh(\ln{|z|})$.
\end{description}

\section{The confluent hypergeometric parametrix}\label{sec:CHF}
The confluent hypergeometric parametrix $\Phi^{(\CHF)}(z)=\Phi^{(\CHF)}(z;\beta)$ with $\beta$ being a parameter is a solution of the following RH problem.

\subsection*{RH problem for $\Phi^{(\CHF)}$}
\begin{description}
  \item(a)   $\Phi^{(\CHF)}(z)$ is analytic in $\mathbb{C}\setminus \{\cup^6_{j=1}\widehat\Sigma_j\cup\{0\}\}$, where the contours $\widehat\Sigma_j$, $j=1,\ldots,6,$ are indicated in Figure \ref{fig:jumps-Phi-C}.

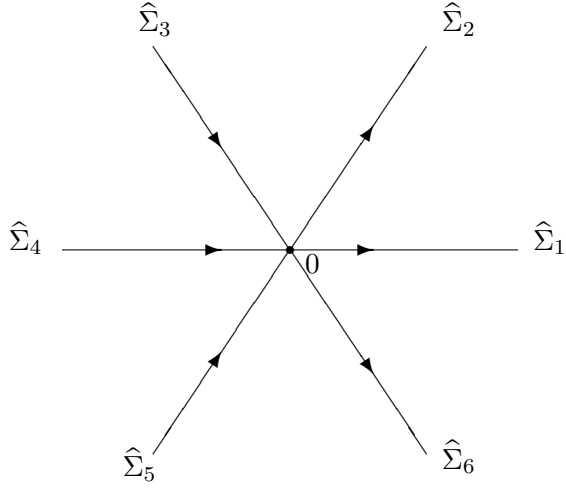
\begin{figure}[h]
\begin{center}
   \setlength{\unitlength}{1truemm}
   \begin{picture}(100,70)(-5,2)
       \put(40,40){\line(-2,-3){18}}
       \put(40,40){\line(-2,3){18}}
       \put(40,40){\line(-1,0){30}}
       \put(40,40){\line(1,0){30}}
      \put(40,40){\line(2,-3){18}}
      \put(40,40){\line(2,3){18}}

       \put(30,55){\thicklines\vector(2,-3){1}}
       \put(30,40){\thicklines\vector(1,0){1}}
       \put(50,40){\thicklines\vector(1,0){1}}
       \put(30,25){\thicklines\vector(2,3){1}}
      \put(50,25){\thicklines\vector(2,-3){1}}
       \put(50,55){\thicklines\vector(2,3){1}}


       \put(42,36.9){$0$}
         \put(72,40){$\widehat \Sigma_1$}
           \put(60,69){$\widehat \Sigma_2$}
             \put(20,69){$\widehat \Sigma_3$}
              \put(3,40){$\widehat \Sigma_4$}
       \put(18,10){$\widehat \Sigma_5$}
          \put(60,11){$ \widehat \Sigma_6$}

%

       \put(40,40){\thicklines\circle*{1}}
\end{picture}
   \caption{The jump contours for the RH problem for $\Phi^{(\CHF)}$.}
   \label{fig:jumps-Phi-C}
\end{center}
\end{figure}

  \item(b) $\Phi^{(\CHF)}$ satisfies the jump condition
  \begin{equation}\label{HJumps}
  \Phi^{(\CHF)}_+(z)=\Phi^{(\CHF)}_-(z) \widehat J_j(z), \quad z \in \widehat\Sigma_j,\quad j=1,\ldots,6,
  \end{equation}
  where
  \begin{equation*}
 \widehat J_1(z) = \begin{pmatrix}
    0 &   e^{-\beta \pi \ii} \\
    -  e^{\beta \pi \ii} &  0
    \end{pmatrix}, \qquad \widehat J_2(z) = \begin{pmatrix}
    1 & 0 \\
    e^{ \beta \pi \ii } & 1
    \end{pmatrix}, \qquad
    \widehat J_3(z) = \begin{pmatrix}
    1 & 0 \\
    e^{ -\beta\pi \ii} & 1
    \end{pmatrix},                                                         
  \end{equation*}
  \begin{equation*}
  \widehat J_4(z) = \begin{pmatrix}
    0 &   e^{\beta\pi \ii} \\
     -  e^{-\beta\pi \ii} &  0
     \end{pmatrix}, \qquad
      \widehat J_5(z) = \begin{pmatrix}
     1 & 0 \\
     e^{- \beta\pi \ii} & 1
     \end{pmatrix},\qquad
     \widehat J_6(z) = \begin{pmatrix}
   1 & 0 \\
   e^{\beta\pi \ii} & 1
   \end{pmatrix}.
  \end{equation*}

  \item(c) As $z\to \infty$, $\Phi^{(\CHF)}(z)$ satisfies the following asymptotic behavior:
  \begin{multline}\label{H at infinity}
 \Phi^{(\CHF)}(z)=(I + \Boh(z^{-1})) z^{-\beta \sigma_3}e^{-\frac{\ii z}{2}\sigma_3}
  \left\{\begin{array}{ll}
                         I, & ~0< \arg z <\pi,
                         \\
                        \begin{pmatrix}
                                                             0 &   -e^{\beta\pi \ii} \\
                                                            e^{-\beta\pi \ii } &  0
                        \end{pmatrix}, &~ \pi< \arg z<\frac{3\pi}{2},
                        \\
                        \begin{pmatrix}
                        0 &   -e^{-\beta\pi \ii} \\
                        e^{\beta\pi \ii} &  0
                        \end{pmatrix}, & -\frac{\pi}{2}<\arg z<0.
 \end{array}\right.
\end{multline}
\item(d) As $z\to 0$, we have $\Phi^{(\CHF)}(z)=\Boh(\ln |z|)$.
\end{description}

It follows from \cite{IK} that the above RH problem can be solved explicitly in terms of the confluent hypergeometric functions. Moreover, as $z \to 0$, we have

\begin{equation}\label{eq:H-expand-2}
\Phi^{(\CHF)}(z) e^{-\frac{\beta \pi \ii}{2} \sigma_3} = \Upsilon_0\left( I+ \Upsilon_1z+\Boh(z^2) \right) \begin{pmatrix} 1 & -\frac{\gamma}{2\pi \ii} \ln (e^{-\frac{\pi \ii}{2}}z) \\
0 & 1  \end{pmatrix},
\end{equation}
for $z$ belonging to the region bounded by the rays $\widehat \Sigma_2$ and $\widehat \Sigma_3$, where $\gamma=1-e^{2\beta \pi \ii}$,
\begin{align}\label{eq:H-expand-coeff-0}
\Upsilon_0
=\begin{pmatrix} \Gamma\left(1-\beta\right) e^{-\beta \pi \ii} &\frac{1}{\Gamma(\beta)} \left( \frac{\Gamma'\left(1-\beta\right)}{\Gamma\left(1-\beta\right)} +2\gamma_{\textrm{E}} \right) \vspace{5pt} \\
\Gamma\left(1+\beta\right) & -\frac{e^{\beta \pi \ii}}{\Gamma(-\beta)} \left( \frac{\Gamma'\left(-\beta\right)}{\Gamma\left(-\beta\right) } +2\gamma_{\textrm{E}}\right) \end{pmatrix}
\end{align}
with $\gamma_{\textrm{E}}$ being the Euler's constant,
and
 \begin{equation}\label{eq:H-expand-coeff-1}
(\Upsilon_1)_{21}=\frac{ \beta \pi \ii \, e^{-\beta \pi \ii} }{\sin(\beta \pi )}.
\end{equation}

\end{appendices}

\section*{Acknowledgements}
This work was partially supported by National Natural Science Foundation of China under grant numbers 12271105, 11822104, and ``Shuguang Program'' supported by Shanghai Education Development Foundation and Shanghai Municipal Education Commission.


\end{document}